\pdfoutput=1
\RequirePackage[l2tabu,orthodox]{nag}
\documentclass[english]{article}
\usepackage[letterpaper, margin=1in]{geometry}

\usepackage[T1]{fontenc}
\usepackage[utf8]{inputenc}
\usepackage{lmodern}
\usepackage{microtype}

\usepackage{float}
\RequirePackage[pdfborder={0 0 0}]{hyperref}

\usepackage{tkz-graph}
\tikzset{VertexStyle/.style = {minimum size = .5pt, inner sep = .5pt}}
\usetikzlibrary{calc,matrix}
\usepackage{ifthen}

\usepackage{amsmath}
\usepackage{amssymb}
\usepackage{bm}

\usepackage{enumitem}

\newcounter{nlistcounter}
\newenvironment{nlist}[1]{
  \renewcommand{\thenlistcounter}{\upshape(#1.\arabic{nlistcounter})}
  \begin{list}{\thenlistcounter}{%
      \usecounter{nlistcounter}
      \setlength{\labelwidth}{1.5em}%
      \setlength{\leftmargin}{\labelwidth}%
      \addtolength{\leftmargin}{\labelsep}%
      \setlength{\listparindent}{0em}%
      \setlength{\topsep}{5pt}%
      \setlength{\itemsep}{5pt}%
      \setlength{\parsep}{0pt}%
    }
  }{
  \end{list}
}

\usepackage[amsmath,amsthm,hyperref,thmmarks]{ntheorem}
\theoremseparator{.}
\newtheorem{thm}{Theorem}[section]
\newtheorem{lem}[thm]{Lemma}
\newtheorem{prop}[thm]{Proposition}
\newtheorem{cor}[thm]{Corollary}
\newtheorem{claim}[thm]{Claim}
\newtheorem{fact}[thm]{Fact}

\theoremstyle{plain}
\theorembodyfont{\upshape}
\newtheorem{example}[thm]{Example}
\newtheorem{rem}[thm]{Remark}
\newtheorem{defn}[thm]{Definition}

\theoremstyle{empty}
\theorembodyfont{\itshape}
\newtheorem{rep}{}
\newcommand{\newreptheorem}[2]{%
	\newenvironment{rep#1}[1]{%
  \begin{rep}[#2 \ref{##1}, restated.]}%
		{\end{rep}}}

\newreptheorem{theorem}{Theorem}
\newreptheorem{lemma}{Lemma}
\newreptheorem{proposition}{Proposition}
\newreptheorem{corollary}{Corollary}

\newcommand{\cone}{\gamma}
\newcommand{\comp}{\alpha}
\newcommand{\bag}{\beta}
\newcommand{\sep}{\sigma}

\global\long\def\partitions#1{\mathrm{Part}(#1)}
\global\long\def\contract#1#2{#1 / #2}
\global\long\def\spasm#1{\mathrm{Spasm}\paren*{#1}}

\global\long\def\hom{\operatorname{Hom}}
\global\long\def\surj{\operatorname{Surj}}
\global\long\def\aut{\operatorname{Aut}}
\global\long\def\iso{\operatorname{Iso}}
\global\long\def\sub{\operatorname{Sub}}
\global\long\def\ext{\operatorname{Ext}}
\global\long\def\indsub{\operatorname{IndSub}}
\global\long\def\emb{\operatorname{Emb}}
\global\long\def\stremb{\operatorname{StrEmb}}
\global\long\def\part{\operatorname{Part}}

\global\long\def\Aut#1{\mathrm{Aut}(#1)}
\global\long\def\Emb#1#2{\mathrm{Emb}(#1\to#2)}
\global\long\def\StrEmb#1#2{\mathrm{StrEmb}(#1\to#2)}
\global\long\def\Hom#1#2{\mathrm{Hom}(#1\to#2)}
\global\long\def\Sub#1#2{\mathrm{Sub}(#1\to#2)}
\global\long\def\IndSub#1#2{\mathrm{IndSub}(#1\to#2)}
\global\long\def\PartitionedSub#1#2{\mathrm{PartitionedSub}(#1\to#2)}

\global\long\def\IndProb#1{\mathsf{\#Ind}(#1)}
\global\long\def\IndPropProb#1{\mathsf{\#IndProp}(#1)}
\global\long\def\HomProb#1{\mathsf{\#Hom}(#1)}
\global\long\def\SubProb#1{\mathsf{\#Sub}(#1)}
\global\long\def\PartitionedSubProb#1{\mathsf{\#PartitionedSub}(#1)}

\global\long\def\FPT{\mathsf{FPT}}
\global\long\def\sharpWone{\mathsf{\#W[1]}}
\global\long\def\sharpETH{\mathsf{\#ETH}}
\global\long\def\Wone{\mathsf{W[1]}}
\global\long\def\ETH{\mathsf{ETH}}
\global\long\def\Graphsl{{\mathcal{G}}}
\global\long\def\Graphsu{{\mathcal{G}^*}}

\global\long\def\poly#1{\mathrm{poly}\paren*{#1}}
\global\long\def\supp#1{\mathrm{supp}\paren*{#1}}
\global\long\def\tw#1{\mathrm{tw}(#1)}

\global\long\def\vc#1{\mathrm{vc}(#1)}
\global\long\def\sparseTWconstant{\xi}

\newcommand{\cc}[1]{\ensuremath{\mathsf{#1}}}
\newcommand{\NP}{\cc{NP}}
\newcommand{\PolyTime}{\cc{P}}
\newcommand{\sharpP}{\cc{\#P}}

\newcommand{\executeiffilenewer}[3]{%
\ifnum\pdfstrcmp{\pdffilemoddate{#1}}%
{\pdffilemoddate{#2}}>0%
{\immediate\write18{#3}}\fi%
} 

\usepackage{mathtools}
\DeclarePairedDelimiter\paren{\lparen}{\rparen}
\DeclarePairedDelimiter\abs{\lvert}{\rvert}
\DeclarePairedDelimiter\set{\{}{\}}
\DeclarePairedDelimiterX\setc[2]{\{}{\}}{\,#1 \;\colon\; #2\,}
\DeclarePairedDelimiterX\parenc[2]{\lparen}{\rparen}{\,#1 \;\delimsize\vert\; #2\,}

\newcommand{\dotcup}{\mathbin{\dot\cup}}

\newcommand{\cH}{\mathcal{H}}

\usepackage{authblk}

\author[a]{Radu Curticapean}
\author[b]{Holger Dell}
\author[a]{D\'aniel Marx}

\affil[a]{%
  Institute for Computer Science and Control, Hungarian Academy of
  Sciences (MTA SZTAKI), Budapest, Hungary%
}
\affil[b]{%
  Saarland University and Cluster of Excellence (MMCI), Saarbr\"{u}cken, Germany%
}

\input{figures/tw-maximal.tikz}
\newcommand{\drawgraph}[1]{%
\ifthenelse{\equal{#1}{DDW}}{%
\begin{tikzpicture}[baseline=1mm]
\useasboundingbox (0,0) rectangle (0.8cm,0.4cm);
\definecolor{cv0}{rgb}{0.0,0.0,0.0}
\definecolor{cfv0}{rgb}{0.0,0.0,0.0}
\definecolor{cv1}{rgb}{0.0,0.0,0.0}
\definecolor{cfv1}{rgb}{0.0,0.0,0.0}
\definecolor{cv2}{rgb}{0.0,0.0,0.0}
\definecolor{cfv2}{rgb}{0.0,0.0,0.0}
\definecolor{cv3}{rgb}{0.0,0.0,0.0}
\definecolor{cfv3}{rgb}{0.0,0.0,0.0}
\definecolor{cv4}{rgb}{0.0,0.0,0.0}
\definecolor{cfv4}{rgb}{0.0,0.0,0.0}
\definecolor{cv0v3}{rgb}{0.0,0.0,0.0}
\definecolor{cv1v4}{rgb}{0.0,0.0,0.0}
\definecolor{cv2v3}{rgb}{0.0,0.0,0.0}
\definecolor{cv2v4}{rgb}{0.0,0.0,0.0}
\Vertex[style={minimum size=0.0cm,draw=cv0,fill=cfv0,shape=circle},NoLabel,x=0.1cm,y=0.35cm]{v0}
\Vertex[style={minimum size=0.0cm,draw=cv1,fill=cfv1,shape=circle},NoLabel,x=0.75cm,y=0.05cm]{v1}
\Vertex[style={minimum size=0.0cm,draw=cv2,fill=cfv2,shape=circle},NoLabel,x=0.4343cm,y=0.2055cm]{v2}
\Vertex[style={minimum size=0.0cm,draw=cv3,fill=cfv3,shape=circle},NoLabel,x=0.2595cm,y=0.2832cm]{v3}
\Vertex[style={minimum size=0.0cm,draw=cv4,fill=cfv4,shape=circle},NoLabel,x=0.6065cm,y=0.1225cm]{v4}
\Edge[lw=0.01cm,style={color=cv0v3,},](v0)(v3)
\Edge[lw=0.01cm,style={color=cv1v4,},](v1)(v4)
\Edge[lw=0.01cm,style={color=cv2v3,},](v2)(v3)
\Edge[lw=0.01cm,style={color=cv2v4,},](v2)(v4)
\end{tikzpicture}
}{}%
\ifthenelse{\equal{#1}{Cr}}{%
\begin{tikzpicture}[baseline=1mm]
\useasboundingbox (0,0) rectangle (0.8cm,0.4cm);
\definecolor{cv0}{rgb}{0.0,0.0,0.0}
\definecolor{cfv0}{rgb}{0.0,0.0,0.0}
\definecolor{cv1}{rgb}{0.0,0.0,0.0}
\definecolor{cfv1}{rgb}{0.0,0.0,0.0}
\definecolor{cv2}{rgb}{0.0,0.0,0.0}
\definecolor{cfv2}{rgb}{0.0,0.0,0.0}
\definecolor{cv3}{rgb}{0.0,0.0,0.0}
\definecolor{cfv3}{rgb}{0.0,0.0,0.0}
\definecolor{cv0v1}{rgb}{0.0,0.0,0.0}
\definecolor{cv0v2}{rgb}{0.0,0.0,0.0}
\definecolor{cv1v3}{rgb}{0.0,0.0,0.0}
\definecolor{cv2v3}{rgb}{0.0,0.0,0.0}
\Vertex[style={minimum size=0.0cm,draw=cv0,fill=cfv0,shape=circle},NoLabel,x=0.7451cm,y=0.05cm]{v0}
\Vertex[style={minimum size=0.0cm,draw=cv1,fill=cfv1,shape=circle},NoLabel,x=0.1cm,y=0.0522cm]{v1}
\Vertex[style={minimum size=0.0cm,draw=cv2,fill=cfv2,shape=circle},NoLabel,x=0.75cm,y=0.3485cm]{v2}
\Vertex[style={minimum size=0.0cm,draw=cv3,fill=cfv3,shape=circle},NoLabel,x=0.1102cm,y=0.35cm]{v3}
\Edge[lw=0.01cm,style={color=cv0v1,},](v0)(v1)
\Edge[lw=0.01cm,style={color=cv0v2,},](v0)(v2)
\Edge[lw=0.01cm,style={color=cv1v3,},](v1)(v3)
\Edge[lw=0.01cm,style={color=cv2v3,},](v2)(v3)
\end{tikzpicture}
}{}%
\ifthenelse{\equal{#1}{CN}}{%
\begin{tikzpicture}[rotate=180,baseline=-3mm]
\useasboundingbox (0,0) rectangle (0.8cm,0.4cm);
\definecolor{cv0}{rgb}{0.0,0.0,0.0}
\definecolor{cfv0}{rgb}{0.0,0.0,0.0}
\definecolor{cv1}{rgb}{0.0,0.0,0.0}
\definecolor{cfv1}{rgb}{0.0,0.0,0.0}
\definecolor{cv2}{rgb}{0.0,0.0,0.0}
\definecolor{cfv2}{rgb}{0.0,0.0,0.0}
\definecolor{cv3}{rgb}{0.0,0.0,0.0}
\definecolor{cfv3}{rgb}{0.0,0.0,0.0}
\definecolor{cv0v3}{rgb}{0.0,0.0,0.0}
\definecolor{cv1v2}{rgb}{0.0,0.0,0.0}
\definecolor{cv1v3}{rgb}{0.0,0.0,0.0}
\definecolor{cv2v3}{rgb}{0.0,0.0,0.0}
\Vertex[style={minimum size=0.0cm,draw=cv0,fill=cfv0,shape=circle},NoLabel,x=0.1cm,y=0.2659cm]{v0}
\Vertex[style={minimum size=0.0cm,draw=cv1,fill=cfv1,shape=circle},NoLabel,x=0.7258cm,y=0.05cm]{v1}
\Vertex[style={minimum size=0.0cm,draw=cv2,fill=cfv2,shape=circle},NoLabel,x=0.75cm,y=0.35cm]{v2}
\Vertex[style={minimum size=0.0cm,draw=cv3,fill=cfv3,shape=circle},NoLabel,x=0.4512cm,y=0.2296cm]{v3}
\Edge[lw=0.01cm,style={color=cv0v3,},](v0)(v3)
\Edge[lw=0.01cm,style={color=cv1v2,},](v1)(v2)
\Edge[lw=0.01cm,style={color=cv1v3,},](v1)(v3)
\Edge[lw=0.01cm,style={color=cv2v3,},](v2)(v3)
\end{tikzpicture}
}{}%
\ifthenelse{\equal{#1}{CF}}{%
\begin{tikzpicture}[baseline=1mm]
\useasboundingbox (0,0) rectangle (0.8cm,0.4cm);
\definecolor{cv0}{rgb}{0.0,0.0,0.0}
\definecolor{cfv0}{rgb}{0.0,0.0,0.0}
\definecolor{cv1}{rgb}{0.0,0.0,0.0}
\definecolor{cfv1}{rgb}{0.0,0.0,0.0}
\definecolor{cv2}{rgb}{0.0,0.0,0.0}
\definecolor{cfv2}{rgb}{0.0,0.0,0.0}
\definecolor{cv3}{rgb}{0.0,0.0,0.0}
\definecolor{cfv3}{rgb}{0.0,0.0,0.0}
\definecolor{cv0v3}{rgb}{0.0,0.0,0.0}
\definecolor{cv1v3}{rgb}{0.0,0.0,0.0}
\definecolor{cv2v3}{rgb}{0.0,0.0,0.0}
\Vertex[style={minimum size=0.0cm,draw=cv0,fill=cfv0,shape=circle},NoLabel,x=0.75cm,y=0.0871cm]{v0}
\Vertex[style={minimum size=0.0cm,draw=cv1,fill=cfv1,shape=circle},NoLabel,x=0.3636cm,y=0.35cm]{v1}
\Vertex[style={minimum size=0.0cm,draw=cv2,fill=cfv2,shape=circle},NoLabel,x=0.1cm,y=0.05cm]{v2}
\Vertex[style={minimum size=0.0cm,draw=cv3,fill=cfv3,shape=circle},NoLabel,x=0.4057cm,y=0.1584cm]{v3}
\Edge[lw=0.01cm,style={color=cv0v3,},](v0)(v3)
\Edge[lw=0.01cm,style={color=cv1v3,},](v1)(v3)
\Edge[lw=0.01cm,style={color=cv2v3,},](v2)(v3)
\end{tikzpicture}
}{}%
\ifthenelse{\equal{#1}{CR}}{%
\begin{tikzpicture}[baseline=1mm]
\useasboundingbox (0,0) rectangle (0.8cm,0.4cm);
\definecolor{cv0}{rgb}{0.0,0.0,0.0}
\definecolor{cfv0}{rgb}{0.0,0.0,0.0}
\definecolor{cv1}{rgb}{0.0,0.0,0.0}
\definecolor{cfv1}{rgb}{0.0,0.0,0.0}
\definecolor{cv2}{rgb}{0.0,0.0,0.0}
\definecolor{cfv2}{rgb}{0.0,0.0,0.0}
\definecolor{cv3}{rgb}{0.0,0.0,0.0}
\definecolor{cfv3}{rgb}{0.0,0.0,0.0}
\definecolor{cv0v2}{rgb}{0.0,0.0,0.0}
\definecolor{cv1v3}{rgb}{0.0,0.0,0.0}
\definecolor{cv2v3}{rgb}{0.0,0.0,0.0}
\Vertex[style={minimum size=0.0cm,draw=cv0,fill=cfv0,shape=circle},NoLabel,x=0.75cm,y=0.05cm]{v0}
\Vertex[style={minimum size=0.0cm,draw=cv1,fill=cfv1,shape=circle},NoLabel,x=0.1cm,y=0.35cm]{v1}
\Vertex[style={minimum size=0.0cm,draw=cv2,fill=cfv2,shape=circle},NoLabel,x=0.5482cm,y=0.1559cm]{v2}
\Vertex[style={minimum size=0.0cm,draw=cv3,fill=cfv3,shape=circle},NoLabel,x=0.3237cm,y=0.2638cm]{v3}
\Edge[lw=0.01cm,style={color=cv0v2,},](v0)(v2)
\Edge[lw=0.01cm,style={color=cv1v3,},](v1)(v3)
\Edge[lw=0.01cm,style={color=cv2v3,},](v2)(v3)
\end{tikzpicture}
}{}%
\ifthenelse{\equal{#1}{Bw}}{%
\begin{tikzpicture}[baseline=1mm]
\useasboundingbox (0,0) rectangle (0.8cm,0.4cm);
\definecolor{cv0}{rgb}{0.0,0.0,0.0}
\definecolor{cfv0}{rgb}{0.0,0.0,0.0}
\definecolor{cv1}{rgb}{0.0,0.0,0.0}
\definecolor{cfv1}{rgb}{0.0,0.0,0.0}
\definecolor{cv2}{rgb}{0.0,0.0,0.0}
\definecolor{cfv2}{rgb}{0.0,0.0,0.0}
\definecolor{cv0v1}{rgb}{0.0,0.0,0.0}
\definecolor{cv0v2}{rgb}{0.0,0.0,0.0}
\definecolor{cv1v2}{rgb}{0.0,0.0,0.0}
\Vertex[style={minimum size=0.0cm,draw=cv0,fill=cfv0,shape=circle},NoLabel,x=0.1cm,y=0.05cm]{v0}
\Vertex[style={minimum size=0.0cm,draw=cv1,fill=cfv1,shape=circle},NoLabel,x=0.3378cm,y=0.35cm]{v1}
\Vertex[style={minimum size=0.0cm,draw=cv2,fill=cfv2,shape=circle},NoLabel,x=0.75cm,y=0.0993cm]{v2}
\Edge[lw=0.01cm,style={color=cv0v1,},](v0)(v1)
\Edge[lw=0.01cm,style={color=cv0v2,},](v0)(v2)
\Edge[lw=0.01cm,style={color=cv1v2,},](v1)(v2)
\end{tikzpicture}
}{}%
\ifthenelse{\equal{#1}{BW}}{%
\begin{tikzpicture}[baseline=1mm]
\useasboundingbox (0,0) rectangle (0.8cm,0.4cm);
\definecolor{cv0}{rgb}{0.0,0.0,0.0}
\definecolor{cfv0}{rgb}{0.0,0.0,0.0}
\definecolor{cv1}{rgb}{0.0,0.0,0.0}
\definecolor{cfv1}{rgb}{0.0,0.0,0.0}
\definecolor{cv2}{rgb}{0.0,0.0,0.0}
\definecolor{cfv2}{rgb}{0.0,0.0,0.0}
\definecolor{cv0v2}{rgb}{0.0,0.0,0.0}
\definecolor{cv1v2}{rgb}{0.0,0.0,0.0}
\Vertex[style={minimum size=0.0cm,draw=cv0,fill=cfv0,shape=circle},NoLabel,x=0.1cm,y=0.35cm]{v0}
\Vertex[style={minimum size=0.0cm,draw=cv1,fill=cfv1,shape=circle},NoLabel,x=0.75cm,y=0.05cm]{v1}
\Vertex[style={minimum size=0.0cm,draw=cv2,fill=cfv2,shape=circle},NoLabel,x=0.4259cm,y=0.2004cm]{v2}
\Edge[lw=0.01cm,style={color=cv0v2,},](v0)(v2)
\Edge[lw=0.01cm,style={color=cv1v2,},](v1)(v2)
\end{tikzpicture}
}{}%
\ifthenelse{\equal{#1}{A_}}{%
\begin{tikzpicture}[baseline=1mm]
\useasboundingbox (0,0) rectangle (0.8cm,0.4cm);
\definecolor{cv0}{rgb}{0.0,0.0,0.0}
\definecolor{cfv0}{rgb}{0.0,0.0,0.0}
\definecolor{cv1}{rgb}{0.0,0.0,0.0}
\definecolor{cfv1}{rgb}{0.0,0.0,0.0}
\definecolor{cv0v1}{rgb}{0.0,0.0,0.0}
\Vertex[style={minimum size=0.0cm,draw=cv0,fill=cfv0,shape=circle},NoLabel,x=0.1cm,y=0.35cm]{v0}
\Vertex[style={minimum size=0.0cm,draw=cv1,fill=cfv1,shape=circle},NoLabel,x=0.75cm,y=0.05cm]{v1}
\Edge[lw=0.01cm,style={color=cv0v1,},](v0)(v1)
\end{tikzpicture}
}{}%
\ifthenelse{\equal{#1}{@}}{%
\begin{tikzpicture}[baseline=1mm]
\useasboundingbox (0,0) rectangle (0.8cm,0.4cm);
\definecolor{cv0}{rgb}{0.0,0.0,0.0}
\definecolor{cfv0}{rgb}{0.0,0.0,0.0}
\Vertex[style={minimum size=0.0cm,draw=cv0,fill=cfv0,shape=circle},NoLabel,x=0.425cm,y=0.2cm]{v0}
\end{tikzpicture}
}{}%
}

\begin{document}

\title{Homomorphisms Are a Good Basis for \\Counting~Small~Subgraphs\footnote{An extended abstract of this paper appears at STOC 2017.
    Part of this work was done while the authors were visiting the
    Simons Institute for the Theory of Computing and the Dagstuhl Seminar 17041
-- ``Randomization in Parameterized Complexity''.}}

\maketitle

\begin{abstract}
We introduce \emph{graph motif parameters}, a class of graph parameters that depend only on the frequencies of constant-size induced subgraphs.
Classical works by Lovász show that many interesting quantities have this form,
including, for fixed graphs~$H$,
the number of $H$-copies (induced or not) in an input graph~$G$,
and the number of homomorphisms from~$H$ to $G$.

Using the framework of graph motif parameters,
we obtain faster algorithms for counting subgraph copies of fixed graphs~$H$ in host graphs~$G$:
For graphs $H$ on $k$ edges, 
we show how to count subgraph copies of $H$ 
in time $k^{O(k)}\cdot n^{0.174k + o(k)}$
by a surprisingly simple algorithm.
This improves upon previously known running times,
such as $O(n^{0.91k + c})$ time for $k$-edge matchings
or $O(n^{0.46k + c})$ time for $k$-cycles.

Furthermore, we prove a general complexity dichotomy for evaluating graph motif parameters:
Given a class $\mathcal C$ of such parameters,
we consider the problem of evaluating $f\in \mathcal C$ on input graphs $G$,
parameterized by the number of induced subgraphs that $f$ depends upon.
For every recursively enumerable class $\mathcal C$, we prove the above problem to be either $\FPT$ or $\sharpWone$-hard, with an explicit dichotomy criterion.
This allows us to recover known dichotomies for counting subgraphs, induced subgraphs, and homomorphisms in a uniform and simplified way,
together with improved lower bounds.

Finally, we extend graph motif parameters to colored subgraphs and prove a complexity trichotomy:
For vertex-colored graphs $H$ and $G$, 
where $H$ is from a fixed class $\mathcal H$, 
we want to count color-preserving $H$-copies in $G$.
We show that this problem is either polynomial-time solvable or $\FPT$ or $\sharpWone$-hard,
and that the FPT cases indeed need FPT time under reasonable assumptions.

\end{abstract}

\section{Introduction}

Deciding the existence of subgraph patterns $H$ in
input graphs $G$ constitutes the classical \emph{subgraph isomorphism
problem} \cite{cook1971complexity,ullmann1976algorithm}, which generalizes
$\NP$-complete problems like the Hamiltonian cycle problem or the clique
problem.
In some applications however, it is not sufficient to merely know whether $H$ occurs in $G$, 
but instead one wishes to determine the \emph{number} of such occurrences.
This is clearly at least as hard as deciding their existence, but it can be much harder:
The existence of perfect matchings can be tested in polynomial time,
but the counting version is $\sharpP$-hard \cite{Valiant1979a}.

Subgraph counting problems have applications in areas like statistical physics, probabilistic inference, and network analysis.
In particular, in network analysis, such problems arise in the context of discovering \emph{network motifs}.
These are small patterns that occur more often in a network than would be expected if the network was random.
Through network motifs, the problem of counting subgraphs has found applications in the study of gene transcription networks,
neural networks, and social networks~\cite{Milo824},
and there is a large body of work dedicated to the algorithmic discovery of network motifs~\cite{grochow2007network,alon2008biomolecular,omidi2009moda,kashani2009kavosh,schreiber2005frequency,chen2006nemofinder,kashtan2004efficient,wernicke2006efficient,DBLP:conf/alcob/SchillerJHS15}.

Inspired by these applications, we study the algorithmic problem of counting occurrences of \emph{small} patterns~$H$ in large host graphs $G$.
The abstract notion of a ``pattern occurrence'' may be formalized in various ways,
which may result in vastly different problems:
To state only some examples,
we may be interested in counting subgraph copies of a graph $H$, or induced subgraph copies of $H$, or homomorphisms from $H$ to $G$, and we can also consider settings where both pattern $H$ and host graph $G$ are colored
and we wish to count subgraphs of $G$ that are color-preserving isomorphic to $H$.

It may seem daunting at first to try to deal with all different types of pattern occurrences.
Fortunately, Lovász~\cite{lovasz1967operations,lovaszbook} defined a framework
that allows us to express virtually all kinds of pattern types in a unified way.
As it turns out, graph parameters such as the number of subgraph copies of $H$ (induced or not) in a host graph~$G$, 
or the number of graph homomorphisms from~$H$ to~$G$ are actually just ``linear combinations'' of each other in a well-defined sense.
We build on this and define a general framework of so-called \emph{graph motif parameters} to capture counting linear combinations of small patterns,
into which (induced) subgraph or homomorphism numbers embed naturally as special cases.

In the remainder of the introduction,
we first discuss algorithmic and complexity-theoretic aspects of counting (induced) subgraphs and homomorphisms 
in \S\ref{sec: intro subgraphs}--\S\ref{sec: intro induced subgraphs}
and state the results we derive for these special cases.
In~\S\ref{sec: intro lovasz}, we then give an introduction into the general framework of graph motif parameters, our interpretation of Lovász's unified framework,
which also provides the main techniques for our proofs.
Finally, in~\S\ref{sec: intro colored subgraphs} we give an exposition of our results for vertex-colored subgraphs.

\subsection{Counting small subgraphs}
\label{sec: intro subgraphs}

For any fixed $k$-vertex pattern graph~$H$, 
we can count all subgraph copies of $H$ in an $n$-vertex host graph~$G$ using brute-force for a running time of $O(n^k)$.
While this running time is polynomial for any fixed $H$, 
it quickly becomes infeasible as $k$ grows.
Fortunately enough, non-trivial improvements on the exponent are known,
albeit only for specific classes of patterns:
\begin{itemize}

	\item We can count triangles in the same
    time~$O(n^{\omega})$ that it takes to multiply two
    $(n\times n)$-matrices~\cite{itai1978finding}.
    It is known that $\omega<2.373$
    holds~\cite{DBLP:conf/stoc/Williams12,DBLP:conf/issac/Gall14a}.
    This approach can be generalized from triangles to $k$-cliques with $k\in \mathbb N$~\cite{nevsetvril1985complexity},
    for a running time of $n^{\omega k /3 + O(1)}$.
    Fast matrix multiplication is also used to improve on exhaustive search for counting cycles of length at most seven~\cite{alon1997finding} and various other problems~\cite{kloks2000finding,eisenbrand2004complexity}.

    \item For $k$-edge paths or generally any pattern of bounded
    pathwidth, a ``meet in the middle'' approach yields $
    n^{k/2 + O(1)}$ time algorithms~\cite{koutis2009limits,bjorklund2009counting}.
    For a while, this approach appeared to be a barrier for faster algorithms,
    until Björklund et al.~\cite{fasterthanmeetinthemiddle} gave an algorithm
    for counting $k$-paths, matchings on~$k$ vertices, and other $k$-vertex
    patterns of bounded pathwidth in time $n^{0.455 k+ O(1)}$.
  
    \item If $\vc{H}$ is the \emph{vertex-cover number of~$H$}, that is, the size of its
    smallest vertex-cover, then we can count $H$-copies in time
    $n^{\vc{H}+O(1)}$ \cite{williams2013finding} (also cf.\ \cite{DBLP:journals/siamdm/KowalukLL13,curticapean2014complexity}).
    Essentially, one can exhaustively iterate over the image of the minimum
    vertex-cover in~$G$, which gives rise to $n^{\vc{G}}$ choices; the rest of
    $H$ can then be embedded by dynamic programming.
    Note that $\vc{H}$ may be constant even for large graphs $H$, e.g., if $H$ is a star.
\end{itemize}

In this paper, we unify some of the algorithms above and generalize them to arbitrary subgraph patterns;
in many cases our algorithms are faster.
For two graphs~$H$ and~$G$, let $\#\Sub H G$ be the number of subgraphs of~$G$
that are isomorphic to~$H$.
Our main algorithmic result states that $\#\Sub H G$ can be determined in time $O(n^{t+1})$,
where $t$ is the maximum treewidth (a very popular measure of tree-likeness) among the homomorphic images of $H$.
For our purposes, a homomorphic image of $H$ is any simple graph that can be obtained from~$H$ by possibly merging non-adjacent vertices.
For instance, identifying the first and the last vertex in the $4$-path
\drawgraph{DDW} yields the $4$-cycle \drawgraph{Cr}, and further identifying two
non-adjacent vertices in the $4$-cycle yields the $2$-path \mbox{\drawgraph{BW}.}
We define the \emph{spasm of~$H$} as the set of all homomorphic images of~$H$,
that is, as the set of ``all possible non-edge contractions'' of~$H$.
As an example, for the $4$-path, we have
\begin{align}\label{eq: spasm of P4}
  \spasm{\drawgraph{DDW}}
  =
  \Big\{
    &
    \drawgraph{DDW}
    ,\,
    \drawgraph{CR}
    ,\,
    \drawgraph{CN}
    ,\,
    \drawgraph{Cr}
    ,\,
    \drawgraph{CF}
    ,\,
    \drawgraph{Bw}
    ,\,
    \drawgraph{BW}
    ,\,
    \drawgraph{A_}
  \Big\}
  \,.
\end{align}

Our main algorithmic result can then be stated as follows:

\begin{thm}\label{thm: algorithm count subgraphs}
  Given as input a $k$-edge graph~$H$ and an~$n$-vertex graph~$G$,
  we can compute the number $\#\Sub H G$ in time $k^{O(k)}\cdot n^{t+1}$, 
  where $t$ is the maximum treewidth in the spasm of $H$.
\end{thm}
As an example, for the $4$-path,
the largest treewidth among the graphs in the spasm is~$2$, 
and so Theorem~\ref{thm: algorithm count subgraphs} yields a
running time of~$O(n^3)$ for counting $4$-paths.
In fact, even the $6$-path has only graphs with treewidth at most~$2$ in its
spasm, so the same cubic running time applies.

Theorem~$\ref{thm: algorithm count subgraphs}$ generalizes the vertex-cover based algorithm \cite{williams2013finding} mentioned before:
Merging vertices can never increase the size of the smallest vertex-cover, 
and so the maximum treewidth in the spasm of~$H$ is bounded by~$\vc{H}$,
since the vertex-cover number is an upper bound for the treewidth of a graph.
Note that, while contracting edges cannot increase the treewidth of a graph, 
it is apparent from~\eqref{eq: spasm of P4} that contracting non-edges might.
In fact, if~$H$ is the $k$-edge matching, \emph{all} $k$-edge graphs can be
obtained by contracting non-edges, including expander graphs with
treewidth~$\Omega(k)$.
However, Scott and Sorkin~\cite[Corollary~21]{Scott2007260} proved that every
graph with at most~$k$ edges has treewidth at most $0.174\cdot k+o(k)$.
This bound enables the following immediate corollary to
Theorem~\ref{thm: algorithm count subgraphs}.
\begin{cor}\label{cor: algorithm count subgraphs}
Given as input a $k$-edge graph~$H$ and an~$n$-vertex graph~$G$,
we can compute the number $\#\Sub H G$ in time
  $
  k^{O(k)}\cdot n^{0.174\cdot k+o(k)}.
  $
\end{cor}
Our exponent is obviously smaller than the previously known
$0.455\cdot\abs{V(H)}=0.91\cdot k$ for the $k$-edge matching and
$0.455\cdot k$ for the $k$-path, but somewhat surprisingly, our algorithm is also significantly simpler.

When~$H$ is a triangle, the algorithm from Theorem~\ref{thm:
algorithm count subgraphs} matches the running time~$O(n^3)$ of the exhaustive
search method.
To achieve a smaller exponent, we have to use matrix
multiplication, since faster triangle detection is equivalent to faster Boolean
matrix multiplication~\cite{DBLP:conf/focs/WilliamsW10}.
Indeed, we are able to generalize the~$O(n^\omega)$ time algorithm for counting
triangles to arbitrary graphs~$H$ whose spasm has treewidth at most~$2$.
\begin{thm}
  \label{thm: algorithm count subgraphs matrixmult}
  If all graphs in the spasm of $H$ have treewidth at most two,
  we can compute ${\#\Sub H G}$ in time $f(H)\cdot {|V(G)|}^{\omega}$, where
  $f(H)$ is a function that only depends on~$H$.
\end{thm}
This algorithm applies, for example, to paths of length at most~$6$, thus
providing an alternative and simplified way to obtain the corresponding results
of~\cite{alon1997finding}.

We now turn to hardness results for counting subgraphs.
Here, the vertex-cover number of the pattern~$H$ plays a special role:
When it is bounded by a fixed constant~$b\in\mathbb N$, we have an $n^{b+O(1)}$
time algorithm even when the size of the pattern is otherwise unbounded.
However when it is unbounded, e.g., for $k$-paths, 
the best known running times are of the form $n^{\epsilon k}$ for some $\epsilon\in(0,1)$.
Given these modest improvements for counting subgraph patterns of unbounded vertex-cover number,
it is tempting to conjecture that ``the exponent cannot remain
constant'' for such patterns.
A result by a subset of the authors~\cite{curticapean2014complexity} shows
that this conjecture is indeed true--for an appropriate formalization of
the respective computational problem and under appropriate complexity-theoretic
assumptions.

The \emph{counting exponential time hypothesis} ($\sharpETH$) by Impagliazzo and
Paturi~\cite{IP01}, adapted to the counting setting~\cite{DHMTW14}, states that
there is no $\exp(o(n)) \cdot \poly{m}$-time algorithm to count all satisfying
assignments of a given $3$-CNF formula with $n$ variables and $m$ clauses.
More convenient for us is the following consequence of~$\sharpETH$: There
is no $f(k)\cdot n^{o(k)}$ time algorithm to count all cliques of size
exactly~$k$ in a given~$n$-vertex graph \cite{chen2004tight} --- thus, the
$k$-clique problem is a hard special case of the subgraph counting problem and
clearly $k$-cliques have large vertex-cover number.
Of course, this worst-case hardness of the most general subgraph counting
problem does not directly help us understand the complexity of particular
cases, such as counting $k$-matchings or $k$-paths.

We can model our interest in special cases by restricting the pattern graphs~$H$
to be from a fixed class~$\mathcal H$ of graphs:
The computational problem $\SubProb {\mathcal H}$ is to compute the number
$\#\Sub HG$ of $H$-copies in~$G$ when given two graphs $H\in \mathcal H$ and $G$
as input.
To prove that $\sharpETH$ implies that no fixed-parameter tractable algorithm
can exist for~$\SubProb{\mathcal H}$, one ultimately establishes a
\emph{parameterized reduction} from the $k$-clique problem to the $H$-subgraph
counting problem, which has the important property that $\vc H$ is bounded
by~$g(k)$, a function only depending on~$k$.

The parameterized reduction in~\cite{curticapean2014complexity} was very complex
with various special cases and a Ramsey argument that made~$g$ a very large function.
While it was sufficient to conditionally rule out $f(k)\cdot n^c$ time
algorithms for~$\SubProb{\mathcal H}$ for any constant~$c$ and graph
class~$\mathcal H$ of unbounded vertex-cover number, it left open the
possibility of, for example, $n^{\sqrt{\vc{H}}}$ time algorithms.
Running times of the form~$n^{o(k/\log k)}$ could be ruled out under~$\sharpETH$
only for certain special cases, such as counting $k$-matchings or $k$-paths.
In this paper, we obtain the stronger hardness result for all hard
families~$\mathcal H$.
For technical reasons, we assume that
$\mathcal H$ can be recursively enumerated; without this assumption, we would
however obtain a similar result under a non-uniform version of $\sharpETH$.
\begin{thm}\label{thm: ETH hardness count subgraphs}
  Let $\mathcal H$ be a recursively enumerable class of graphs of unbounded
  vertex-cover number.
  If $\sharpETH$ holds, then 
  $\SubProb{\mathcal H}$
  cannot be computed in time $f(H)\cdot n^{o(\vc H/\log\vc H)}$.
\end{thm}
The $\log$-factor here is related to an open problem in parameterized
complexity, namely whether you can ``beat treewidth''~\cite{marx2007can}, i.e., 
whether there is an algorithm to find~$H$ as a subgraph in time $f(H)\cdot
n^{o(\tw{H})}$, or whether such an algorithm is ruled out by~$\ETH$ for every
graph class~$\mathcal H$ of unbounded treewidth.
Indeed, replacing $\vc{H}$ with $\tw{H}$ in Theorem~\ref{thm: ETH hardness count
subgraphs} essentially yields the hardness result in
\cite[Corollary~6.3]{marx2007can}, and since $\tw{H}\le\vc{H}$, our theorem can
be seen as a strengthening of this result in the counting world.
In fact, our hardness proof is based on this weaker version.

Instead of relying on~$\sharpETH$, we can also consider $\SubProb{\mathcal H}$
from the viewpoint of fixed-parameter tractability.
In this framework, the problem is parameterized by~$|V(H)|$.
A problem is $\sharpWone$-complete if it is equivalent under parameterized
reductions to the problem of counting $k$-cliques, where the allowed
reductions are Turing reductions that run in time~$f(k)\cdot\poly{n}$.
As mentioned before, it is known that $\sharpETH$ implies $\FPT\ne\sharpWone$.
In this setting, Theorem~\ref{thm: ETH hardness count subgraphs} takes on the following form:

\begin{thm}[\cite{curticapean2014complexity}]
  \label{thm: dichotomy-sub}
  Let $\mathcal H$ be a recursively enumerable class of graphs.
  If $\mathcal H$ has bounded vertex-cover number, then $\SubProb{\mathcal H}$
  is polynomial-time computable.
  Otherwise, it is $\sharpWone$-complete when parameterized by the pattern size~$\abs{V(H)}$.
\end{thm}

The original proof of Theorem~\ref{thm: dichotomy-sub} relied on the
$\sharpWone$-com\-plete\-ness of counting $k$-matchings, which was nontrivial on its
own~\cite{DBLP:conf/icalp/Curticapean13,DBLP:conf/iwpec/BlaserC12}.
Then an extensive graph-theoretic analysis was used to find ``$k$-matching
gadgets'' in any graph class~$\mathcal H$ of unbounded vertex-cover number.
These gadgets enable a parameterized reduction from counting $k$-matchings
to~$\SubProb {\mathcal H}$, but they are also responsible for the uncontrollable
blowup in the parameter that lead to highly non-tight results under~$\sharpETH$.
We obtain a much simpler proof of Theorem~\ref{thm:
dichotomy-sub} that does not assign a special role to~$k$-matchings.

Interestingly, Theorem~\ref{thm: dichotomy-sub} implies that no problem
of the form $\SubProb{\mathcal H}$ is ``truly'' FPT.
That is, every such problem that is not $\sharpWone$-complete is in fact already
polynomial-time solvable.
Thus, if we assume the widely-believed claim that $\FPT \neq \sharpWone$ holds,
then Theorem~\ref{thm: dichotomy-sub} exhaustively classifies the
polynomial-time solvable problems $\SubProb {\mathcal H}$.
Indeed, such a sweeping dichotomy would not have been possible by merely
assuming that $\PolyTime \neq \PolyTime^\sharpP$, since there exist
artificial classes~$\mathcal H$ with $\sharpP$-intermediate~\cite{chen2008understanding} $\SubProb {\mathcal H}$.

We remark that a decision version of Theorem~\ref{thm: dichotomy-sub} is an
open problem.
It is known only for graph classes that are hereditary, that is, closed under
induced subgraphs \cite{DBLP:conf/soda/JansenM15}, where the dichotomy criterion
is different.
Moreover, certain non-hereditary cases, such as the $\Wone$-completeness of
deciding the existence of a bipartite clique or a grid, have been resolved only
recently~\cite{DBLP:conf/soda/Lin15,chen2017hardness}.

\subsection{Counting small homomorphisms}
\label{sec: intro homomorphisms}

Similar classifications as Theorem~\ref{thm: dichotomy-sub} for counting subgraphs
were previously known for counting homomorphisms \cite{dalmau2004complexity} from a given class $\mathcal H$.
Recall that a homomorphism from a pattern $H$ to a host~$G$ is a mapping $f:
V(H) \to V(G)$ such that $uv\in E(H)$ implies $f(u)f(v) \in E(G)$; we write $\#\Hom HG$ for the number of such
homomorphisms.
In the context of pattern counting problems, we can interpret homomorphisms as a
relaxation of the subgraph notion: 
If a homomorphism $f$ from $H$ to $G$ is injective, then it constitutes a
subgraph embedding from $H$ to $G$.
However, since we generally do \emph{not} require injectivity in homomorphisms,
these may well map into other \emph{homomorphic images}.
(Recall that the spasm of~$H$ contains exactly the loop-free homomorphic images of~$H$.)

From an algorithmic viewpoint, not requiring injectivity makes counting patterns
easier: One can now use separators to divide $H$ into subpatterns and compose
mappings of different subpatterns to a global one for~$H$ via dynamic programming.
In this process, only the locations of separators under the subpattern mapping
need to be memorized, while non-separator vertices of subpatterns may be forgotten.
As an example, note that counting $k$-paths is $\sharpWone$-complete by
Theorem~\ref{thm: dichotomy-sub}, but counting homomorphisms from a $k$-path to
a host graph $G$ is polynomial-time solvable.
Indeed, the latter problem amounts to counting $k$-walks in $G$, 
which can be achieved easily by taking the $k$-th power of the adjacency matrix of~$G$, 
a process that can be interpreted as a dynamic programming algorithm:
Given a table with the number of $\ell$-walks from $s$ to $u$ for each~$u$, we
can compute a table with the number of~$(\ell+1)$-walks from~$s$ to~$v$ for
all~$v$. That is, we only need to store the last vertex seen in a walk.
This idea can be generalized easily to graphs of bounded treewidth 
using a straightforward dynamic programming approach on the tree decomposition of~$H$.
\begin{prop}[D{\'{\i}}az et al.~\cite{DBLP:journals/tcs/DiazST02}]%
  \label{prop: hom-treewidth-algo}
  There is a deterministic $\exp\paren*{O(k)}+\poly{k}\cdot n^{\tw H+1}$ time algorithm
  to compute the number of homomorphisms from a given graph~$H$ to a given
  graph~$G$, where $k=\abs{V(H)}$, $n=\abs{V(G)}$, and $\tw H$ denotes the treewidth of~$H$.
\end{prop}
This proposition is the basis of our main algorithmic result for subgraphs
(Theorem~\ref{thm: algorithm count subgraphs}) and other graph motif
parameters, so we include a formal proof in the appendix.
For graphs~$H$ of treewidth at most two, we can speed up the
dynamic programming algorithm by using fast matrix multiplication:
\begin{thm}\label{thm: algorithm hom matrixmult}
  If $H$ has treewidth at most~$2$, we can compute $\#\Hom H{G}$ in time
  $\poly{\abs{V(H)}}\cdot |V(G)|^\omega$.
\end{thm}

Apart from being more well-behaved than subgraphs in terms of algorithmic tractability, 
homomorphisms also allow for simpler hardness proofs: 
Several constructions are significantly easier to analyze for homomorphisms than
for subgraphs, as we will see in our proofs.
This might explain why a dichotomy for counting homomorphisms from a fixed class
$\mathcal H$ was obtained an entire decade before its counterpart for subgraphs;
it establishes the treewidth of~$\mathcal H$ as the tractability criterion for
the problems $\HomProb{\mathcal H}$.

\begin{thm}[Dalmau and Jonsson \cite{dalmau2004complexity}]
  \label{thm: dalmau jonsson}
  Let $\mathcal H$ be a recursively enumerable class of graphs.
  If~$\mathcal H$ has bounded treewidth, then the problem
  $\HomProb {\mathcal H}$ of counting homomorphisms from graphs in $\mathcal H$
  into host graphs is polynomial-time solvable.
  Otherwise, it is $\sharpWone$-complete when parameterized by~$\abs{V(H)}$.
\end{thm}
In the remainder of the paper, we will also require the following lower bound under $\sharpETH$, the proof of which is a
simple corollary of~\cite[Corollary~6.2 and~6.3]{marx2007can}.
\begin{prop}\label{prop: hom ETH hard}
  Let $\mathcal H$ be a recursively enumerable class of graphs of unbounded
  treewidth.
  If $\sharpETH$ holds, then there is no $f(H)\cdot \abs{V(G)}^{o(\tw H / \log \tw
  H)}$ time algorithm to compute $\HomProb {\mathcal H}$ for graphs~$H\in\mathcal H$ and~$G$.
\end{prop}

Finally, we remark that the decision version of Theorem~\ref{thm: dalmau jonsson}, that
is, the dichotomy theorem for deciding the existence of a homomorphism
from~${H\in\mathcal H}$ to~$G$ is known and has a different
criterion~\cite{DBLP:journals/jacm/Grohe07}:
The decision problem is polynomial-time computable even when only the
homomorphic cores of all graphs in~$\mathcal H$ have bounded treewidth, and it
is $\Wone$-hard otherwise.

\subsection{Counting small induced subgraphs}
\label{sec: intro induced subgraphs}

Let us also address the problem of counting small induced subgraphs from a class $\mathcal H$.
This is a natural and well-studied variant of counting subgraph copies
\cite{kloks2000finding,chen2008understanding,jerrum2015parameterised,jerrum2015some,jerrum2014parameterised,meeks2016challenges}, and for several applications it
represents a more appropriate notion of ``pattern occurrence''.
From the perspective of dichotomy results however, it is less intricate than subgraphs or homomorphisms:
Counting induced subgraphs is known to be $\sharpWone$-hard for any infinite pattern class $\mathcal H$, 
and even the corresponding decision version is $\Wone$-hard.
\begin{thm}[\cite{chen2008understanding}]\label{thm: chen weyer}
Let $\mathcal H$ be a recursively enumerable class of graphs.
If $\mathcal H$ is finite, then the problem $\IndProb {\mathcal H}$ of counting induced subgraphs from $\mathcal H$ is polynomial-time solvable.
  Otherwise, it is $\sharpWone$-complete when parameterized by~$\abs{V(H)}$.
\end{thm}
Jerrum and Meeks~\cite{jerrum2015parameterised,jerrum2015some,jerrum2014parameterised,meeks2016challenges}
introduced the following generalization of the problems $\IndProb {\mathcal H}$
to fixed graph properties $\Phi$:
Given a graph $G$ and $k \in \mathbb N$, the task is to compute the number of
induced $k$-vertex subgraphs that have property~$\Phi$.
Let us call this problem~$\IndPropProb{\Phi}$.
They identified some classes of properties $\Phi$ that render this problem
$\sharpWone$-hard.
Using our machinery, we get a full dichotomy theorem for this class of problems.
\begin{thm}[simple version]\label{cor: jerrum meeks dichotomy}
  If $\Phi$ is a decidable graph property, then
  $\IndPropProb{\Phi}$ is fixed-parameter tractable or $\sharpWone$-hard when parameterized by~$\abs{V(H)}$.
\end{thm}

\subsection{A unified view: Graph motif parameters}
\label{sec: intro lovasz}

We now discuss our proof techniques on a high level.
From a conceptual perspective, our most important contribution lies in finding a framework for understanding the parameterized complexity of subgraphs, induced subgraphs, and homomorphisms in a uniform context.
Note that we quite literally \emph{find} this framework: That is, we do not develop
it ourselves, but we rather adapt works by Lovász et al.~dating back to the
1960s~\cite{lovasz1967operations,borgs2006counting}.
The most important observation is the following:
\begin{quote}
  \textit{Many counting problems are actually linear combinations of
  homomorphisms in disguise!}
\end{quote}
That is, there are elementary transformations to express, say, linear
combinations of subgraphs as linear combinations of homomorphisms, and vice
versa.

The algorithms for subgraphs (Theorem~\ref{thm: algorithm count subgraphs} and
Corollary~\ref{cor: algorithm count subgraphs}) are based on a reduction from
subgraph counting to homomorphism counting, so we want to find relations between
the number of subgraphs and the number of homomorphisms.
To get things started, note that injective homomorphisms from~$H$ to~$G$, also
called \emph{embeddings}, correspond to a subgraph~$F$ of~$G$ that is isomorphic
to~$H$, and in fact, the number $\#\Emb{H}{G}$ of embeddings is equal to the
number $\#\Sub{H}{G}$ of subgraphs times $\#\Aut{H}$, the number of
automorphisms of~$H$.

Homomorphisms cannot map two adjacent vertices of~$H$ to the same vertex of~$G$,
assuming that~$G$ does not have any loops.
For instance, every homomorphism from $\drawgraph{Bw}$ to~$G$ must be injective,
and therefore the number of triangles in~$G$ is equal to the number of such
homomorphisms, up to a factor of~$6$: the number of automorphisms of the
triangle.
Formally, we have
$
\#\Sub{\drawgraph{Bw}}{G}
=
\frac 16\cdot
\#\Hom{\drawgraph{Bw}}{G}$
for every graph~$G$ that does not have loops.

More interesting cases occur when homomorphisms from~$H$ to~$G$ are not automatically injective.
Clearly, the set of all homomorphisms contains the injective ones, which
suggests we should simply count all homomorphisms and then subtract the ones
that are not injective.
Any non-injective homomorphism~$h$ from~$H$ to~$G$ has the property that there
are at least two (non-adjacent) vertices that it maps to the same vertex; in
other words, its image~$h(H)$ is isomorphic to some member of~$\spasm{H}$ other
than~$H$ itself.
For example,
$\#\Hom {\drawgraph{BW}} G
-\#\Hom {\drawgraph{A_}} G$
is the number of injective homomorphisms from \drawgraph{BW} to~$G$ since the
only way for such a homomorphism to be non-injective is that it merges the two
degree-1 vertices.
In general, the number~$\#\Emb HG$ of injective
homomorphisms is
\begin{equation*}
\#\Hom HG-\sum_{F\in\spasm{H}\setminus\set{H}} \#\Emb{F}{G}\,.
\end{equation*}
Since each such~$F$ is strictly smaller than~$H$, this fact yields a recursive
procedure to compute~$\#\Emb HG$.
However, there is a better way: We can use Möbius inversion over the partition
lattice to obtain a closed formula.
We already mentioned that the spasm of $H$ can be obtained by consolidating
non-adjacent vertices of~$H$ in all possible ways.
This means that we consider partitions~$\rho$ of~$V(H)$ in which each block is an
independent set, and then form the \emph{quotient graph}~$\contract
H\rho$ obtained from~$H$ by merging each block of~$\rho$ into a single vertex.
To express the injective homomorphisms from $H$ to~$G$ (and hence the number of
$H$-subgraphs) as a linear combination of homomorphisms, we consider all
possible types in which a homomorphism~$h$ from~$H$ to~$G$ can fail to be
injective.
More precisely, we define this type $\rho_h$ of $h$ to be the partition of
$V(H)$, where each block is the set of vertices of~$H$ that map to the same
vertex of~$G$ under~$h$.
The homomorphism $h$ is injective if and only if~$\rho_h$ is the finest
partition, i.e., the partition where each block has size one.

The homomorphisms from $\contract H \rho$ to $G$ are precisely those
homomorphisms from~$H$ to~$G$ that fail to be injective ``at least as badly
as~$\rho$'', that is, those homomorphisms $f$ whose type $\rho_f$ is a
coarsening of $\rho$.
As remarked by Lovász~et~al.~\cite{lovasz1967operations,borgs2006counting}, one
can then use Möbius inversion, a generalization of the inclusion--exclusion
principle, to turn this observation into the ``inverse'' identity
\begin{align}
\label{eq: emb2hom-intro}
  \#\Sub HG
  &=
  \sum_\rho
  \frac{
  (-1)^{\abs{V(H)}-\abs{V(\contract H\rho)}}
  \cdot
  \prod_{B\in\rho} (|B|-1)!
  }{\#\Aut H}
  \cdot \#\Hom {\contract H\rho}G
  \,.
\end{align}
The sum in \eqref{eq: emb2hom-intro} ranges over all partitions~$\rho$ of $V(H)$.
Hence, the number of $H$-subgraphs in~$G$ is equal to a linear combination of the
numbers of homomorphisms from graphs~$\contract H\rho$ to~$G$.
Each $\contract H\rho$ is isomorphic to a graph in $\spasm{H}$, and
so by collecting terms for isomorphic graphs,~\eqref{eq: emb2hom-intro}
represents the number of $H$-subgraphs as a linear combination of homomorphism
numbers from graphs~$F$ in the spasm of~$H$; see Figure~\ref{fig: P4 spasms} for
an example.

\begin{figure}[tp]
  \begin{align*}
    &
    \Sub {\drawgraph{DDW}} \star
    \;=
    \\
    &
    \begin{array}{ll}
      &
      \makebox[1.5em][l]{$\frac12$}
      \Hom{\drawgraph{DDW}}\star
      \\
      -&
      \makebox[1.5em][l]{       }
      \Hom {\drawgraph{CR}} \star
      \;-\;
      \makebox[1.5em][l]{       }
      \Hom {\drawgraph{CN}} \star
      \;-\;
      \makebox[1.5em][l]{$\frac12$}
      \Hom {\drawgraph{Cr}} \star
      \;-\;
      \makebox[1.5em][l]{$\frac12$}
      \Hom {\drawgraph{CF}} \star
      \\
      +&
      \makebox[1.5em][l]{$\frac32$}
      \Hom {\drawgraph{Bw}} \star
      \;+\;
      \makebox[1.5em][l]{$\frac52$}
      \Hom {\drawgraph{BW}} \star
      \\
      -&
      \makebox[1.5em][l]{       }
      \Hom {\drawgraph{A_}} \star
      \,.
    \end{array}
  \end{align*}
  \caption[%
    An example for \eqref{eq: emb2hom-intro} for the path with four edges%
    ]{\label{fig: P4 spasms}%
    An example for \eqref{eq: emb2hom-intro}, where $H$ is 
    the path \drawgraph{DDW} with four edges.
    The number of subgraphs is represented as a linear combination 
    of homomorphisms from graphs $F\in\spasm{H}$.
    Each such~$F$ has treewidth at most two, so we can compute the homomorphism
    numbers in time $O(n^3)$ via Proposition~\ref{prop: hom-treewidth-algo}.
    Computing the linear combination on the right side yields an $O(n^3)$-time
    algorithm to count $4$-paths, and in fact this is the algorithm in
    Theorem~\ref{thm: algorithm count subgraphs}.
  }
\end{figure}

The identity~\eqref{eq: emb2hom-intro} can be viewed as a basis transformation
in a certain vector space of graph parameters, and we formalize this perspective
in \S\ref{sec: graph motif parameters}.
A similar identity turns out to hold for counting induced subgraphs as well, so
all three graph parameter types can we written as finite linear combinations of
each other.
This motivates the notion of a \emph{graph motif parameter}, which is any graph parameter~$f$ that is a finite linear combination of induced subgraph numbers.
That is, there are coefficients $\alpha_1,\dots,\alpha_t\in\mathbb Q$ and graphs $H_1,\dots,H_t$ such that, for all graphs~$G$, we have
\begin{equation}
  \label{eq: first-linear-combination}
  f(G) = \sum_{i=1}^t \alpha_i \cdot \#\IndSub{H_i}G\,.
\end{equation}

We study the problem of computing graph motif parameters~$f$.
For our results in parameterized complexity, we parameterize this problem by the
description length~$k$ of $\alpha_1,\dots,\alpha_t$ and $H_1,\ldots, H_t$.
Due to the basis transformation between induced subgraphs, subgraphs, and
homomorphisms, writing $\#\Hom{H_i}{G}$ instead of ${\#\IndSub{H_i}{G}}$ in~\eqref{eq: first-linear-combination} yields an equivalent class of problems --- switching bases only leads
to a factor~$g(k)$ overhead in the running time for some computable function~$g$, which we can neglect for
our purposes.

Our main result is that the complexity of computing any graph parameter~$f$ is exactly governed by the maximum complexity of counting the homomorphisms occurring in its representation over the homomorphism basis.
More precisely, let~$\alpha_1,\dots,\alpha_t\in\mathbb Q$ and $H_1,\dots,H_t$ be graphs with $f(G)=\sum_i \alpha_i \cdot \#\Hom{H_i}{G}$ for all graphs~$G$.
Our algorithmic results are based on the following observation: If each
$\#\Hom{H_i}{G}$ can be computed in time $O(n^c)$ for $n=\abs{V(G)}$ and some constant~$c\ge 0$, then~$f(G)$ can be computed in time $O(n^c)$ for the same constant~$c$.
However, we show that the reverse direction also holds: If $f$ can be
computed in time $O(n^c)$ for some $c\ge 0$, then each $\#\Hom{H_i}{G}$ with $\alpha_i\ne0$ can be computed in time $O(n^c)$ for the same constant~$c$.
The reduction that establishes this fine-grained equivalence gives rise to our
results under $\sharpETH$ and our new $\sharpWone$-hardness proof.
Note that such an equivalence is \emph{not} true for linear combinations of
embedding numbers, as can be seen from the following example.
\begin{example}\label{example: 4walks}
  Consider the following linear combination:
  \begin{align*}
    &\makebox[1.5em][l]{                } \Emb {\drawgraph{DDW}} \star
    \;+\;\makebox[1.5em][l]{                } \Emb {\drawgraph{Cr}} \star
    +\;
    \makebox[1.5em][l]{                } \Emb {\drawgraph{CF}} \star
    \;+\;\makebox[1.5em][l]{$ 2\cdot\mbox{}$} \Emb {\drawgraph{CN}} \star
    \\
    +\;
    &\makebox[1.5em][l]{$ 2\cdot\mbox{}$} \Emb {\drawgraph{CR}} \star
    \;+\;\makebox[1.5em][l]{$ 3\cdot\mbox{}$} \Emb {\drawgraph{Bw}} \star
    +\;
    \makebox[1.5em][l]{$ 4\cdot\mbox{}$} \Emb {\drawgraph{BW}} \star
    \;+\;\makebox[1.5em][l]{                } \Emb {\drawgraph{A_}} \star
    \,.
  \end{align*}
  When this linear combination of embeddings is transformed into the homomorphism 
  basis via~\eqref{eq: emb2hom-intro}, most terms cancel, and it turns out that
  it is equal to $\Hom {\drawgraph{DDW}} \star$, that is, it counts the number
  of walks of length~$4$.
  Counting $4$-walks can be done in time~$O(n^2)$ via
  Proposition~\ref{prop: hom-treewidth-algo}, but counting, for example,
  triangles is not known to be possible faster than~$O(n^\omega)$.
  More generally, counting walks of length~$k$ is in~$O(n^2)$-time, but counting
  \emph{paths} of length~$k$ is $\sharpWone$-hard and not in time
  $g(k)\cdot n^{o(k/\log k)}$ under $\sharpETH$.
\end{example}
Thus, even a linear combination of subgraph numbers that looks complex at first
and contains as summands embedding numbers that are fairly hard to compute can
actually be quite a bit easier due to cancellation effects that occur when
rewriting it as linear combination of homomorphism numbers.

As in the case of subgraphs in~\S\ref{sec: intro subgraphs}, we consider classes
$\mathcal A$ of linear combinations to get more expressive hardness results.
That is, each element of $\mathcal A$ is a pattern-coefficient list
$(\alpha_1,H_1),\dots,(\alpha_t,H_t)$ as above.
The evaluation problem $\IndProb{\mathcal A}$ for graph motif parameters from $\mathcal A$ is then given a
pattern-coefficient list from~$\mathcal A$ and a graph~$G$, and is supposed to
compute the linear combination~\eqref{eq: first-linear-combination}.
We have the following result for the complexity of computing graph motif
parameters.
\begin{thm}[intuitive version]
  \label{thm: lincomb indsub}
  Let $\mathcal A$ be a recursively enumerable class of pattern-coefficient
  lists.
  If the linear combinations \eqref{eq: first-linear-combination}, re-expressed
  as linear combinations of homomorphisms, contain non-zero coefficients only for
  graphs of treewidth at most~$t$, then the problem $\IndProb{\mathcal A}$ can be
  computed in time $f(\alpha)\cdot n^{t+1}$.
  Otherwise, the problem is $\sharpWone$-hard parameterized by $\abs{\alpha}$
  and does not have $f(\alpha) \cdot n^{o(t/\log t)}$ time algorithms under
  $\sharpETH$.
\end{thm}

With respect to fixed-parameter tractability vs.\ $\sharpWone$-hardness, this
theorem fully classifies the problems $\IndProb {\mathcal A}$ for fixed classes
of linear combinations $\mathcal A$.
Of course we have similar (equivalent) formulations for $\SubProb{\mathcal A}$
and $\HomProb{\mathcal A}$, and thus we generalize the dichotomy theorems for
subgraphs (Theorem~\ref{thm: dichotomy-sub}), homomorphisms (Theorem~\ref{thm:
dalmau jonsson}), and induced subgraphs (Theorem~\ref{thm: chen weyer}).
The dichotomy criterion is somewhat indirect; it addresses $\mathcal A$ only
through its representation as a linear combination of homomorphism numbers.
However, we do not believe that there is a more 'native' dichotomy criterion on
$\mathcal A$ since seemingly complicated linear combinations can turn out to be
easy -- Example~\ref{example: 4walks} gives an indication of this phenomenon;
perturbing the coefficients just a tiny bit can turn a computationally easy
linear combination into one that is hard.

Nevertheless, we can exhibit some interesting sufficient conditions.
For example, for the problem $\SubProb{\mathcal A}$, if all linear combinations
of $\mathcal A$ in fact feature exactly one pattern (as in the situation
of Theorem~\ref{thm: dichotomy-sub}), then the linear combination re-expressed
over homomorphisms uses graphs of unbounded treewidth if and only if the
patterns in $\mathcal A$ have bounded vertex-cover number.
We can hence recover Theorem~\ref{thm: dichotomy-sub} from Theorem~\ref{thm:
lincomb indsub}.

\subsection{Counting vertex-colored subgraphs}
\label{sec: intro colored subgraphs}

The techniques introduced above are sufficiently robust to handle
generalizations to, e.g., the setting of vertex-colored subgraphs, where the
vertices of $H$ and $G$ have colors and we count only subgraphs of $G$ with
isomorphisms to $H$ that respect colors.
A dichotomy for the special case of {\em colorful} patterns, where every vertex
of the pattern $H$ has a different color, follows from earlier results by
observing that embeddings and homomorphisms are the same for colorful patterns.
For colorful patterns, bounded treewidth is the tractability criterion.

\begin{thm}[\cite{dalmau2004complexity,curticapean2014complexity,meeks2016challenges}]
\label{thm: dichotomy-sub-colorful}
Let $\mathcal H$ be a recursively enumerable class of colorful vertex-colored graphs.
If~$\mathcal H$ has bounded treewidth, then the problem $\SubProb {\mathcal H}$ of counting colorful subgraphs from $\mathcal H$ in vertex-colored host graphs is polynomial-time solvable.
Otherwise, it is $\sharpWone$-complete when parameterized by~$\abs{V(H)}$.
\end{thm}

Theorems~\ref{thm: dichotomy-sub} and \ref{thm: dichotomy-sub-colorful} characterize the two extreme cases of counting colored subgraphs: the uncolored and the fully colorful cases. But there is an entire spectrum of colored problems in between these two extremes. 
 What happens when we consider
  vertex-colored graphs with some colors appearing on more
  than one vertex? As we gradually move from colorful to
  uncolored graphs, where exactly is the point when a jump in
  complexity occurs?
 Answering such questions can be nontrivial even for simple patterns such as paths and matchings 
 and can depend very much on how the colors appear on the pattern. 
 Fortunately, by a basis change to (vertex-colored) homomorphisms via \eqref{eq: emb2hom-intro}, 
 we can answer such questions as easily as in the uncolored setting. 
 The only technical change required is that we should consider only partitions $\rho$ that respect the coloring of~$H$, 
 that is, the vertices of~$H$ that end up in the same block should have the same color.
 With these modifications, we can derive the following corollary from
 Theorem~\ref{thm: lincomb indsub}

\begin{thm}
\label{thm: dicho-sub-color}
Let $\cH$ be a recursively enumerable class of vertex-colored patterns (or linear combinations thereof)
and let $\mathcal A_\mathit{hom}$ be the class of linear combinations of homomorphisms derived from $\cH$ by the identity~\eqref{eq: emb2hom-intro} as discussed above.
If there is a finite bound on the treewidth of graphs in $\mathcal A_\mathit{hom}$, then $\SubProb {\cH}$ is FPT.
Otherwise, the problem is $\sharpWone$-hard when parameterized by~$\abs{V(H)}$.
\end{thm}

Theorem~\ref{thm: dicho-sub-color} raises a number of questions.
First, being a corollary of Theorem~\ref{thm: lincomb indsub}, the tractability
criterion is quite indirect, whereas we may want to have a more direct
structural understanding of the FPT cases. Secondly, Theorem~\ref{thm:
dicho-sub-color} does not tell us whether the FPT cases are actually
polynomial-time solvable or not.
It is quite remarkable that in Theorems~\ref{thm: dichotomy-sub} and \ref{thm: dichotomy-sub-colorful}, all FPT cases are actually polynomial-time solvable, leaving no room for ``true'' FPT cases that are not polynomial-time solvable. It turns out however that, if we consider vertex-colored patterns in their full generality, then such pattern classes actually \emph{do} appear. A prime example of this phenomenon is the case of \emph{half-colorful matchings}, which are vertex-colored $k$-matchings such that one endpoint of each edge $e_i$ for $i\in \set{1,\ldots,k}$ is colored with $0$, while the other is colored with~$i$.
Since counting perfect matchings in bipartite graphs is $\sharpP$-hard, a
trivial argument shows that counting half-colorful matchings is also
$\sharpP$-hard: if a bipartite graph with $n+n$ vertices is colored such
that one part has color 0 and each vertex of the other part
has a distinct color from 1 to $n$, then the number of
half-colorful matchings of size $n$ is exactly the number of perfect
matchings. On the other hand, it is not difficult to show that
counting the number of half-colorful matchings of size $k$ in a graph
colored with colors $0$, $1$, $\dots$, $k$ is fixed-parameter
tractable. It is essentially a dynamic programming exercise: for any
subgraph~${H'\subseteq H}$ of the half-colorful matching and for any
integer~$i$, we want to compute the number of subgraphs of~$G$
isomorphic to $H'$ that are allowed to use only the first $i$ vertices
of color class 0.

We give a complete classification of
the polynomial-time solvable cases of counting colored patterns from a class $\mathcal H$. 
For classes of patterns (but not linear combinations thereof), we 
refine the FPT cases of Theorem~\ref{thm: dicho-sub-color} into two
classes: the polynomial-time solvable cases, and those that are
not polynomial-time solvable, assuming the
\emph{Nonuniform Counting Exponential Time Hypothesis}. This shows
that the existence of half-colorful matchings is the canonical reason
why certain classes of patterns require dynamic programming and
therefore the full power given by the definition of FPT: those cases
are polynomial where the size of the largest half-colorful matching
appearing as a subgraph is at most logarithmic in the size of the pattern.

\subsection*{Organization of the paper}

In~\S\ref{sec: prelim}, we provide basics on parameterized complexity and the graph-theoretical notions used in this paper.
We formalize \emph{graph motif parameters} in~\S\ref{sec: graph motif parameters},
and in~\S\ref{sec: graph motif parameter relations} we show how to switch between different useful representations of graph motif parameters.
In \S\ref{sec: graph motif parameter complexity}, we then address computational aspects of graph motif parameters.
These results are first put to use in~\S\ref{sec: algo}, where
we count subgraph patterns by reduction to homomorphisms.
In \S\ref{sec: complexity of lincombs}, we prove hardness results for linear combinations of subgraphs and induced subgraphs
under $\sharpETH$ 
and
$\FPT\ne\sharpWone$.
Finally, we prove our results for counting vertex-colored subgraphs in~\S\ref{sec:polytime}.

\section{Preliminaries}
\label{sec: prelim}

For a proposition~$P$, we use the \emph{Iverson bracket} $[P]\in\set{0,1}$ to indicate whether~$P$ is satisfied.
For a potentially infinite matrix~$M$, a \emph{principal submatrix} $M_S$ is a
submatrix of~$M$ where the selected row and column index sets are the same
set~$S$.

\subsection{Parameterized complexity theory}

We refer to the
textbooks~\cite{cygan2015parameterized,FlumGrohebook,Niedermeierbook06} for
background on parameterized complexity theory.  Briefly, a
\emph{parameterized counting problem} is a function $\Pi:\set{0,1}^*\to\mathbb
N$ that is endowed with a \emph{parameterization} $\kappa:\set{0,1}^*\to\mathbb
N$; it is \emph{fixed-parameter tractable} ($\FPT$) if there is a computable
function $f:\mathbb N\to\mathbb N$ and an algorithm to compute $\Pi(x)$ in time $f(k)\cdot \poly n$, where $n=\abs{x}$ and $k=\kappa(x)$.

A \emph{parameterized Turing reduction} is a Turing reduction from a parameterized problem
$(\Pi,\kappa)$ to a parameterized problem $(\Pi',\kappa')$ such that the reduction runs in
$f(\kappa (x))\cdot \poly{|x|}$ time on instances $\Pi(x)$ and each oracle query~$\Pi'(y)$ satisfies
$\kappa'(y)\le g(k)$.
Here, both $f$ and $g$ are computable functions.
A parameterized problem is \emph{$\sharpWone$-hard} if there is a parameterized Turing
reduction from the problem of counting the $k$-cliques in a given graph;
since it is believed that the latter does not have an $\FPT$-algorithm,
$\sharpWone$-hardness is a strong indicator that a problem is not $\FPT$.

The exponential time hypothesis ($\ETH$) by Impagliazzo and Paturi~\cite{IP01} asserts 
that satisfiability of $3$-CNF formulas cannot be decided substantially faster 
than by trying all possible assignments.
The counting version of this hypothesis~\cite{DHMTW14} states that
there is a constant $c>0$ such that no deterministic algorithm
can compute \textup{\#$3$-SAT} in time~$\exp(c\cdot n)$, where $n$ is the number of variables.
  
\noindent
Chen et al.~\cite{chen2004tight} proved that $\ETH$ implies the
hypothesis that there is no $f(k)\cdot n^{o(k)}$-time algorithm to decide whether an $n$-vertex graph~$G$ contains a $k$-clique.
Their reduction is parsimonious, so~$\sharpETH$ rules out
$f(k)\cdot n^{o(k)}$ time algorithms for counting $k$-cliques.

\subsection{Graphs, subgraphs, and homomorphisms}

Let $\Graphsl$ be the set of all labeled, finite, undirected, and simple graphs;
in particular, these graphs contain neither loops nor parallel edges.
That is, there is a suitable fixed and countably infinite universe~$U$,
and~$\Graphsl$ contains all finite graphs~$G$ with vertex set $V(G)\subseteq U$
and edge set $E(G)\subseteq\binom{V(G)}{2}$.

\paragraph{Subgraphs.}

If $G$ is a graph, a \emph{subgraph} $F$ of $G$ is a graph with $V(F)\subseteq V(G)$ and 
$E(F)\subseteq E(G)$, and $F$ is an \emph{induced subgraph} of $G$ if it is a
subgraph with the additional property that, for all $uv\not\in E(F)$ we
have ${uv\not\in E(G)}$.
The set of subgraphs of $G$ that are isomorphic to $H$ is denoted with $\Sub
H G$, and the set of induced subgraphs of $G$ that are isomorphic to $H$ is
denoted with $\IndSub H G$.

\paragraph{Homomorphisms and related notions.}

If $H$ and $G$ are graphs, a \emph{homomorphism} from $H$ to $G$ is a
function $f:V(H)\to V(G)$ such that edges map to edges under~$f$.
That is, for all $\set{u,v} \in E(H)$, we have $\set{f(u),f(v)} \in E(G)$.
The set of all homomorphisms from $H$ to $G$ is denoted with $\Hom H G$.
\emph{Embeddings} are injective homomorphisms, and we denote the corresponding set with $\Emb H G$.
\emph{Strong embeddings} are embeddings with the additional property that 
non-edges map to non-edges, that is, ${\set{f(u),f(v)}\not\in E(G)}$ holds for all 
$\set{u,v}\notin E(G)$. We denote the set of strong embeddings with $\StrEmb H G$.
A homomorphism~$f\in\Hom{H}{G}$ is \emph{surjective} if it hits all vertices and
edges of $G$, that is, $f(V(H))=V(G)$ and $f(E(H))=E(G)$ hold. 
Note that this is a stronger requirement than $f$ being surjective on its codomain.
An \emph{isomorphism} from~$H$ to~$G$ is a strong embedding from $H$ to $G$ that 
is also surjective, and it is an \emph{automorphism} if additionally $H=G$ 
holds. We write $H \simeq G$ if $H$ and $G$ are isomorphic.

\paragraph{Colored graphs.}
We also use vertex-colored graphs~$G$, where each vertex has a color
from a finite set~$C$ of colors via a function~$f:V(G)\to
C$.
We note that such a coloring is not necessarily proper, in the sense that any two adjacent vertices need to receive distinct colors.
Each set $V_i(G) := f^{-1}(i)$ for $i\in C$ is a \emph{color class} of $G$.
A subgraph~$H$ of~$G$ is called (vertex-)\emph{colorful} if~$V(H)$ intersects
each color class in exactly one vertex.

\paragraph{Treewidth.}
\label{sub: treewidth}%
A \emph{tree decomposition} of a graph~$G$ is a pair $(T,\bag)$, where $T$ is a tree and $\bag$ is a mapping from~$V(T)$ to $2^{V(G)}$ such that, for all vertices $v\in V(G)$, the
set $\setc{t\in V(T)}{v\in\bag(t)}$ is nonempty and connected in $T$, and for all edges $e\in E(G)$, there is some node $t\in
V(T)$ such that $e\subseteq\bag(t)$.
The set $\bag(t)$ is the \emph{bag at~$t$}.
The \emph{width} of $(T,\bag)$ is the integer $\max\setc{|\bag(t)|-1}{t\in
V(T)}$, and the {\em treewidth} $\tw{G}$ of a graph~$G$ is the minimum possible
width of any tree decomposition of~$G$.

It will be convenient for us to view the tree $T$ as being
directed away from the root, and we define the following mappings
$\sep,\cone,\comp:V(T)\to2^{V(G)}$ for all $t\in V(T)$:
\begin{align}
  \label{eq:bagsep}
\text{(the \emph{separator at $t$}) \qquad }
  \sep(t)&=
  \begin{cases}
    \emptyset&\text{if $t$ is the root of $T$},\\
    \bag(t)\cap\bag(s)&\text{if $s$ is the parent of $t$ in $T$},
  \end{cases}\\
  \text{(the \emph{cone at $t$}) \qquad }
  \label{eq:bagcone}
  \cone(t)&=\bigcup_{\text{$u$ is a descendant of $t$}}\bag(u),\\
  \text{(the \emph{component at $t$}) \qquad }
  \label{eq:bagcomp}
  \comp(t)&=\cone(t)\setminus\sep(t).
\end{align}
The \emph{adhesion} of $(T,\bag)$ is defined as $\max\setc{|\sep(t)|}{t\in V(T)}$.
The following conditions are easily checked:
\begin{nlist}{TD}
  \item\label{li:t1}
    $T$ is a directed tree.
  \item\label{li:t2} For all $t\in V(T)$ we have
    $\comp(t)\cap\sep(t)=\emptyset$ and
    $N^G(\comp(t))\subseteq\sep(t)$.

  \item\label{li:t3} For all $t\in V(T)$ and $u\in N_+^T(t)$ we have
    $\comp(u)\subseteq\comp(t)$ and $\cone(u)\subseteq\cone(t)$.
  \item\label{li:t4} For all $t\in V(T)$ and all distinct $u_1,u_2\in N_+^T(t)$
    we have
  $\cone(u_1)\cap\cone(u_2)=\sep(u_1)\cap\sep(u_2)$.
  \item\label{li:t5} For the root $r$ of $T$ we have
    $\sep(r)=\emptyset$ and $\comp(r)=V(G)$.
\end{nlist}
Conversely, consider any triple $(T,\sep,\comp)$ where $T$ is a directed graph
and $\sep,\comp$ are functions of type ${V(T)\to 2^{V(G)}}$.
We can then define functions $\cone,\bag:V(T)\to2^{V(G)}$ such that, for all~$t\in V(T)$, we have 
$\cone(t)=\sep(t)\cup\comp(t)$ and 
$\bag(t)=\cone(t)\setminus\bigcup_{u\in N_+^T(t)}\comp(u)$.
If \ref{li:t1}--\ref{li:t5} are satisfied, then it can be verified that $(T,\bag)$ is a tree
decomposition (see \cite{gro10+a} for a proof).
Hence \ref{li:t1}--\ref{li:t5} yield an alternative definition of tree
decompositions, which we may use as is convenient.

\section{The space of graph motif parameters}

\label{sec: graph motif parameters}

We develop our interpretation of the general setup of Lovász~\cite{lovaszbook}.
To this end, it will be useful to consider unlabeled graphs:
For concreteness, we say that a graph $H\in\Graphsl$ is \emph{unlabeled}
if it is \emph{canonically labeled}, that is, if it is the lexicographically
first graph that is isomorphic to~$H$.
Then the set~$\Graphsu$ of unlabeled graphs is the subset of~$\Graphsl$
that contains exactly the canonically labeled graphs.
\emph{Graph parameters} are functions $f:\Graphsl\to\mathbb Q$ that are
invariant under isomorphisms, and we view them as
functions~$f:\Graphsu\to\mathbb Q$.

For all $H,G\in\Graphsu$, we define $\indsub(H,G)$ as the number of (labeled) induced subgraphs of~$G$ that are isomorphic to~$H$.
We can view this function as an infinite matrix with indices from $\Graphsu\times\Graphsu$ and 
entries from~$\mathbf N$.
The matrix indices are ordered according to some fixed total order on~$\Graphsu$ that respects the total size~$\abs{V(F)}+\abs{E(F)}$ of the graphs~$F\in\Graphsu$. Among graphs of the same total size, ties may be broken arbitarily.
If~$H$ and~$G$ are graphs such that~$H$ has larger total size than~$G$, then $H$ cannot be an induced subgraph of~$G$, that is,
$\indsub(H,G)=0$.
We conclude that $\indsub$ is an upper triangular matrix.

We define graph motif parameters as graph parameters that can be expressed as finite linear combinations of induced subgraph numbers. To obtain a clean formulation in terms of linear algebra, we represent these linear combinations as infinite vectors $\alpha\in\mathbb
Q^\Graphsu$ of finite support. Here, the \emph{support~$\supp\alpha$} of a vector $\alpha$ is the set of all graphs~$F\in\Graphsu$ with $\alpha_F\ne 0$.

\begin{defn}
  A graph parameter~$f:\Graphsu\to\mathbb Q$ is a
  \emph{graph motif parameter}
  if there is a vector~$\alpha\in\mathbb Q^\Graphsu$ with finite support such
  that $f(G) = \sum_{F\in\Graphsu} \alpha_F\cdot\indsub(F,G)$ holds for
  all~$G\in\Graphsu$.
\end{defn}
If we interpret $f$ and $\alpha$ as row vectors, this definition can also be phrased as requiring $f = \alpha\cdot\indsub$ for $\alpha$ of finite support.
Here, for two vectors $\alpha,\beta\in{\mathbf Q}^{\Graphsu}$, we
define the scalar product $(\alpha,\beta)$ as ${\sum_{F\in\Graphsu}
\alpha_F\cdot \beta_F}$ if this sum is defined.\footnote{In this paper, such scalar products degenerate into finite sums, since the support of at least one of the vectors will be finite.}
The definition of the matrix-vector and matrix-matrix products is then as usual.

The set of all graph motif parameters, endowed with the operations of scalar multiplication and pointwise addition, forms an infinite-dimensional vector space.
More specifically, it is the finitely supported row-span of the matrix~$\indsub$.
We remark that even if we drop the condition of~$\alpha$ being finitely supported,
the scalar product~$\alpha\cdot\indsub$ remains well-defined, since every column
of~$\indsub$ has finite support (as a consequence of every graph $G$ having only finitely many induced subgraphs $H$).
In fact, it can be verified that \emph{every} graph parameter~$f$ can be written as $\alpha\cdot\indsub$ for some~$\alpha$.

\subsection{Relations between graph motif parameters}
\label{sec: graph motif parameter relations}

One may wonder why we chose induced subgraph numbers for our definition of
graph motif parameters and not, say, the numbers of subgraphs or homomorphisms.
It turns out that all of these choices lead to the same vector space:
Subgraph and homomorphism numbers are graph motif parameters themselves, and indeed
they also span the space of graph motif parameters. 

Since some properties of graph motif parameters, such as their computational complexity, turn out to be easier to understand over the homomorphism basis, we show explicitly how to perform basis transformations. To this end, we first present the basis transformation between subgraphs and induced subgraphs, then we proceed with the basis transformation between subgraphs and homomorphisms.

\paragraph{Subgraphs and induced subgraphs.}
For graphs~$H,G\in\Graphsu$, we first show how to express $\sub(H,G)$ as a linear combination of numbers $\indsub(F,G)$.
To this end, note that every subgraph copy of $H$ in $G$ is contained in some induced subgraph $F$ of $G$ on $|V(H)|$ vertices. This induced subgraph $F$ is isomorphic to a supergraph of $H$, and we call these supergraphs $F$ \emph{extensions}. More precisely, an \emph{extension} of~$H$ is a (labeled) supergraph~$X$ of~$H$ with~$V(X)=V(H)$.

Note that~$H$ might have different extensions that are isomorphic.
Thus, given a graph $F$, let $\ext(H,F)$ be the number of extensions~$X$ of~$H$ that are isomorphic
to~$F$; equivalently, we have 
\[
\ext(H,F) = [\abs{V(H)}=\abs{V(F)}]\cdot\sub(H,F)\,,
\]
and we thus obtain
\begin{equation}\label{eq: sub ext indsub}
  \sub(H,G)
  =
  \sum_{F\in\Graphsu}
  \ext(H,F)\cdot\indsub(F,G)\,.
\end{equation}
Every graph~$H$ admits only finitely many extensions, and so the function
$\sub(H,\star)$ is a graph motif parameter for every fixed~$H$.
In matrix notation, the identity \eqref{eq: sub ext indsub} takes on the concise form 
\begin{equation}
\sub = \ext\cdot\indsub.
\end{equation}
Since $\sub$, $\ext$, and $\indsub$ are upper triangular matrices with diagonal
entries equal to~$1$, every finite principal submatrix of any of these
triangular matrices is invertible; indeed the entire matrix~$\ext$ has an
inverse with~$\indsub=\ext^{-1}\cdot\sub$.
This implies that $\sub$ also \emph{spans} the space of graph motif parameters: 
Every function $\indsub(H,\star)$ is a finite linear combination of functions $\sub(F,\star)$ with coefficients $\ext^{-1}(H,F)$.

We remark that the values of the coefficients $\ext^{-1}(H,F)$ are actually well understood:
The identity \eqref{eq: sub ext indsub} can be interpreted as a zeta transform over
the subset lattice \cite[eq.~(13) and~(14)]{borgs2006counting},
so we can perform Möbius inversion to prove that
\begin{equation}
\label{eq: ext inverse}
  \ext^{-1}(H,F) = (-1)^{\abs{E(F)} - \abs{E(H)}} \cdot \ext(H,F)
\end{equation}
holds for all graphs $H$ and $F$. We will use this identity later to check that $\ext^{-1}(H,F) \neq 0$ holds for specific pairs $(H,F)$ of graphs.

\paragraph{Homomorphisms and subgraphs.}
We wish to express $\hom(H,G)$ as a finitely supported linear combination~$\sum_{F}\alpha_F\sub(F,G)$ of subgraph numbers.
For a homomorphism $h$ from~$H$ to~$G$, let~$I$ be the image of $h$, that is, the graph with vertex set~$f(V(H))$ and edge set~$f(E(H))$; 
we observe that $h$ is a surjective homomorphism from~$H$ to~$I$ and~$I$ is a subgraph of~$G$.
That is, every homomorphism from $H$ to $G$ can be written as a surjective homomorphism into a subgraph $F$ of~$G$.
Writing $\surj(H,F)$ for the number of surjective homomorphisms from~$H$ to~$F$,
we have
\begin{equation}\label{eq: hom surj sub}%
  \hom(H,G)
  =
  \sum_{F\in\Graphsu}
  \surj(H,F)
  \cdot
  \sub(F,G)
  \,.
\end{equation}
Note that $\surj(H,F)=0$ holds if $H$ is smaller than~$F$ in total size.
Thus, analogously to the case of subgraphs, for each fixed~$H$, we have $\surj(H,F)\ne 0$ only for finitely
many~$F\in\Graphsu$.
Therefore $\hom(H,\star)$ is indeed a graph motif parameter for every
fixed~$H$.
In matrix notation, we have 
\begin{equation}\label{eq: hom surj sub matrix}%
\hom=\surj\cdot\sub\,,
\end{equation}
where $\surj$ is a lower triangular matrix.
Moreover, the diagonal entries of $\surj$ satisfy $\surj(F,F)=\aut(F)\neq 0$,
and hence each finite principal submatrix is invertible.
In fact the entire matrix has an
inverse $\surj^{-1}$ satisfying $\sub=\surj^{-1}\cdot\hom$, and so $\hom$ spans
the space of graph motif parameters as well.

The inverse of $\surj$ can be understood in terms of a Möbius inversion on a
partition lattice.
To see this, let us first consider the support of the vector~$\surj(H,\star)$,
that is, the set of all unlabeled graphs that are homomorphic images of~$H$.
This set will play an important role throughout this paper, and we call it the \emph{spasm} of~$H$:
\begin{equation}
  \spasm H
  =
  \setc[\big]{ F\in\Graphsu }{ \surj(H,F)>0 }
  \,.
\end{equation}
In a more graph-theoretical interpretation, the elements in the spasm of~$H$ can also be understood as the
unlabeled representatives of all graphs that can be obtained from~$H$ by merging
independent sets.
We make this more formal in the following definition.
For each~$H\in\Graphsl$, let $\partitions{H}$ be the set of all partitions
of~$V(H)$, where a partition is a set of disjoint non-empty subsets~$B\subseteq V(H)$ whose
union equals~$V(H)$.
\begin{defn}\label{def: consolidation}
  For a graph $H\in\Graphsl$ and a partition $\rho\in\partitions{H}$, the
  \emph{quotient $\contract H\rho$} is the graph obtained by
  identifying, for each block $B\in\rho$, the vertices in $B$ to a single
  vertex.
  This process may create loops or parallel edges; 
  we keep loops intact in~$\contract H\rho$, and we turn parallel edges into simple edges.
\end{defn}

For $F\in\Graphsu$, we have $F\in\spasm H$ if and
only if there is a partition~$\rho\in\partitions{H}$ with~$F\simeq\contract H\rho$.
Note that graphs $F\in\Graphsu$ does not have loops, since we explicitly restricted the graphs in $\Graphsu$ to be simple. Consequently, $F\simeq\contract H\rho$ can only hold if all blocks of~$\rho$ are independent
sets of~$H$, that is, if~$\rho$ represents a proper vertex-coloring of~$H$.

Every surjective homomorphism from~$H$ to~$F$ can be interpreted as a pair~$(\rho,\pi)$ where $\contract H\rho\simeq F$ and~${\pi\in\Aut{F}}$.
Hence we have
\begin{align}
  \label{eq: surj}
  \surj(H,F)
  &=
  \#\Aut{F}\cdot
  \sum_{\rho\in\partitions H}
  [ \contract H\rho \simeq F ] 
  \,,
\end{align}

For two partitions~$\rho,\rho'\in\partitions{H}$, we write $\rho\ge\rho'$ if
$\rho$ is coarser than $\rho'$, that is, if every block of $\rho'$ is contained in
a block of~$\rho$.
This partial order gives rise to the \emph{partition lattice
$(\partitions{H},\geq)$} whose minimal element~$\bot$ is the finest partition, i.e., the partition whose blocks all have size one.
Now \eqref{eq: hom surj sub} can be viewed as a zeta-transformation on the
partition lattice:
Let $H$ and $G$ be fixed graphs.
Let $f(\rho)=\#\Emb{\contract H\rho}{G}$.
Then consider its \emph{upwards zeta-transform} on the partition lattice, i.e., the function $\hat f$ defined by
\begin{align*}
  \hat f(\rho)
  =
  \sum_{\rho'\ge\rho} f(\rho')
  \,.
\end{align*}
We observe that $\hom(H,G)=\hat f(\bot)$.
By Möbius inversion, we get (see also~\cite[eq.~(15)]{borgs2006counting}):
\begin{align*}
  f(\rho)
  =
  \sum_{\rho'\ge\rho}
  (-1)^{|\rho|-|\rho'|}
  \cdot
  \paren[\Big]{
  \prod_{B\in\rho'}\paren{\lambda(\rho,\rho',B)-1}!
  }\cdot
  \hat f(\rho')
  \,,
\end{align*}
where $\lambda(\rho,\rho',B)$ is the number of blocks $C\in\rho$ with
$C\subseteq B$.
We set $\rho=\bot$ and collect terms~$\rho'$ that lead to isomorphic graphs
$\contract{H}{\rho'}$. Note that, for a given graph isomorphism type, all terms leading to this type are non-zero and have the same sign. We obtain
\begin{align}
  \label{eq: surj inverse}
  \surj^{-1}(H,F)
  &=
  \frac{
  (-1)^{\abs{V(H)}-\abs{V(F)}}
  }{\#\Aut{H}}
  \cdot
  \sum_{%
  \substack{%
  \rho\in\partitions H\\
  \contract H\rho \simeq F
  }}
  \prod_{B\in\rho}\paren{\abs{B}-1}!
  \,.
\end{align}
In particular, this yields $\surj^{-1}(H,F)\ne 0$ if and only ${\surj(H,F)\ne 0}$;
that is, $F\in \spasm H$ is equivalent to $\surj^{-1}(H,F)\ne 0$.
This observation will be crucial in the proof of our hardness result.

While we established before that~$\hom$ spans the space of graph motif parameters, it is not immediately clear
that the homomorphism numbers form a \emph{basis}, that is, that the rows of
$\hom$ are linearly independent.
The following proposition on the invertibility of certain principal submatrices
will be important for our hardness results, and it implies that the rows of
$\hom$ are linearly independent with respect to finite linear combinations.

\begin{lem}[Proposition~5.43 in \cite{lovaszbook}]
  \label{lem: hom is invertible}
  Let $S\subseteq\Graphsu$ be a finite set of graphs that is closed under
  surjective homomorphisms, that is, we have $\spasm H\subseteq S$ for all $H\in
  S$.
  Then the principal submatrix $\hom_S$ of $\hom$ is invertible and satisfies
  $\hom_S=\surj_S\cdot\sub_S$.
\end{lem}
\begin{proof}
  Let $F,G\in S$ and consider the expansion of~$\hom(H,G)$ from~\eqref{eq: hom surj sub}.
  Since $S$ is closed under surjective homomorphisms, only terms with $F\in S$
  contribute to the sum.
  Hence we have $\hom_S=\surj_S\cdot\sub_S$.
  Since $\surj_S$ and $\sub_S$ both are triangular matrices with non-zero diagonal entries, they are invertible, and consequently so is their product~$\hom_S$.
\end{proof}

\def\gz{\textcolor{gray}0}
\begin{figure}[tp]
  \centering
    \begin{equation*}
      \begin{tikzpicture}[baseline=-\the\dimexpr\fontdimen22\textfont2\relax]
        \matrix (m) [matrix of math nodes,left delimiter=(,right delimiter=),inner
        sep=3pt]
        {
        2 & 4 & 6 & 6 \\
        2 & 6 &12 &10 \\
        0 & 0 & 6 & 0 \\
        2 & 8 &24 &16 \\
        };
        \node[shift=(m-1-1.north),yshift=0.6cm,rotate=90](0,0) { \drawgraph{A_} };
        \node[shift=(m-1-2.north),yshift=0.6cm,rotate=90](0,0) { \drawgraph{BW} };
        \node[shift=(m-1-3.north),yshift=0.6cm,rotate=90](0,0) { \drawgraph{Bw} };
        \node[shift=(m-1-4.north),yshift=0.6cm,rotate=90](0,0) { \drawgraph{CR} };
        \node[shift=(m-1-1.west),xshift=-1.0cm](0,0) { \drawgraph{A_} };
        \node[shift=(m-2-1.west),xshift=-1.0cm](0,0) { \drawgraph{BW} };
        \node[shift=(m-3-1.west),xshift=-1.0cm](0,0) { \drawgraph{Bw} };
        \node[shift=(m-4-1.west),xshift=-1.0cm](0,0) { \drawgraph{CR} };
      \end{tikzpicture}
      =
      \begin{tikzpicture}[baseline=-\the\dimexpr\fontdimen22\textfont2\relax ]
        \matrix (m) [matrix of math nodes,left delimiter=(,right delimiter=),inner
        sep=3.5pt]
        {
        2 &\gz&\gz&\gz\\
        2 & 2 &\gz&\gz\\
        0 & 0 & 6 &\gz\\
        2 & 4 & 6 & 2 \\
        };
        \node[shift=(m-1-1.north),yshift=0.6cm,rotate=90](0,0) { \drawgraph{A_} };
        \node[shift=(m-1-2.north),yshift=0.6cm,rotate=90](0,0) { \drawgraph{BW} };
        \node[shift=(m-1-3.north),yshift=0.6cm,rotate=90](0,0) { \drawgraph{Bw} };
        \node[shift=(m-1-4.north),yshift=0.6cm,rotate=90](0,0) { \drawgraph{CR} };
      \end{tikzpicture}
      \cdot
      \begin{tikzpicture}[baseline=-\the\dimexpr\fontdimen22\textfont2\relax ]
        \matrix (m) [matrix of math nodes,left delimiter=(,right delimiter=),inner
        sep=3.5pt]
        {
        1  & 2 & 3 & 3 \\
        \gz& 1 & 3 & 2 \\
        \gz&\gz& 1 & 0 \\
        \gz&\gz&\gz& 1 \\
        };
        \node[shift=(m-1-1.north),yshift=0.6cm,rotate=90](0,0) { \drawgraph{A_} };
        \node[shift=(m-1-2.north),yshift=0.6cm,rotate=90](0,0) { \drawgraph{BW} };
        \node[shift=(m-1-3.north),yshift=0.6cm,rotate=90](0,0) { \drawgraph{Bw} };
        \node[shift=(m-1-4.north),yshift=0.6cm,rotate=90](0,0) { \drawgraph{CR} };
      \end{tikzpicture}
    \end{equation*}
  \caption[Example for the matrix identity $\hom_S=\surj_S\cdot\sub_S$]{%
    \label{fig: hom surj sub example}%
    The matrix identity $\hom_S=\surj_S\cdot\sub_S$,
  where $S$ is the spasm of the path $\drawgraph{CR}$, consisting of the four graphs $\drawgraph{A_}$, $\drawgraph{BW}$, $\drawgraph{Bw}$, $\drawgraph{CR}$.}
\end{figure}

See Figure~\ref{fig: hom surj sub example} for an example of Lemma~\ref{lem: hom
is invertible}.
Let us also note three simple and useful properties that $\spasm
H$ inherits from $H$.
\begin{fact}
\label{fact: spasm}
For all graphs $H$, the following properties hold:
\begin{enumerate}
  \item Every graph $F\in\spasm H$ has at most $\abs{V(H)}$ vertices and
    at most $\abs{E(H)}$ edges.
  \item If $H$ has a vertex-cover of size $b\in\mathbb{N}$, then every graph $F\in\spasm 
    H$ has a vertex-cover of size $b$.
  \item
    If $H$ contains a matching with $k\in\mathbb{N}$ edges, then
    every graph with~$k$ edges (and no isolated vertices) can be found as a
    minor of some graph in~$\spasm{H}$.
\end{enumerate}
\end{fact}
\begin{proof}
  Only the third claim merits some explanation.
  If $M_k$ is the (not necessarily induced) $k$-matching in~$H$ and~$F$ is the
  $k$-edge graph we want to find, then we determine an arbitrary surjective
  homomorphism $g:V(M_k)\to V(F)$, which hits every edge of $F$.
  We are allowed to contract edges (due to the minor operation) or consolidate
  non-edges of~$H$ (due to the quotient operation) to build~$F$.
  To this end, we simply identify all vertices of~$g^{-1}(i)$, for each~$i\in
  V(F)$, and delete all vertices in~$V(H)\setminus V(M_k)$.
\end{proof}

\paragraph{Embeddings and strong embeddings.}
For completeness, we define matrices for embeddings and strong embeddings.
For ${H,G\in\Graphsu}$, let $\emb(H,G)$, $\stremb(H,G)$, and $\iso(H,G)$ be the
number of embeddings, strong embeddings, and isomorphisms from~$H$ to~$G$,
respectively.
Clearly $\iso$ is a diagonal matrix with $\iso(F,F)=\aut(F)$.
We have $\emb=\iso\cdot\sub$ and
$\stremb=\iso\cdot\indsub$.

\subsection{The complexity of graph motif parameters}
\label{sec: graph motif parameter complexity}

Several computational problems can be associated with graph motif parameters,
but perhaps the most natural one is the \emph{evaluation problem}: 
Given as input a graph motif parameter~$f:\Graphsu\to\mathbb Q$ and a graph~$G\in\Graphsu$, 
compute the value~$f(G)$.

This problem requires a suitable representation of the input~$f$,
and while we could choose any basis to represent $f$, the homomorphism basis turns out to be particularly useful for algorithmic purposes.
That is, in this subsection, we represent graph motif parameters $f$ as vector-matrix products $f=\alpha\cdot\hom$ for finitely supported row vectors $\alpha\in\mathbb{Q}^\Graphsu$.
The input is then the coefficient vector $\alpha$, encoded as a list of pairs~$(F,\alpha_F)$ for~$F\in\supp\alpha$.
Let~$\abs{\alpha}$ be the description length of $\alpha$, and let $\tw\alpha$ be
the maximum treewidth~$\tw F$ among all graphs~$F\in\supp\alpha$.
The following lemma is immediate:

\begin{lem}[Algorithm]\label{lem: graph motif algorithm}
  There is a deterministic algorithm that is given $\alpha$ and $G$ to compute
  $(\alpha\cdot\hom)(G)$
  in time
  $g(\alpha) + \poly{\abs{\alpha}}\cdot \abs{V(G)}^{\tw\alpha+1}$
  for some computable function $g$ depending only on $\alpha$.
\end{lem}
\begin{proof}
  For each $F\in\supp\alpha$, run the algorithm from
  Proposition~\ref{prop: hom-treewidth-algo} to compute $\hom(F,G)$
  in time $\exp(O(k))+\poly{k} \cdot n^{\tw{F}+1}$, where $k=\abs{V(F)}$ and $n=\abs{V(G)}$.
  Then output $\sum_F \alpha_F \cdot \hom(F,G)$.
\end{proof}

We could choose other representations for~$f$,
such as coefficient vectors $\alpha$ with $f=\alpha\cdot\sub$ or $f=\alpha\cdot\indsub$.
Switching between these representations only adds an overhead of
$g(\alpha)$ in the running time, for some function~$g$, as can be seen from~\S\ref{sec: graph motif
parameter relations}.

It is clear that the generic evaluation problem for graph motif parameters is $\sharpWone$-hard,
since it subsumes counting $k$-cliques as a special case.
The following reduction shows that evaluating~$f$ with $f=\alpha\cdot\hom$ is at least as hard as every individual homomorphism problem~$\hom(F,\star)$ for ~$F\in\supp\alpha$.
That is, if a linear combination of homomorphisms contains a ``hard'' pattern graph, then the entire linear combination is ``hard''.

\begin{lem}[Extracting summands]\label{lem: extraction reduction}
  There is a deterministic Turing reduction that is given a finitely supported vector~$\alpha\in\mathbb Q^{\Graphsu}$, a graph~$F\in\supp{\alpha}$, and a graph~$G\in\Graphsu$ to compute the number~$\hom(F,G)$ with an oracle for the function $(\alpha\cdot\hom)(\star)$.
  The reduction runs in time $g(\alpha)\cdot\poly{\abs{V(G)}}$ for some
  computable function~$g$, makes at most $g(\alpha)$ queries to $(\alpha\cdot\hom)(\star)$, and each queried graph has at most $\max_{H\in\supp\alpha}\abs{V(H)}\cdot\abs{V(G)}$ vertices.
\end{lem}
\begin{proof}
  On input $(\alpha,F,G)$, the reduction only makes queries of the form $(\alpha\cdot\hom)(G \times X)$ for graphs $X$, where $G\times X$ is the categorical product, that is, the graph with vertex set $V(G)\times V(X)$ such that $(v,x)$ and $(v',x')$ are adjacent in $G\times X$ if and only if $vv'\in E(G)$ and $xx'\in E(X)$.
  The following holds~\cite[(5.30)]{lovaszbook}:
  \begin{equation}\label{eq: hom tensor product}
    \hom(F,G\times X)
    =
    \hom(F,G)
    \cdot
    \hom(F,X)
    \,.
  \end{equation}
  Using this identity for various $X$, we aim at setting up a linear equation system that can be solved uniquely for~$\hom(F,G)$.
  We expand the sum~$(\alpha\cdot\hom)(G\times X)$ and
  apply~\eqref{eq: hom tensor product} to obtain
  \begin{equation}\label{eq: hardness full equation}
    \sum_{H} \alpha_H \cdot \hom(H,G)\cdot\hom(H,X)
    =
    (\alpha\cdot\hom)(G\times X)
    \,.
  \end{equation}
  For each graph $X$, the reduction can compute the right side of this linear equation using the oracle, and it can determine the numbers~$\alpha_H$ and $\hom(H,X)$ in some time $f(\alpha)$.
  It remains to choose a suitable set~$S$ of graphs~$X$ so that the resulting system of linear equations can be uniquely solved for~$\hom(F,G)$.

  Let $S=\bigcup_{H\in\supp\alpha}\spasm H$ be the closure of $\supp\alpha$ under spasms.
  By Lemma~\ref{lem: hom is invertible}, the matrix~$\hom_S$ is invertible.
  Moreover, we have $\alpha\cdot\hom=\alpha_S\cdot\hom_S$.
  We rewrite~\eqref{eq: hardness full equation} as $\hom_S\cdot x=b$ where
  $b\in\mathbb{Q}^S$ and $\hom_S\in\mathbb{Q}^{S\times S}$ represent the known
  quantities with $b_X=(\alpha\cdot\hom)(G\times X)$ and
  $\hom_S(H,X)=\hom(H,X)$, and the vector~$x\in\mathbb{Q}^S$ represents the
  indeterminates with $x_H=\alpha_H\cdot\hom(H,G)$.
  We have $x=(\hom_S)^{-1}\cdot b$, so we can solve uniquely for the indeterminates.
  In particular, we can compute $\hom(F,G)=x_F/\alpha_F$, since~$\alpha_F\ne 0$ holds by assumption.

  The set~$S$ and the matrices $\hom_S$ and $(\hom_S)^{-1}$ can be computed in time $g(\alpha)$ for some computable function~$g$.
  For the vector~$b$, we need to compute the product graphs and query the
  oracle.
  The number of queries is~$\abs{S}$ and thus bounded by~$g(\alpha)$.
  Overall, the reduction takes time $g(\alpha)\cdot\poly{\abs{V(G)}}$.
\end{proof}
Analogously to the problems~$\SubProb{\mathcal H}$ in \S\ref{sec: intro
subgraphs}, we consider restricted classes of graph motif parameters (represented as linear combinations of homomorphism numbers) to obtain more
expressive hardness results.
\begin{defn}\label{def:problems}
  Let $\mathcal A\subseteq\mathbb{Q}^\Graphsu$ be a set of finitely supported vectors.
  We let $\HomProb{\mathcal{A}}$ be the computational problem whose task is to
  compute $(\alpha\cdot\hom)(G)$ on input $\alpha\in\mathcal{A}$ and
  $G\in\Graphsu$.
\end{defn}
We apply Lemma~\ref{lem: extraction reduction} to establish the hard cases
of~$\HomProb{\mathcal{A}}$.
\begin{lem}[Hardness]\label{lem: graph motif hardness}
  Let $\mathcal A\subseteq\mathbb Q^\Graphsu$ be a recursively enumerable class
  of finitely supported vectors.
  If~$\mathcal A$ contains vectors of arbitrarily large treewidth $\tw{\alpha}$,
  then $\HomProb{\mathcal A}$ is $\sharpWone$-hard when parameterized by~$\abs{\alpha}$.
  Moreover, the problem does not have
  $g(\alpha)\cdot\abs{V(G)}^{o(\tw{\alpha}/\log\tw{\alpha})}$ time algorithms if
  $\sharpETH$ holds.
\end{lem}
\begin{proof}
  For each $\alpha\in\mathcal A$, we select a graph~$F_\alpha\in\supp{\alpha}$
  of maximum treewidth.
  Then the set $\mathcal F$ with ${\mathcal F=\setc{F_\alpha}{\alpha\in \mathcal A}}$ has unbounded treewidth.
  To prove the hardness, we provide a parameterized Turing reduction from $\HomProb{\mathcal F}$ to $\HomProb{\mathcal A}$.
  Since~$\mathcal F$ has unbounded treewidth, Theorem~\ref{thm: dalmau jonsson} implies that $\HomProb{\mathcal F}$ is $\sharpWone$-hard when parameterized by $\abs{V(F)}$ and Proposition~\ref{prop: hom ETH hard} implies it cannot be computed in time $g(F)\cdot\abs{V(G)}^{o(\tw F/\log\tw F)}$ for any computable~$g$ if $\sharpETH$ holds.

  Let $(F,G)$ with $F\in\mathcal F$ be an input for the reduction, whose goal is
  to compute $\hom(F,G)$ with oracle access to $\HomProb{\mathcal A}$.
  First the reduction computes some $\alpha\in\mathcal A$ with $F_\alpha=F$.
  This is possible, since $\mathcal A$ is recursively enumerable.
  The reduction then applies the algorithm from Lemma~\ref{lem: extraction
  reduction} on input $(\alpha,F,G)$.
  The algorithm runs in time $g(\alpha)\cdot\poly{\abs{V(G)}}$ for a computable function~$g$ and makes queries to the function~$(\alpha\cdot\hom)(\star)$; its output is the desired number $\hom(F,G)$.

  For the running time of the reduction, note that finding~$\alpha$ takes time $h(k)$ for some computable function~$h$, where $k=\abs{V(F)}$.
  The algorithm from Lemma~\ref{lem: extraction reduction} runs in some time
  $g(\alpha)\cdot\poly{n}\le h(k)\cdot\poly{n}$.
  Hence we indeed obtain a parameterized reduction from
  $\HomProb{\mathcal F}$ to $\HomProb{\mathcal A}$, which proves that
  $\HomProb{\mathcal A}$ is $\sharpWone$-hard when parameterized
  by~$\abs{\alpha}$.
  Finally, we have $\tw{\alpha}=\tw F$, so if $\HomProb{\mathcal A}$ can be computed in time $f(\alpha)\cdot
  n^{o(\tw{\alpha}/\log\tw\alpha)}$, then $\HomProb{\mathcal F}$ can be solved
  in time $f(F)\cdot n^{o(\tw F / \log\tw F)}$, which by Proposition~\ref{prop:
  hom ETH hard} is impossible if $\sharpETH$ holds.
\end{proof}

From the perspective of fine-grained complexity, 
for every \emph{fixed} graph motif parameter~$f$, 
Lemma~\ref{lem: graph motif algorithm} yields an algorithm for evaluating $f$ on $n$-vertex graphs in some time~$O(n^{c})$.
Here, $c$ is a constant that depends on the largest treewidth in the homomorphism representation of $f$.
To express this connection more precisely, 
given a graph motif parameter $f$,
we define the constant
\begin{align}
  C(f)
&=\inf\setc{c\in\mathbb R}{\text{$f$ can be computed in time $O(n^c)$ }}
\,.
\end{align}
In particular, we can consider this constant for the graph parameters $\hom(F,\star)$ for fixed graphs $F$.
Then the constant $C(\hom(F,\star))$ is the smallest possible exponent required for computing~$\hom(F,G)$ on $n$-vertex graphs~$G$.
The proof of Lemma~\ref{lem: graph motif algorithm} implies
\[
C(f)\le\max_{F\in\supp{\alpha}} C({\hom(F,\star)})\,,
\] 
where $\alpha\in\mathbb Q^\Graphsu$ is the representation of $f$ over the homomorphism basis, that is, the vector $\alpha$ with $f=\alpha\cdot\hom$.
Lemma~\ref{lem: extraction reduction} implies the corresponding
lower bound, so in fact we have
\begin{equation}
C(f) = \max_{F\in\supp{\alpha}} C({\hom(F,\star)})\,,
\end{equation}

Proposition~\ref{prop: hom ETH hard} implies that $C({\hom(F,\star)})\le\tw{F}+1$ holds.
Under the assumption $\FPT\ne\sharpWone$, Theorem~\ref{thm: dalmau jonsson} implies that 
$C({\hom(F,\star)})$ cannot be bounded by a universal constant, and under the stronger assumption~$\sharpETH$, Proposition~\ref{prop: hom ETH
hard} implies that $C({\hom(F,\star)})$ is not bounded by $o(\tw F / \log\tw F)$, even when~$F$ is
restricted to be from any fixed family~$\mathcal F$ of graphs of unbounded treewidth.
The $k$-clique hypothesis is that
the current fastest algorithm for $k$-clique is optimal \cite{DBLP:conf/focs/AbboudBW15a}, which can be formalized as
$C({\hom(K_k,\star)})=\omega k /3$.
These facts suggest that the representation of a graph motif parameter~$f$ in the homomorphism basis
and your favorite complexity hypotheses are all that is needed to
understand the complexity of~$f$ (concerning the exponent of $n$).

\begin{rem}
  \label{rem: colored graphs}
  Lovász's framework is easily adapted to the setting of vertex-colored pattern and
  host graphs.
  Here, we only wish to count homomorphisms (or embeddings, or strong embeddings) from~$H$ into~$G$
  that map vertices of~$H$ to vertices of~$G$ of the same colors.
  In the definition of the spasm of~$H$, one is then only allowed to identify
  non-adjacent vertices of the same color, that is, the allowed partitions
  $\rho$ are those where each block is a monochromatic independent set.
  The algorithm (Lemma~\ref{lem: graph motif algorithm}) and hardness result
  (Lemma~\ref{lem: graph motif hardness}) can also be adapted to the setting of
  vertex-colored homomorphisms without modifications.
  We excluded these variants from the main text to simplify the presentation.
\end{rem}

\section{Algorithms for counting subgraphs}
\label{sec: algo}

In this section, we obtain algorithms for counting subgraphs and embeddings
by expressing these quantities over the homomorphism basis via~\eqref{eq: hom surj sub}
and running an algorithm for counting homomorphisms.
More concretely, recall that
\begin{equation}
  \label{eq: sub surj hom}
  \sub(H,G)
  =
  \sum_F \surj^{-1}(H,F)\cdot\hom(F,G)
  \,.
\end{equation}
We know from~\eqref{eq: surj inverse} that~$\surj^{-1}(H,F)\ne 0$ is equivalent to
$F\in\spasm{H}$.
Hence for fixed patterns~$H$, if we can compute the homomorphism numbers
$\hom(F,G)$ for all $F\in\spasm{H}$ in time~$O(n^c)$ on $n$-vertex graphs $G$, then we can compute
$\sub(H,G)$ in time~$O(n^c)$.
With the basic treewidth-based algorithm from Proposition~\ref{prop:
hom-treewidth-algo}, this running time $O(n^c)$ is governed by the maximum treewidth among~$\spasm H$. 
\begin{reptheorem}{thm: algorithm count subgraphs}
  There is an algorithm that is given a $k$-edge graph~$H$ and
  an~$n$-vertex graph~$G$ to compute $\#\Sub H G$ in time
  $k^{O(k)} n^{t+1}$, where
  $t$ is the maximum treewidth among all graphs in~$\spasm H$.
\end{reptheorem}
\begin{proof}
  Use~\eqref{eq: sub surj hom} and evaluate the right-hand side in the straightforward manner:
  First compute the spasm of $H$ and all coefficients $\surj^{-1}(H,F)$ for
  $F\in\spasm H$ by first computing~$\surj_{\spasm{H}}$ via brute-force
  and then inverting this triangular matrix.
  Then, for each $F\in\spasm H$, compute $\#\Hom FG$ via
  Proposition~\ref{prop: hom-treewidth-algo}.
  Finally, we compute the sum on the right side of~\eqref{eq: sub surj hom} to
  obtain $\#\Sub HG$.

  For the running time claim, note that the size of $\spasm
  H$ is bounded by the number of partitions of the set~$V(H)$,
  which in turn can be bounded crudely by $k^{O(k)}$.
  Each term $\#\Hom FG$ can be computed in time $\exp({O(k)}) \cdot n^{t+1}$ by Proposition~\ref{prop: hom-treewidth-algo}.
\end{proof}
Theorem~\ref{thm: algorithm count subgraphs} is particularly useful for sparse patterns
$H$: By Fact~\ref{fact: spasm}, any $k$-edge graph $H$
only contains graphs with at most $k$ edges in its spasm.
We can thus exploit known bounds on the treewidth of sparse
graphs to bound the running time guaranteed by Theorem~\ref{thm: algorithm count
subgraphs}:
If $F$ has $k$ edges, then $\tw F\leq ck+o(k)$ is known to hold for fairly
small constants~${c<1}$.
Furthermore, there are $k$-edge graphs $F$ with $\tw F = \Theta(k)$.
Since the best known upper and lower bounds on the treewidth of $k$-edge graphs
are not tight, it will be useful to dedicate a universal constant~$\sparseTWconstant$ to the linear coefficient in the treewidth bound.
\begin{defn}
\label{def: sparseTWconstant}
  For $k\in\mathbb{N}$, let $\mathrm{tw}^*(k)$ be defined as the maximum treewidth among all graphs 
  with~$k$ edges.
  We define the constant $\sparseTWconstant\in\mathbb{R}$ as the limit superior
  \[
    \sparseTWconstant
    =
    \limsup_{k\to\infty}
    \frac{\mathrm{tw}^*(k)}{k}
    \,.
  \]
\end{defn}
Using the existence of good $3$-regular expanders, Dvořák and Norin 
\cite[Corollary 7]{dvorak2016strongly} find a family of $3$-regular graphs on $k$ edges and
$n=\frac{2}{3}k$ vertices with treewidth at least $\frac{1}{24} n-1=\frac{1}{36} 
k - 1$ for all large enough~$k$.
On the other hand, Scott and Sorkin~\cite[Corollary~21]{Scott2007260} prove that 
$\mathrm{tw}^*(k)\leq \frac{13}{75}k+o(k)$ holds.
Collecting these two results, we get the following bounds on $\sparseTWconstant$.

\begin{thm}[\cite{dvorak2016strongly,Scott2007260}]\label{thm: sparse tw constant}
  We have $\frac{1}{36}\le\sparseTWconstant\leq\frac{13}{75}$.
\end{thm}
Since $\sparseTWconstant\leq\frac{13}{75}< 0.174$, 
this immediately implies upper bounds on the running times obtained from
Theorem~\ref{thm: algorithm count subgraphs}.
\begin{repcorollary}{cor: algorithm count subgraphs}
  There is an algorithm that is given~$H$ and~$G$ to compute
  $\#\Sub H G$ in time $f(H)\cdot \abs{V(G)}^{\xi k+o(k)}$,
  where~$k=\abs{E(H)}$. Here, $\sparseTWconstant< 0.174$ is the constant from Definition~\ref{def: sparseTWconstant}.
\end{repcorollary}

Instead of relying on Proposition~\ref{prop: hom-treewidth-algo} to count
homomorphisms in Theorem~\ref{thm: algorithm count subgraphs}, we can use more
sophisticated methods where available.
For instance, we can use fast matrix multiplication to count homomorphisms
from patterns~$H$ of treewidth at most two, thus proving Theorem~\ref{thm: algorithm count subgraphs matrixmult}.
\begin{reptheorem}{thm: algorithm count subgraphs matrixmult}
  If $H$ has treewidth at most $2$, we can determine $\#\Hom H{G}$ in time
  $\poly{\abs{V(H)}}\cdot |V(G)|^\omega$, where $\omega<2.373$ is the matrix
  multiplication constant and $f$ is some computable function.
\end{reptheorem}
\begin{proof}
  Let $H$ be a graph and $(T,\beta)$ be a tree decomposition of width at most~$2$.
  We can compute an optimal tree decomposition in time $\poly{\abs{V(H)}}$, for
  example via Bodlaender's $\exp(O({\tw{H}}^3))\cdot\abs{V(H)}$ time
  algorithm~\cite{bodlaender1996linear}.
  It will simplify notation if we assume the decomposition to have the following
  properties, which can be achieved by easy modifications:
  \begin{enumerate}
    \item The root bag $r$ has size $2$ and every other bag has size exactly $3$.
    \item The root bag $r$ has a unique child $r'$.
    \item For every $t\in V(T)\setminus\set{r}$, we have $|\sep(t)|=2$.
  \end{enumerate}
  Let us sketch how these properties can be achieved.
  If $\bag(t_1)\subseteq \bag(t_2)$ holds for two adjacent nodes $t_1,t_2\in V(T)$,
  then the two nodes can be merged.
  If $|\bag(t)|<3$ and $t$ has a neighbor $t'$ in $T$ with $\bag(t')\not\subseteq
  \bag(t)$, then the size of $\bag(t)$ can be increased by adding an element of
  $\bag(t')\setminus \bag(t)$ to it.
  After applying these two rules exhaustively, every bag has size exactly~$3$.
  If two adjacent bags $\bag(t_1)=\{a,b,c\}$ and
  $\bag(t_2)=\{c,d,e\}$ intersect only in one element $c$, then we can
  insert a new bag $\{b,c,d\}$ between them.
  Similarly, if $\bag(t_1)=\{a,b,c\}$ and $\bag(t_2)=\{d,e,f\}$ are disjoint,
  then we insert two new bags $\{b,c,d\}$ and $\{c,d,e\}$ between them.
  This way, we achieve property 3 above.
  Finally, we take any bag $\bag(r')=\{a,b,c\}$, attach a new bag
  $\bag(r)=\{a,b\}$ to it, and make $r$ the root of the tree.

  We may assume that $V(H)=\set{1,\dots,|V(H)|}$, that is, the vertices are
  represented by integers from~$1$ to~$|V(H)|$.
  We may further assume that the numbering is consistent with the structure
  of the tree decomposition.
  More precisely, we assume that the following condition holds for all $t,t'\in
  V(T)$ and $u,u'\in V(H)$ with $u\in\bag(t)$ and $u'\in\bag(t')$:
  If $t$ is an ancestor of $t'$ in the tree~$T$, the vertex~$u$ does not occur
  in the bag of any ancestor of~$t$, and $u'$ does not occur in the bag of any
  ancestor of~$t$, then the numbering reflects this fact in that $u<u'$ holds.

  For any $t\in V(H) \setminus\set{r}$, we write
  $u_1(t),u_2(t),u_3(t)\in\bag(t)$ for the three elements of the bag at~$t$,
  chosen in such a way that $u_1(t)<u_2(t)<u_3(t)$ holds.
  By the assumption on the numbering, we have $\sep(t)=\set{u_1(t),u_2(t)}$
  since the two vertices in $\sep(t)$ are have their topmost bag above~$t$ while
  the topmost bag of the vertex in~$\bag(t)\setminus\sep(t)$ is at~$t$.

  For any $t\in V(H) \setminus\set{r}$ and $v_1,v_2\in V(G)$, let  
  $h_t(v_1,v_2)$ be the number of homomorphisms from $\Hom {H[\cone(t)]} G$ that 
  map $u_1(t)$ to $v_1$ and $u_2(t)$ to $v_2$.
  By summing over all $v_1$ and $v_2$, we obtain the number of all homomorphisms 
  from the cone at~$t$, that is, we have
  \begin{equation*}
    \#\Hom {H[\cone(t)]}{G}=\sum_{v_1,v_2\in V(G)}h_{t}(v_1,v_2)
    \,.
  \end{equation*}
  As $H[\cone(r')]=H$ holds, the function~$h_{r'}$ at the root~$r'$ of the tree 
  decomposition can be thus used to determine the quantity $\#\Hom HG$ that we 
  wish to compute.
  With random access to the function table of~$h_{r'}$, we can evaluate the sum 
  in $O(\abs{V(G)}^2)$ arithmetic operations.
  
  We compute the function tables~$h_t$ in a bottom-up fashion.
  At the leaves~$t$ of $T$, we compute~$h_t$ in constant time using brute force.
  Now suppose that~$t$ is a non-leaf vertex of~$T$ and that the function table 
  of~$h_{t'}$ has already been computed for every child $t'$ of $t$.
  For each $ij\in\set{12,13,23}$, let $C_{ij}$ be the set of children~$t'$ of~$t$ that satisfy $\sep(t')=\{u_{i}(t),u_{j}(t)\}$.
  Note that $C_{12} \dotcup C_{13} \dotcup C_{23}$ is the set of all children 
  of~$t$.
  The following identity holds:
  \begin{equation}\label{eq: simple DP}
    h_t(v_1,v_2)
    =
    \sum_{v_3\in V(G)}
    I_{t,v_1,v_2,v_3}
    \cdot
    \prod_{t'\in C_{12}}
    h_{t'} (v_1,v_2)
    \cdot
    \prod_{t'\in C_{13}}
    h_{t'} (v_1,v_3)
    \cdot
    \prod_{t'\in C_{23}}
    h_{t'} (v_2,v_3)\,,
  \end{equation}
  where $I_{t,v_1,v_2,v_3}\in\set{0,1}$ indicates whether the mapping 
  $\bag(t)\to\set{v_1,v_2,v_3}$ that we want to fix here is a homomorphism 
  in~$\Hom{H[\bag(t)]}{G}$.
  The identity in \eqref{eq: simple DP} should be read as follows:
  In order to count the number of homomorphisms from~$\cone(t)$ that 
  fix~$u_1(t)$ to $v_1$ and $u_2(t)$ to~$v_2$, we first sum over all possible 
  values~$v_3$ that $u_3(t)$ might map to and count those homomorphisms 
  from~$\cone(t)$ that have all three values and thus all function values of 
  $\bag(t)$ fixed to $v_1$, $v_2$, and $v_3$, respectively.
  In order to count the latter, we discard with the factor $I_{t,v_1,v_2,v_3}$ 
  those fixings that violate the homomorphism property locally.
  Finally, we observe that $\cone(t'_1)\cap\cone(t'_2) = \sep(t'_1) = 
  \set{u_i(t),u_j(t)}$ holds for any two distinct children~$t'_1,t'_2\in 
  C_{ij}$.
  Hence, after fixing the values for the homomorphism on $\bag(t)$, the 
  extensions to the cones $\cone(t')$ are independent for different 
  children~$t'$, and so the total number of extensions is the product of the 
  $h_{t'}$.

  Using the identity \eqref{eq: simple DP} naively to perform the computation at 
  each bag would result in an overall running time of roughly $\abs{V(G)}^3$.
  Instead, we want to use matrix multiplication.
  We define three functions $a_{12},a_{13},a_{23}:V(G)\times V(G) \to \mathbb N$ 
  which we later interpret as matrices whose indices range over~$V(G)$.
  For each $v,v'\in V(G)$, let
  \begin{equation}
    a_{ij}(v,v')
    =
    \big[\text{$u_iu_j\in E(H)$ implies $v v'\in E(G)$}\big]
    \cdot
    \prod_{t'\in C_{ij}} h_{t'}(v,v')\,,
  \end{equation}
  where $[P] \in \set{0,1}$ for a proposition~$P$ is equal to~$1$ if and only 
  if~$P$ is true.
  Then \eqref{eq: simple DP} implies
  \begin{equation}\label{eq:matrixproduct}%
    h_t(v_1,v_2)
    =
    a_{12}(v_1,v_2) \cdot
    \sum_{v_3\in V(G)}
    \paren[\Big]{a_{13}(v_1,v_3) \cdot a_{23}(v_2,v_3)}
  \end{equation}

  Let $A_{13}$ be a $|V(G)|\times |V(G)|$ matrix where the rows and columns are 
  indexed by $|V(G)|$ and the value in row $v_1$ and column $v_3$ is 
  $a_{13}(v_1,v_3)$.
  Similarly, let $A_{23}$ be a $|V(G)|\times |V(G)|$ matrix where the rows and 
  columns are indexed by $|V(G)|$ and the value in row $v_3$ and column $v_2$ is 
  $a_{23}(v_2,v_3)$. The sum in \eqref{eq:matrixproduct} is exactly the value of 
  the matrix product $A_{13}A_{23}$ in row $v_1$ and column $v_2$. Thus 
  $h_t(v_1,v_2)$ can be computed by constructing the matrices $A_{13}$ and 
  $A_{23}$, performing a matrix multiplication, and multiplying the entries with 
  the values $a_{12}(v_1,v_2)$. The running time is dominated by the matrix 
  multiplication, which yields a total running time of 
  $f(H)\cdot|V(G)|^{\omega}$.
  Note that the entries of the matrices are nonnegative integers not greater 
  than~$|V(G)|^{|V(H)|}$, hence the arithmetic operations are on $O(|V(H)|\log 
  |V(G)|)$ bit integers; arithmetic operations on such integers are in time
  $f(\abs{V(H)})$ in the standard word-RAM model with $\log\abs{V(G)}$-size words.
\end{proof}

Theorem~\ref{thm: algorithm count subgraphs matrixmult}, our algorithm for
counting subgraphs whose spasm has maximum treewidth~$2$, follows 
by replacing the homomorphism counting subroutine in the proof of
Theorem~\ref{cor: algorithm count subgraphs} with the one above.

\section{Complexity of linear combination problems}
\label{sec: complexity of lincombs}

\subsection{Linear combinations of subgraphs}
\label{sec: lincomb embs}

We prove the dichotomy stated in Theorems~\ref{thm: ETH hardness count subgraphs}
and~\ref{thm: dichotomy-sub}.
Recall that due to the basis transformation \eqref{eq: sub surj hom}, 
we can express linear combinations of subgraph numbers as equivalent linear combinations of homomorphism numbers.
By Lemma~\ref{lem: graph motif hardness}, the most difficult homomorphism number
in this linear combination governs the complexity of the problem.
Since the hardness criterion for homomorphisms is treewidth,
and expressing subgraph numbers in the homomorphism basis yields terms for all graphs in the spasm, 
we first make the following observation.
\begin{fact}
  \label{fact: spasm-linear-tw}
  Let $H$ be a graph.
  If~$H$ has a maximum matching of size~$k$, the maximum treewidth among all
  graphs in~$\spasm{H}$ is~$\Theta(k)$.
\end{fact}
\begin{proof}
  Let $k$ be the size of a maximum matching of~$H$.
  Then the vertex-cover number of $H$ is at least~$k$ and at most~$2k$.
  By the second item of Fact~\ref{fact: spasm}, the vertex-cover number and thus
  the treewidth of graphs in the spasm of~$H$ is then also at most~$2k$.
  For the lower bound, we use the third item of Claim~\ref{fact: spasm}:
  Since~$H$ contains a matching of size~$k$, every $k$-edge graph occurs as a
  minor of some graph in $\spasm{H}$.
  By Theorem~\ref{thm: sparse tw constant}, there exist $k$-edge graphs~$F$ with
  treewidth ${\mathrm{tw}^*(k) \ge \Omega(k)}$.
  Moreover, taking minors does not increase the treewidth, so there is a graph
  in~$\spasm{H}$ with treewidth at least~$\Omega(k)$.
\end{proof}

Let $\mathcal A\subseteq\mathbb Q^\Graphsu$ be a family of finitely supported vectors.
Recall the matrices $\sub$ and $\surj$ from Section~\ref{sec: graph motif parameter relations}.
We define the set $\mathcal A\cdot\surj^{-1}$ as the set of all vectors
$\alpha\cdot\surj^{-1}$ for $\alpha\in\mathcal A$.
That is, for each $\alpha \in \mathcal A$, the graph motif parameter $\alpha \cdot \sub$
appears in the set $\mathcal A\cdot\surj^{-1}$ via its representation in the homomorphism basis.

\begin{thm}\label{thm: lincomb subs}
  Let $\mathcal A\subseteq\mathbb Q^\Graphsu$ be a recursively enumerable family
  of finitely supported vectors.
  If there is a constant~$t \in \mathbb N$ such that all vectors $\beta\in\mathcal
  A\cdot\surj^{-1}$ have treewidth~$\tw{\beta} \leq t$,
  then the problem~$\SubProb{\mathcal A}$ to compute
  the quantity
  \[
    \sum_{H\in\Graphsu} \alpha_H\cdot\#\Sub{H}{G}\,,
  \]
  on input $\alpha \in \mathcal A$ and an $n$-vertex graph $G$, 
  admits an algorithm with running time~$g(\alpha)\cdot n^{t+1}$.
  Otherwise, it is $\sharpWone$-hard parameterized by the description length of~$\alpha$, and it cannot be computed in time $g(\alpha)\cdot
  n^{o(t/\log t)}$ for $t=\tw{\alpha\cdot\surj^{-1}}$ unless $\sharpETH$ fails.
\end{thm}
\begin{proof}
  Recall that $\sub=\surj^{-1}\hom$ holds by \eqref{eq: hom surj sub matrix}, 
  so we have $\alpha\sub = \alpha\surj^{-1}\hom$ for all $\alpha\in\mathcal A$.
  This implies the algorithmic claim via Lemma~\ref{lem: graph motif
  algorithm} and the hardness claims via Lemma~\ref{lem: graph motif hardness}.
\end{proof}

Since any family of subgraph patterns $\mathcal H$ can be represented as a
family~$\mathcal A$ of linear combinations in which each vector has support one,
we obtain Theorem~\ref{thm: ETH hardness count subgraphs} and Theorem~\ref{thm:
dichotomy-sub} as a special case.
The quantitative lower bound regarding the vertex-cover number follows from the
relationship between the vertex-cover number and the largest treewidth in the
spasm via Fact~\ref{fact: spasm-linear-tw}.
\begin{proof}[Proof of Theorems~\ref{thm: ETH hardness count subgraphs} and~\ref{thm:
  dichotomy-sub}]
  Let $\mathcal H$ be a graph family of unbounded vertex-cover number.
  Its closure $\bigcup_{H\in\mathcal
  H}\spasm{H}$ has unbounded treewidth by Fact~\ref{fact: spasm-linear-tw}.
  We want to apply Theorem~\ref{thm: lincomb subs}.
  Let~$\mathcal A$ be the set of all $\alpha^H$ where $\alpha^H_F=[H=F]$ holds for some graph $H\in\mathcal H$.
  That is, $\alpha^H$ contains $H$ with coefficient $1$, and no other graphs.
  Clearly $\SubProb{\mathcal A}$ is equivalent to $\SubProb{\mathcal H}$.

  Now consider the set $\mathcal B$ with $\mathcal B=\mathcal A\cdot\surj^{-1}$.
  We need to prove that the graph class $\bigcup_{\beta\in\mathcal B} \supp{\beta}$
  has unbounded treewidth.
  We do so by showing that this class is in fact equal to $\bigcup_{H\in\mathcal
  H}\spasm{H}$.
  By definition, for each~$H\in\mathcal H$, the set~$\mathcal B$ contains a vector~$\beta^H$
  with $\beta^H=\alpha^H\cdot\surj^{-1}$.
  Expanding this vector-matrix product, we get
  \begin{align*}
    \beta^H_F
    &=
    \sum_{J\in\Graphsu} \alpha^H_J \cdot \surj^{-1}(J,F)
    \,.
  \end{align*}
  Since $\alpha^H_J=[H=J]$, we have $\beta^H_F=\surj^{-1}(H,F)$.
  By~\eqref{eq: surj inverse}, the latter is non-zero if and only if
  $F\in\spasm{H}$, 
  so the claim follows.
\end{proof}
  
We present an example that does not directly follow
from~\cite{curticapean2014complexity}, but that does follow from
Theorem~\ref{thm: lincomb subs}.
The following statement was proved recently using a more complicated method.
\begin{cor}[\cite{BrandRoth17}]
\label{cor: rothbrand}
  Given $k$ and $G$, counting all trees with $k$ vertices in $G$ 
  is $\sharpWone$-hard on parameter~$k$.
\end{cor}
\begin{proof}
  The number of $k$-vertex trees can be seen as a linear combination of subgraph numbers; 
  for each fixed~$k$, we set $\alpha_F = 1$ for all unlabeled
  graphs~$F\in\Graphsu$ such that~$F$ is a~$k$-vertex tree, and $\alpha_F=0$ otherwise.
  Let $\mathcal A$ be the family of all such~$\alpha$ over all~$k\in\mathbb N$.
  Since the class of all trees has unbounded vertex cover number, 
  Fact~\ref{fact: spasm-linear-tw} shows that the union of spasms of $k$-vertex trees
  has unbounded treewidth as $k$ grows.

  Write $\mathcal T_k$ for the set of all $k$-vertex trees.
  For each $k\in \mathbb N$, pick a graph $F_k\in\spasm{\mathcal T_k}$ 
  such that the sequence $F_1, F_2, \ldots $ has unbounded treewidth.
  In order to apply Theorem~\ref{thm: lincomb subs}, we need to prove that some
  vector~$\beta\in\mathcal A\cdot\surj^{-1}$ indeed has~$F_k$ in its support:
  To this end, let $\alpha\in\mathcal A$ be the vector corresponding to all $k$-vertex trees
  and let $\beta = \alpha\cdot\surj^{-1}$.
  We claim that $F_k$ is contained in the support of $\beta$.
  To see this, we expand the matrix-vector product:
  \begin{align}\label{eq: trees beta}
    \beta_{F_k} = (\alpha\cdot\surj^{-1})_{F_k}
    &=
    \sum_{T\in\mathcal T_k} \surj^{-1}(T,F_k)
    \,.
  \end{align}
  Recall from~\eqref{eq: surj inverse} that $\surj^{-1}(T,F_k)\ne 0$ if and only if ${F_k\in\spasm{T}}$.
  The same equation implies that the sign of $\surj^{-1}(T,F_k)$ is equal to~$(-1)^{\abs{V(T)}-\abs{V(F_k)}}$ for all~$T\in\mathcal T_k$ with $F_k\in\spasm{T}$.
  Since $\abs{V(T)}=k$ holds for all~$T\in\mathcal T_k$, the sign is in fact~$(-1)^{k-\abs{V(F_k)}}$.
  Therefore, all terms in the sum of \eqref{eq: trees beta} have the same sign and at least one term is non-zero; thus $\beta_{F_k}\ne0$ holds as claimed.
\end{proof}
The preceding corollary also holds for counting $k$-vertex forests, and
can be generalized to families $\mathcal A$ where each linear combination~$\alpha$ contains only graphs with the same number of vertices, and where
$\mathcal A$ contains graphs of unbounded vertex-cover number in their support.
In fact, the graphs in $\supp{\alpha}$ do not even need to have the same number
of vertices, but the same \emph{parity} of number of vertices suffices.
The proof is analogous to that of Corollary~\ref{cor: rothbrand}.

\subsection{Linear combinations of induced subgraphs}
\label{sec: lincomb strembs}

Chen, Thurley, and Weyer~\cite{chen2008understanding} consider the restriction 
of computing $\#\StrEmb H G$, parameterized by $k=\abs{V(H)}$, when the graphs~$H$ are chosen from some class~$\mathcal H$.
They prove that the problem is $\sharpWone$-hard if $\mathcal H$ is infinite, otherwise it is polynomial-time computable.
Recall that counting strong embeddings is equivalent to counting induced subgraphs.
In full analogy to Theorem~\ref{thm: lincomb subs},
we can generalize their result to a classification of linear combinations of induced subgraph numbers.
For the following statement, recall the matrices $\ext$ and $\surj$ from Section~\ref{sec: graph motif parameter relations}.

\begin{reptheorem}{thm: lincomb indsub}
  Let $\mathcal A\subseteq\mathbb Q^\Graphsu$ be a recursively enumerable family
  of finitely supported vectors.
  If there is a constant~$t \in \mathbb N$ such that all vectors $\beta\in\mathcal
  A\cdot\ext^{-1}\cdot\surj^{-1}$ have treewidth~$\tw{\beta} \leq t$,
  then the problem~$\IndProb{\mathcal A}$ to compute
  \[
    \sum_{H\in\Graphsu} \alpha_H\cdot\#\IndSub{H}{G}\,,
  \]
  on input $\alpha \in \mathcal A$ and an $n$-vertex graph $G$, admits an algorithm with running time~$g(\alpha)\cdot n^{t+1}$.
  Otherwise, it is $\sharpWone$-hard parameterized by the description length of~$\alpha$, and it
  cannot be solved in time $g(\alpha)\cdot n^{o(t/\log t)}$ for $t=\tw{\alpha\cdot\ext^{-1}\cdot\surj^{-1}}$ unless $\sharpETH$ fails.
\end{reptheorem}
The proof is analogous to that of Theorem~\ref{thm: lincomb subs},
only the basis change matrix needs to be chosen as $\ext^{-1}\cdot\surj^{-1}$ rather than $\surj^{-1}$.

On a related note, Jerrum and Meeks~\cite{jerrum2015parameterised,jerrum2015some,jerrum2014parameterised,meeks2016challenges} 
introduced the following generalization of counting induced subgraphs: 
Let~$\Phi$ be some graph property.
The problem~$\IndPropProb{\Phi}$ is, given a graph~$G$ and an integer~$k$,
to compute the number of induced $k$-vertex subgraphs that have
property~$\Phi$.
Let us write this number as~$I_{\Phi,k}(G)$.
Since $I_{\Phi,k}(G)$ can be expressed as a sum~$\sum_H\indsub(H,G)$ over all
$k$-vertex graphs~$H$ that satisfy~$\Phi$, Theorem~\ref{thm: lincomb indsub} immediately implies a complexity dichotomy for these problems.

\begin{repcorollary}{cor: jerrum meeks dichotomy}
  Let $\Phi$ be any decidable graph property.
  The problem $\IndPropProb{\Phi}$ is fixed-parameter tractable if all
  $I_{\Phi,k}$ can be represented as linear combinations of homomorphisms from
  graphs of bounded treewidth.
  Otherwise, the problem is $\sharpWone$-hard when parameterized by~$k$.
\end{repcorollary}

Finally, let us sketch how to recover the hardness result of Chen, Thurley, and
Weyer~\cite{chen2008understanding} for counting induced subgraphs from a fixed
class $\mathcal H$ as a special case of~Corollary~\ref{cor: jerrum meeks
dichotomy}:
When representing ${\#\IndSub H G}$ for a $k$-vertex graph $H$ as a linear combination of subgraph numbers $\#\Sub {H'}G$ 
via \eqref{eq: sub ext indsub}, this linear combination has a non-zero coefficient $\ext^{-1}(H,H')$ for the clique $H'=K_k$.
Indeed, the $k$-clique extends every graph on $k$ vertices, so we have $\ext(H,H')\neq 0$,
which in turn implies $\ext^{-1}(H,H')\neq 0$ by~\eqref{eq: ext inverse}.
When further writing each term $\#\Sub {H'} G$ as a linear combination of
homomorphisms $\#\Hom {H''} G$ for $H''\in\spasm {H'}$, we get exactly one term
for $H''=H'=K_k$, since $K_k$ only occurs in its own spasm, 
and we have $\surj^{-1}(H',H'')=1/\#\Aut{H'}$.
Hence the coefficient of $\#\Hom {K_k} G$ in the representation of ${\#\IndSub H G}$ is non-zero.
Thus, if $\mathcal H$ is infinite, we have unbounded cliques in the
homomorphism representation, and the problem of counting induced subgraphs from $\mathcal H$ is $\sharpWone$-hard by
Corollary~\ref{cor: jerrum meeks dichotomy}.

\section{Counting colored subgraphs in polynomial time}
\label{sec:polytime}

\newcommand{\contr}[1]{\widehat{#1}}
\newcommand{\Hcontr}{\contr{H}}
\newcommand{\Hcontrx}[1]{\contr{H}^{\setminus #1}}
\newcommand{\Hclique}[1]{H^{\textup{\textbullet}\setminus #1}}
\newcommand{\Hcliquex}{H^{\textup{\textbullet}}}

In this section, our goal is to determine which classes $\cH$ of vertex-colored graphs make $\SubProb{\cH}$ polynomial-time solvable. From Theorem~\ref{thm: dicho-sub-color}, we know that if the graphs in $\spasm H$ have treewidth bounded by some constant $c$ for every $H\in \cH$, then $\SubProb{\cH}$ is FPT and it is $\sharpWone$-hard otherwise. However, this does not tell us which of the FPT cases are actually polynomial-time solvable. Answering this question seems to require different techniques than what we have seen in the previous sections: in particular, the hardness proofs of Section~\ref{sec: lincomb embs} based on the basis transformation to homomorphisms gives reductions with superpolynomial running time and hence cannot be used to rule out polynomial-time algorithms in the cases when the problem is FPT.

To understand the limits of polynomial-time algorithms, we give two different structural characterizations that are equivalent to the condition that $\spasm H$ has bounded treewidth for a colored graph $H$. We need the following definitions:

\begin{itemize}
\item $\Hcontr$ is the simple graph obtained from $H$ by consolidating each color class into a single vertex. That is, vertex $v_i$ of $\Hcontr$ represents color class $i$ of $H$, and $v_i$ and $v_j$ of $\Hcontr$ are adjacent if there is an edge between color class $i$ and color class $j$ in $H$.
\item $\Hcliquex$ is the supergraph of $H$ obtained by making each color class a clique.
\item $\Hclique{i}$ is the supergraph of $H$ obtained by making each color class {\em except class $i$} a clique.
\item A {\em $c$-flower centered at color class $i$ in $H$} is a vertex-disjoint collection of $c$ paths of length at least 1 in $\Hclique{i}$ such that each path has both of its endpoints in class $i$. (Note that it is possible that path in the collection has length 1, that is, consists of only a single edge inside color class $i$.)
\end{itemize}
We show that $\spasm H$ having bounded treewidth can be expressed by saying that $\Hcontr$ has bounded treewidth and there are no large flowers centered at any color class. For the second characterization, we need a novel form of tree decompositions specifically tailored for our colored counting applications.

\begin{defn}
Let $H$ be a colored graph, let $(T,\bag)$ be a tree decomposition of $\Hcliquex$, and let $g:V(T)\to 2^{V(H)}$ be a function with $\sep(t)\subseteq g(t)\subseteq \bag(t)$ for every $t\in V(T)$. The triple $(T,\bag,g)$ is a {\em guarded cutvertex decomposition} if the following two properties hold:
\begin{enumerate}
\item for every $t\in V(T)$, set $g(t)$ is a vertex cover of $H[\bag(t)]$.
\item for every $t\in V(T)$ and every child $t'$ of $t$, we have $|\sep(t')\setminus g(t)|\le 1$.
\end{enumerate}
The {\em guard size} of $(T,\bag,g)$ is defined as $\max_{t\in V(T)}|g(t)|$. We say that the {\em bags have at most $c$ colors} if $\bag(t)$ for each $v\in V(T)$ contains vertices of at most $c$ different colors. 
\end{defn}
Note that this is a decomposition of $\Hcliquex$ (where each color class is a clique), but guard set $g(t)$ is a vertex cover of $H[\bag(t)]$ only (where the color classes can be sparse). It turns out the existence of precisely this kind of decompositions with guard sets of bounded size is a sufficient and necessary condition for $\spasm H$ to have bounded treewidth.

\begin{thm}\label{thm:spasmboundchar}
Let $\cH$ be a class of colored graphs. The following statements are equivalent:
\begin{enumerate}
\item There is a constant $c$ such that for every $H\in \cH$, we have that every graph in $\spasm H$ has treewidth at most $c$.
\item There is a constant $c$ such that for every $H\in \cH$, we have that $\Hcontr$ has treewidth at most $c$ and there is no $c$-flower centered at any class of $H$.
\item There is a constant $c$ such that for every $H\in \cH$, there is a guarded 
  cutvertex decomposition of $H$ with guard size at most $c$ and the bags 
  having at most $c$ colors.
\end{enumerate}
\end{thm}
As we are interested in polynomial-time algorithms, the existence statements in Theorem~\ref{thm:spasmboundchar} are not sufficient for us: we need an algorithmic statement showing that guarded cutvertex decompositions can be found efficiently. The following theorem gives an algorithmic version of the (2) $\Rightarrow$ (3) implication of Theorem~\ref{thm:spasmboundchar}.
\begin{thm}\label{thm:treedec2}
Let $H$ be a colored graph such that 
\begin{itemize}
\item $\Hcontr$ has treewidth at most $w$, and
\item $H$ has no $c$-flower centered at any of the color classes.
\end{itemize}
Then we can compute a guarded cutvertex decomposition $(T,\bag,g)$ of $H$
with guard size $O((w+c)^2w^2)$ such that the bags have at most $O((c+w)w)$ colors.
\end{thm}

We present a dynamic programming algorithm solving $\Sub{H}{G}$, given a guarded cutvertex decomposition of $H$ with bounded guard size. This can be used to give another proof that $\SubProb{H}$ is FPT when $\spasm H$ has bounded treewidth for every $H\in \cH$ (but we have already established that in Theorem~\ref{thm: dicho-sub-color}). More importantly, this dynamic programming algorithm can be used to solve all the polynomial-time solvable cases of $\SubProb{H}$. The running time of the algorithm has a factor that has exponential dependence on the number of certain types of vertices in the guarded decomposition. If the number of these vertices is at most logarithmic in the size of $H$, then this exponential factor is polynomially bounded in the size of $H$ and the running time is polynomial. 

We need the following definition for the formal statement of the running time of the algorithm.
In a guarded cutvertex decomposition of $H$, the vertices of $\bag(t)\setminus g(t)$ are of two types: either there is some child $t'$ of $t$ such that $v$ is adjacent to $\comp(t')$ or not. Moreover, the requirement $|\sep(t')\setminus g(t)|\le 1$ ensures that at most one vertex of $\bag(t)\setminus g(t)$ can be adjacent to $\comp(t')$.   Formally, for every $t\in V(T)$ and $v\in \bag(t)\setminus g(t)$, we define a set $h_t(v)\subseteq V(T)$ the following way: a child $t'$ of $t$ is in $h_t(v)$ if $v$ is adjacent to $\comp(t')$. That is, $h_t(v)$ contains those children of $t$ that are ``hanging'' on the vertex $v$ (and on some of the guards $g(t)$). Let $\lambda(t)\subseteq \bag(t)\setminus g(t)$ contain those vertices $v$ for which $h_t(v)\neq\emptyset$. 
\begin{lem}\label{lem:ordembalg}
Let $H$ be a colored graph and $(T,\bag,g)$ be a tight guarded cutvertex decomposition of $H$ with guard size at most 
$c$ and at most $c$ colors in each bag. Suppose that $|\lambda(t)|\le d$ for every $t\in V(T)$.
Then $\Emb{H}{G}$ can be computed in time $2^{O(d)}\cdot |V(G)|^{2^{O(c)}}$.
\end{lem}

Finally, we need to show that if the algorithm of
Lemma~\ref{lem:ordembalg} cannot be used to give a polynomial-time
algorithm for $\SubProb{\cH}$, then the problem is unlikely to be
polynomial-time solvable. We observe that $\lambda(t)$ being large at
some node $t$ implies the existence of a large half-colorful matching
and then such matchings can be used in a reduction from counting $k$-matchings
to $\SubProb{\cH}$. If half-colorful matchings of super-logarithmic size appear 
in the graphs of $\cH$, then this reduction can be used to count perfect 
matchings in $2^{o(n)}$ time.

\subsection{Obtaining a tree decomposition}
The goal of this subsection is to prove Theorems~\ref{thm:spasmboundchar} and \ref{thm:treedec2}. After stating useful facts about guarded cutvertex decompositions, we discuss two technical tools: we formally state a minor but tedious transformation of the tree decomposition and recall some known results about finding $A$-paths. Then we prove Theorem~\ref{thm:treedec2}, which shows the (2) $\Rightarrow$ (3) implication of Theorem~\ref{thm:spasmboundchar}. Then we complete Theorem~\ref{thm:spasmboundchar} with the implications (3) $\Rightarrow$ (1) and (1) $\Rightarrow$ (2).

\subsubsection{Useful facts}

We make a few observations on the structure  of guarded decompositions.

\begin{lem}\label{cl:nonkappaneighbor}
Let $H$ be a colored graph and let $(T,\bag,g)$ be a guarded cutvertex decomposition of $H$.
If $v$ is a vertex in $\bag(t)\setminus (g(t)\cup \lambda(t))$ for some $t\in V(T)$, then
\begin{enumerate}
\item $v$ does not appear in $\bag(t')$ for any $t'\in V(T)$ with $t\neq t'$.
\item $\bag(t)$ contains every vertex with the same color as $v$.
\item $N_H(v)\subseteq g(t)$.
\end{enumerate}
\end{lem}
\begin{proof}
Suppose $v$ appears in $\bag(t')$ for some $t'\neq t$. If $t'$ is a descendant of $t$, then $v$ has to appear in $\bag(t'')$ for some child $t''$ of $t$, which means that $v$ is in $\sep(t'')$ and hence $v\in\lambda(t)$ follows, a contradiction. If $t'$ is not a descendant of $t$, then $v$ has to appear in $\sep(t)\subseteq g(t)$, again a contradiction. Thus $v$ cannot appear in $\bag(t')$ for any $t'\neq t$, proving statement 1. Furthermore, this implies that every neighbor of $v$ in $\Hcliquex$ has to appear in $\bag(t)$. As every vertex with the same color as $v$ is a neighbor of $v$ in $\Hcliquex$, statement 2 follows. Since $H$ is a subgraph of $\Hcliquex$, it also follows that $N_H(v)$ is in $\bag(t)$. The fact that $g(t)$ is a vertex cover of $H[\bag(t)]$ further implies that every neighbor of $v$ is in $g(t)$, proving statement 3.
\end{proof}

It is not obvious that if $H$ has a guarded cutvertex decomposition with bounded guard size, then $H$ has bounded treewidth: the definition does not say anything about the size of the bags. However, it is not difficult to show that this is indeed the case, although we need to modify the decomposition to make the size of the bags bounded.
\begin{lem}\label{lem:guardedbounded}
If $H$ has a guarded cutvertex decomposition with guard size at most $c$, then $H$ has treewidth at most $c$.
\end{lem}
\begin{proof}
Let $(T,\bag,g)$ be a guarded cutvertex decomposition of $H$ with guard size at most $c$. We obtain a tree decomposition $(T',\bag')$ of $H$ the following way. For every node $t\in V(T)$, we introduce $t$ into $T'$ and set $\bag'(t)=g(t)$. Then for every $v\in \bag(t)\setminus g(t)$, we introduce a child $t_v$ of $t$ into $T'$ and set $\bag'(t_v)=g(t)\cup \{v\}$. If a child $w$ of $t$ in $T$ has $\sep(w)\setminus g(t)=\{v\}$ (by definition of the guarded cutvertex decomposition, there can be at most one vertex in this set), then we attach $w$ to be a child of $t_v$ in $T'$, instead of being a child of $t$.
It is not difficult to verify that the new decomposition satisfies the properties of a tree decomposition.
\end{proof}

The following lemma gives a quick sanity check: a guarded cutvertex decomposition rules out the possibility of a large matching in any color class or between any two color class. This is expected, as such matchings would make the counting problem hard.
\begin{lem}\label{lem:guardnomatching}
If $H$ has a guarded cutvertex decomposition $(T,\bag,g)$ with guard size at most $c$, then $H$ has no matching $x_1y_1$, $\dots$, $x_{c+1}y_{c+1}$ where every $x_i$ is in class $C_x$ and every $y_i$ is in class $C_y$ (possibly $C_x=C_y$).
\end{lem}
\begin{proof}
As each color class is a clique in $\Hcliquex$, there is a node $t_x$ (resp., $t_y$) with $C_x\subseteq \bag(t_x)$ (resp, $C_y\subseteq \bag(t_y)$).
If $C_x=C_y$, then $x_1y_1$, $\dots$, $x_{c+1}y_{c+1}$ is a matching in $H[\bag(t_x)]$, contradicting the assumption that $g(t)$ is a vertex cover of $H[\bag(t_x)]$.
If $C_x\neq C_y$, then assume without loss of generality that $t_x$ is not an ancestor of $t_y$. Then $\sep(t_x)$ separates $\bag(t_x)\setminus \sep(t_x)$ and $\bag(t_y)\setminus \sep(t_x)$. As $|\sep(t_x)|\le c$, there is an edge $x_iy_i$ such that $x_i,y_i\not \in \sep(t_x)$. However, this means that there can be no common bag where both $x_i$ and $y_i$ can appear. 
\end{proof}

\subsubsection{Massaging a tree decomposition}
\label{sec:mass-tree-decomp}

Let $(T,\bag)$ be a tree decomposition of a graph $G$. Let
$t'\in V(T)$ be a child of $t\in V(T)$. There are certain easy
modifications that allow us to make the tree decomposition nicer and
more useful for dynamic programming. 
\begin{defn}
We say that tree decomposition $(T,\bag)$ of a graph $H$ is {\em tight} if for every $t\in V(T)$ and $v\in \sep(t)$, vertex $v$ has at least one  neighbor in $\comp(t)$.
\end{defn}

First, if the decomposition is not tight, then $\sep(t)$ could contain
vertices not in $N^G(\comp(t))$, but we may safely remove any such
vertex from $\bag(t)$. Second, it is not necessarily true that
$G[\comp(t')]$ is connected: it may contain two or more connected
components of $G-\bag(t)$. However, it is not difficult to modify the
tree decomposition in such a way to ensure that $G[\comp(t')]$ is
connected: we need to introduce restricted copies of the subtree of
$T$ rooted at $t'$, one for each component of $G-\bag(t)$ contained in
$\comp(t')$. We will need this transformation in the proof of
Theorem~\ref{cl:treedecremains}, as we can prove the bound on
$\sep(t')$ only if $G[\comp(t')]$ is connected. While these
transformations are easy to perform, in the following lemma we state
them in a very formal way. The reason is that in the proof of
Theorem~\ref{thm:treedec2} we need to verify how certain properties of
the original decomposition survive this transformation.

\begin{lem}\label{lem:treedecmod}
Let $(T^1,\bag^1)$ be a tree decomposition of a connected graph $G$ such that $\bag^1(r^1)\neq \emptyset$ for the root $r^1$ of $T^1$. We can compute in polynomial time another tree decomposition $(T^2,\bag^2)$ of $G$ and a mapping $f\colon V(T^2)\to V(T^1)$ with the following properties:
\begin{nlist}{P}
\item\label{p:treeroot} if $r^2$ is the root of $T^2$, then $f(r^2)=r^1$ is the root of $T^1$ and $\bag^2(r^2)=\bag^1(r^1)\neq \emptyset$,
\item\label{p:treesub} $\bag^2(t)\subseteq \bag^1(f(t))$ for every $t\in V(T^2)$,
\item\label{p:treeex}$(\bag^1(f(t))\setminus \bag^2(t))\cap \cone^2(t)=\emptyset$ for every $t\in V(T^2)$,
\item\label{p:treeconn} $G[\comp^2(t)]$ is connected for every $t\in V(T^2)$, and
\item\label{p:treesep} $\sep^2(t)=N^G(\comp^2(t))$ for every $t\in V(T^2)$.
\end{nlist}
\end{lem}
\begin{proof}
  We give a recursive procedure for computing the required tree
  decomposition $(T^2,\bag^2)$. Let $r^1\in V(T^1)$ be the root of
  $T^1$. By assumption, we have $\bag^1(r^1)\neq\emptyset$.  Let
  $V_1$, $\dots$, $V_t$ be the vertex sets of components of
  $G-\bag^1(r^1)$ and let $G_i=G[V_i\cup N^G(V_i)]$. If $t=0$, then it
  is easy to see that the tree decomposition consisting only of the
  root bag $\bag^1(r_1)$ satisfies the requirements. Otherwise, for
  $1\le i \le t$, the root $r^1$ has a child $r^1_i$ in $T^1$ such
  that the vertices of $V_i$ appear only in bags of the subtree
  $T^1_i$ of $T$ rooted at $r^1_i$ (note that $r^1_i=r^1_j$ is
  possible for $i\neq j$). This gives a tree decomposition
  $(T^1_i, \bag^1_i)$ of $G_i$, where we define
  $\bag^1_i(t)=\bag^1(t)\cap (V_i\cup N^G(V_i))$ for every
  $t\in V(T_i)$. Note that $\bag^1_i(r^1_i)\neq\emptyset$: if
  $\bag^1(r^1_i)$ contains no vertex of $V_i\cup N^G(V_i)$, then this
  would imply that $V_i\neq\emptyset$ and $\bag^1(r^1)\neq\emptyset$
  are not in the same component of $G$, that is, $G$ is not connected.
  Let us recursively call the procedure on $(T^1_i,\bag^1_i)$ and let
  $(T^2_i,\bag^2_i)$ be the resulting tree decomposition of $G_i$ with
  root $r^2_i\in V(T^2_i)$ and  
  let $f_i:V(T^2_i)\to V(T^1_i)$ be the corresponding function.

  We construct tree decomposition $(T^2,\bag^2)$ the following
  way. The tree $T^2$ is constructed by taking a disjoint copy of
  every $T^2_i$, introducing a new root $r^2$, and connecting $r^2$
  with every $r^2_i$. The function $\bag^2$ is defined the obvious
  way: we let $\bag^2(r^2)=\bag(r^1)$ and $\bag^2(t)=\bag^2_i(t)$ for
  every $t\in V(T^2_i)$. To verify that $(T^2,\bag^2)$ is a tree
  decomposition, we need to check for every vertex $v$ that the bags
  containing $v$ form a connected subtree of $T^2$. If
  $v\not\in \bag^1(r^1)$, then $v$ appears only in the bags of $T^1_i$
  for some $i$ and then the statement follows from the fact that
  $(T^2_i,\bag^2_i)$ is a tree decomposition.  If $v\in \bag^1(r^1)$,
  then $v$ may appear in the bags of $T^1_i$ for more than one
  $i$. However, for each such $i$, we have that $v\in \bag^1_i(r^1_i)$
  holds, and then the fact that $(T^2_i,\bag^2_i)$ and $f_i$ satisfy
  property \ref{p:treeroot} implies that $v$ appears in
  $\bag^2_i(r^2_i)=\bag^1_i(r^1_i)$.

  Let us define $f(r^2)=r^1$ and let $f(t)=f_i(t)$ if $t\in V(T^2_i)$.
  We claim that tree decomposition $(T^2,\bag^2)$ and the function $f$
  satisfy the requirements of the lemma.  For the root $r^2$ of $T^2$,
  properties \ref{p:treeroot}--~\ref{p:treeex} hold by construction,
  as we have
  $\bag^2(r^2)=\bag(f(r^2))=\bag(r^1)\neq\emptyset$. Property
  \ref{p:treeconn} follows from the assumption that $G$ is
  connected. We have $\sep^2(r^2)=\emptyset$, satisfying property
  \ref{p:treesep}.

For a node $t\in V(T^2_i)$, we verify properties \ref{p:treesub}--\ref{p:treesep} as follows.
\begin{enumerate}
\item[\ref{p:treesub}]
  $\bag^2(t)=\bag^2_i(t)\subseteq\bag^1_i(f_i(t))\subseteq
  \bag^1(f_i(t))=\bag^1(f(t))$,
  where the first subset relation uses the assumption that
  $(T^2_i,\bag^2_i)$ satisfies property \ref{p:treesub} and the second
  subset relation follows from the definition $\bag^1_i$.
\item[~\ref{p:treeex}] We have $\bag^2(t)=\bag^2_i(t)$,
  $\cone^2(t)=\cone^2_i(t)\subseteq V_i\cup N^G(V_i)$ and
  $\bag^1_i(f_i(t))=\bag^1_i(f(t))=\bag^1(f(t))\cap (V_i\cup
  N^G(V_i))$,
  implying
  $\bag^1_i(f_i(t))\cap\cone^2(t)=\bag^1(f(t))\cap \cone^2(t)$.  Thus
  $(\bag^1(f(t))\setminus \bag^2(t))\cap \cone^2(t) =
  (\bag^1_i(f_i(t))\setminus \bag^2_i(t))\cap \cone^2_i(t)=\emptyset$,
  where the second equality uses that $(T^2_i,\bag^2_i)$ satisfies
  property \ref{p:treeex}.
\item[\ref{p:treeconn}] If $t=r^2_i$, then
  $\cone^2(t)=V_i\cup N_G(V_i)$ and $\comp^2(t)=V_i$, hence
  $G[\comp^2(t)]$ is connected by the definition of $V_i$. If $t\in V(T^2_i)\setminus \{r^2_i\}$, then $\comp^2(t)=\comp^2_i(t)$, hence the fact that $(T^2_i,\bag^2_i)$ satisfies property \ref{p:treeconn} implies that $G[\comp^2(t)]=G[\comp^2_i(t)]$ is connected.
\item[\ref{p:treesep}]  If $t= r^2_i$, then
  $\sep^2(t)=N_G(V_i)$.  If $t\in V(T^2_i)\setminus \{r^2_i\}$, then $\sep^2(t)=\sep^2_i(t)$ and $\comp^2(t)=\comp^2_i(t)$, hence the fact that $(T^2_i,\bag^2_i)$ satisfies property \ref{p:treesep} implies $\sep^2(t)=\sep^2_i(t)= N_{G_i}(\comp^2_i(t))=N_{G_i}(\comp^2(t))=N_G(\comp^2(t))$.
\end{enumerate}

To argue that this recursive procedure terminates, consider the
measure $\zeta(T^1,\bag^1)=|V(T^1)|\cdot |V(G)\setminus \bag^1(r^1)|$, where $r^1$
is the root of $T^1$. If this measure is zero, then the recursion is
terminated: $\bag^1(r^1)$ contains every vertex of $G$, hence $G-\bag^1(r^1)$
has zero components. The measure strictly decreases in each step, as
the tree decomposition has strictly fewer nodes in each recursive
call. Moreover, we claim that when processing tree decomposition
$(T^1,\bag^1)$, the sum of the measures in the recursive calls is at most
$\zeta(T^1,\bag^1)$. To see this, notice that $V(G_i)\setminus \bag^1(r^1_i)$
and $V(G_j)\setminus \bag^1(r^1_j)$ are disjoint for $i\neq j$: we have
$V(G_i)\cap V(G_j)\subseteq \bag^1(r^1)$, but a vertex $v\in \bag^1(r^1)$
cannot appear in a bag of $T^1_i\setminus \{r^1_i\}$ without appearing in
$\bag^1(r^1_i)$ as well. Therefore, we can bound the number of leafs of
the recursion tree by
$|V(T^1)|\cdot |V(G)\setminus \bag^1(r^1)|\le |V(T^1)|\cdot |V(G)|$.
\end{proof}

\subsubsection{A-paths}
If $A\subseteq V(G)$ is a set of vertices, then path $P$ is an {\em
  $A$-path} if both endpoints are in $A$ (we allow here that $P$ has
other vertices in $A$, although this will not make much difference, as
every $A$-path has a subpath whose internal vertices are not in $A$).
The following classical result shows that either there are many
vertex-disjoint $A$-paths, or they can be covered by a bounded number
of vertices.
\begin{thm}\label{thm:apath}
Given a graph $G$, a set $A\subseteq V(G)$ of vertices, and an integer $k$, in polynomial time we can find either
\begin{enumerate}
\item $k$ pairwise vertex-disjoint $A$-paths, or
\item a set $S\subseteq V(G)$ of at most $2k-2$ vertices such that
  $G-S$ has no $A$-path.
\end{enumerate}
\end{thm}

If we want to cover the $A$-paths using only vertices in $A$, then we
may need significantly more vertices: for example, this is the case if
there is a vertex $v\not\in A$ adjacent to every vertex of
$A$. However, what we can show is that selecting vertices outside $A$
can be avoided unless they are highly connected to $A$ in the
following sense.  We say that a vertex $v$ is {\em $\ell$-attached} to
a set $A$ of vertices if there is a set of $\ell$ paths of length at
least 1 connecting $v$ and $A$ such that they share only the vertex
$v$, but otherwise disjoint. Note that the definition is slightly
delicate if $v\in A$: then we need $\ell$ paths connecting $v$ to
$\ell$ other vertices of $A$. It can be tested in polynomial time if a
vertex $v$ is $\ell$-attached using simple flow/disjoint paths
algorithms.

\begin{thm}\label{thm:apath2}
Given a connected graph $G$, a set $A\subseteq V(G)$ of vertices, and integers $k,\ell$ with $\ell\ge 2k$, in polynomial time we can find either
\begin{enumerate}
\item $k$ pairwise vertex-disjoint $A$-paths , or
\item a set $A^*\subseteq A$ of at most $(2k-2)\ell$  vertices and a set
  $S^*\subseteq V(G)$ of at most $2k-2$ vertices such that
  $G-(S^*\cup A^*)$ has no $A$-path and $S^*$
  is precisely the set of vertices that are $\ell$-attached to $A$.
\end{enumerate}
\end{thm}
\begin{proof}
  Let us first invoke the algorithm of Theorem~\ref{thm:apath}. If it
  returns $k$ pairwise vertex-disjoint $A$-paths, then we are
  done. Otherwise, let $S$ be a set of at most $2k-2$ vertices hitting
  every $A$-path. We say that a vertex $v\in S$ {\em sees} vertex
  $t\in A$ if $v\neq t$ and there is a $v-t$ path that intersect $S$
  only in $v$ (it is possible that $v=t$).  We claim that if $v$ sees
  at least $\ell+1$ vertices of $A$, then $v$ is $\ell$-attached to
  $A$. If $v$ sees at least $\ell+1$ vertices of $A$, then at least
  $\ell$ of them are distinct from $v$. Suppose that $v$ sees vertices
  $t_1$, $\dots$, $t_\ell$ of $A\setminus \{v\}$, that is, for every
  $1\le i \le \ell$, there is a $v-t_i$ path $P_i$ of length at least
  1 that intersects $S$ only in $v$. If these paths share only the
  vertex $v$, then $v$ is $\ell$-attached to $A$. Otherwise, suppose
  that $P_i$ and $P_j$ share a vertex different from $v$. Then
  $(V(P_i)\cup V(P_j))\setminus \{v\}$ contains no vertex of $S$ and
  induces a connected graph in $G$ connecting two distinct vertices
  $t_i$ and $t_j$ of $A$, contradicting the assumption that $S$ hits
  every $A$-path.

  Let $S^*\subseteq S$ be the set of vertices that are $\ell$-attached to $A$. By the argument in the previous paragraph, every in $|S\setminus S^*|$ see at most $\ell$ vertices of $A$. Let $A^*\subseteq A$ be the set of
  vertices seen by $S\setminus S^*$. We clearly have $|S^*|\le 2k-2$
  and $|A^*|\le |S\setminus S^*|\ell \le (2k-2)\ell$.

  We claim that $S^*\cup A^*$ hits every $A$-path $P$. Let
  $t_1,t_2\in A$ be the endpoints of $P$; we may assume that $P$
  contains no other vertex of $A$. Let $v$ be the first vertex of $P$
  in $S$ when going from $t_1$ to $t_2$ (it is possible that $v=t_1$
  or $v=t_2$). This means that $v$ sees $t_1$. If $v$ sees at least
  $\ell+1$ vertices of $A$ (that is, at least $\ell$ vertices other than
  $t_1$), then, as we have seen above, $v$ is $\ell$-attached to $A$,
  hence $v\in S^*$ and we are done. Otherwise, $t_1$ is one of the at
  most $\ell$ vertices seen by $v\in S\setminus S^*$, hence
  $t_1\in A^*$ and we are done again.

  Finally, we need to show that $S^*$ is precisely the set of vertices
  $\ell$-attached to $A$. Every vertex is in $S^*$ is $\ell$-attached
  to $A$ by construction, hence we need to show only the converse: if
  $v$ is $\ell$-attached to $A$, then $v\in S^*$. Suppose that $P_1$,
  $\dots$, $P_\ell$ are $v-A$ paths of length at least 1, sharing only
  vertex $v$. If $v\in S$, then $v$ was introduced into $S^*$. If
  $v\not\in S$, then $\ell\ge 2k \ge |S|+2$ implies that there are two
  paths $P_i$ and $P_j$ disjoint from $S$. But the these two paths
  show that there is an $A$-path in $G-S$, a contradiction.
\end{proof}

\subsubsection{Constructing the decomposition}

The following lemma contains the main technical part of the section, where we exploit the lack of $c$-flowers to constructing a certain kind of decomposition.

\begin{lem}\label{lem:treedec1}
Let $H$ be a connected colored graph such that 
\begin{itemize}
\item $\Hcontr$ has treewidth at most $w$, and
\item $H$ has no $c$-flower centered at any of the color classes.
\end{itemize}
Then we can compute in polynomial time a tree decomposition $(T^1,\bag^1)$ of $\Hcliquex$ and a
guard set $g^1(t)$ with $g^1(t)\subseteq \bag^1(t)$
of size $O((c+w)^2w^2)$ for every $t\in V(T^1)$ such that the following
holds:
\begin{enumerate}
\item for every $t\in V(T^1)$, $g^1(t)$ is a vertex cover of $H[\bag^1(t)]$.
\item for every $t\in V(T^1)$ and every component $C$ of $\Hcliquex[\cone^1(t)]-\bag^1(t)$, it holds that $|N^{\Hcliquex}(C)\setminus g^1(t)|\le 1$.
\end{enumerate}
Furthermore, every bag $\bag^1(t)$ contains vertices from $O((c+w)w)$ colors.
\end{lem}
\begin{proof}
It will be convenient to extend the definition of $\Hclique{i}$ the following way: for a set $S$ of colors, $\Hclique{S}$ is the graph $H$ with every color class $C_i$ with $i\neq S$ made a clique.
  Consider first a tree decomposition $(\contr{T},\contr{\bag})$ of
  $\Hcontr$ of width $w$, that is, $|\contr{\bag}(t)|\le w+1$ for every
  $t\in V(\contr{T})$. We can then define a tree decomposition
  $(T^0,\bag^0)$ of $\Hcliquex$ in a natural way: let $T^0=\contr{T}$
  and let $\bag^0(t)=\bigcup_{i\in \contr{\bag}(t)}C_i$. Note that
  this decomposition has the property that every bag $\bag^0(t)$
  contains vertices from at most $w+1$ color classes.

Let $k:=(2c+w)(w+1)$.
\begin{claim}\label{cl:notmanypaths}
  For every $t\in V(T^0)$, graph $\Hclique{\contr{\bag}(t)}$ has no $k$ pairwise
  vertex-disjoint $\bag^0(t)$-paths.
\end{claim}
\begin{proof}
  Suppose that there is a collection of $k$ pairwise vertex-disjoint
  $\bag^0(t)$-paths in $\Hclique{\contr{\bag}(t)}$. Each endpoint of
  such a path is in class $C_i$ for some $i\in \contr{\bag}(t)$ and
  hence there has to be an $i\in \contr{\bag}(t)$ such that at least
  $2c+w$ of these paths have at least one endpoint in $C_i$. If both
  endpoints of a path are in $C_i$, then it is a $C_i$-path in
  $\Hclique{\contr{\bag}(t)}$, and hence also in its supergraph
  $\Hclique{i}$.  If there are two paths such that each of them has an
  endpoint in $C_i$ and an endpoint in $C_j$ for some
  $j\in \contr{\bag}(t)\setminus \{i\}$, then they together form a
  $C_i$-path in $\Hclique{i}$ (as $C_j$ is a clique in
  $\Hclique{i}$). For every $j\in \contr{\bag}(t)\setminus \{i\}$, if
  the collection contains an odd number of paths with one endpoint in
  $C_i$ and the other endpoint in $C_j$, then let us throw away one
  such path; this way we are throwing away at most $w$ paths and hence
  at least $2c$ paths with at least one endpoint in $C_i$ remain. The
  paths with both endpoints in $C_i$ are $C_i$-paths and the paths
  with an endpoint in $C_i$ and an endpoint in $C_j$ can be matched up
  to form $C_i$-paths. Therefore, the $2c$ remaining paths form at
  least $c$ vertex-disjoint $C_i$-paths in $\Hclique{i}$,
  contradicting the assumption that there is no $c$-flower centered at
  class $C_i$.
\end{proof}

We define a new tree decomposition $(T^1,\bag^1)$ of $\Hcliquex$ the
following way. Let $T^1=T^0=\contr{T}$.  Let us use the algorithm of
Theorem~\ref{thm:apath2} on the graph $\Hclique{\contr{\bag}(t)}$, set
$\bag^0(t)$, and integers $k=(2c+w)(w+1)$ and $\ell=2k$.  By
Claim~\ref{cl:notmanypaths}, the algorithm cannot return a set of $k$ pairwise vertex-disjoint
$\bag^0(t)$-paths, hence we obtain a set $A^*_t\subseteq \bag^0(t)$ of size at most $(2k-2)\ell$
and the set $S^*_t$ of all the at most $2k$ vertices $\ell$-attached
to $\bag^0(t)$ in $\Hclique{\contr{\bag}(t)}$. We define
$\bag^1(t)=\bag^0(t)\cup (S^*_t\cap\cone^0(t))$ and
$g^1(t)=(S^*_t\cap \cone^0(t))\cup A^*_t$; note that we have
$|g^1(t)|\le 2k+(2k-2)\ell=O(k^2)=(c+w)^2w^2$, as required.
Moreover, as we add at most $2k$ vertices to $\bag^0(t)$, it follows that $\bag^1(t)$ contains vertices from at most $w+2k+1=O((c+w)w)$ colors classes.

We need to prove that $(T^1,\bag^1)$ is a tree decomposition satisfying the requirements. To show that each vertex corresponds to a connected subtree in $(T^1,\bag^1)$, we need the following claim.
\begin{claim}\label{cl:childconnected}
  If $t'$ is a child of $t$ in $T^1$ and $v\in \cone^0(t')$ is
  $\ell$-attached to $\bag^0(t)$ in $\Hclique{\contr{\bag}(t)}$, then $v\in \bag^1(t')$.
\end{claim}
\begin{proof}
  If $v\in \bag^0(t')$, then $v\in \bag^1(t')$ and we are done. Let us
  suppose that $v\not\in \bag^0(t')$.  Let $P_0$, $\dots$, $P_\ell$ be
  $v-\bag^0(t)$ paths in $\Hclique{\contr{\bag}(t)}$ that share only
  $v$. Each path $P_j$ contains a vertex of $\bag^0(t')$; let $u_j$ be
  the vertex of $P_j$ that is in $\bag^0(t')$ and closest to $v$ on
  $P_j$.  Let $P'_j$ be the subpath of $P_j$ from $v$ to $u_j$.  Note
  that path $P'_j$ is of length at least one: $v\not\in \bag^0(t')$ by
  assumption, hence $v\neq u_j$.  We claim that $P'_j$ is a path not
  only in $\Hclique{\contr{\bag}(t)}$, but also in
  $\Hclique{\contr{\bag}(t')}$. If an edge appears in $\Hclique{\contr{\bag}(t)}$, but not in $\Hclique{\contr{\bag}(t')}$, then it has to connect two vertices of $C_i$ for some $i\in \contr{\bag}(t')$. However, 
  every vertex of $C_i$ is in $\bag^0(t')$ and $P'_j$ contains a only single vertex of $\bag^0(t')$, hence it cannot use such an edge. Thus the paths $P'_1$, $\dots$, $P'_\ell$ remain paths in $\Hclique{\contr{\bag}(t')}$ and they show that $v$
  is $\ell$-attached to $\bag^0(t')$ in $\Hclique{\contr{\bag}(t')}$. This means that $v$ (which is in $\cone^0(t')$ by assumption) appeared in the set $S^*_{t'}$ and hence it was introduced into
  $\bag^1(t')$.  \end{proof}

\begin{claim}\label{cl:treedecremains}
$(T^1,\bag^1)$ is a tree decomposition of $\Hcliquex$.
\end{claim}
\begin{proof}
  We need to show that the subset
  $V_v=\{t\in V(T^1)\mid v\in \bag^1(t)\}$ induces a connected subtree
  of $T^1$ for every $v\in V(H)$. Suppose that $v\in \bag^1(t)\setminus \bag^0(t)$, that is,
  $v$ is a vertex newly introduced into $\bag^1(t)$.  By the definition
  of $\bag^1(t)$, this is only possible if $v\in \cone^0(t)$ and $v$ is
  $\ell$-attached to $\bag^0(t)$ in $\Hclique{\contr{\bag}(t)}$.  As
  $v\not\in \bag^0(t)$, this also means that $v\in \cone^0(t')$ for some
  child $t'$ of $t$. Then by Claim~\ref{cl:childconnected},
  $v\in\bag^1(t')$ also holds.  Thus we have shown that
  $v\in \bag^1(t)\setminus \bag^0(t)$ implies that $v\in \bag^1(t')$ for a
  child $t'$ of $t$. This implies that $V_v$ is connected.
  \end{proof}

Now let us verify that $(T^1,\bag^1)$ satisfies the two
properties required by the lemma. Let $x,y\in \bag^1(t)$ be two
adjacent vertices in $H$ for some $t\in i\in \contr{\bag}(t)$ and
suppose that none of them is in $g^1(t)$. Then $x,y\in \bag^0(t)$ and
$xy$ is an edge of $\Hclique{\contr{\bag}(t)}$, forming a
$\bag^0(t)$-path of length 1 in
$\Hclique{\contr{\bag}(t)}$. However, $g^1(t)=S^*_t\cup A^*_t$
covers every $\bag^0(t)$-path in $\Hclique{\contr{\bag}(t)}$, a contradiction.

For the second property, let $C$ be a connected component of
$G[\cone^1(t)\setminus\bag^1(t)]$. A $(T^1,\bag^1)$ is a tree
decomposition, we have $C\subseteq \cone^1(t')$ for some child $t'$ of
$t$. Let $x,y\in N_{\Hcliquex}[C]\setminus g^1(t)$ be two distinct
vertices. Observe that
$x,y\in \bag^1(t)\setminus g^1(t)\subseteq \bag^0(t)$ and $C$ is
disjoint from $\bag^1(t)\supseteq g^1(t)$. Thus there is an $x-y$ path
$P$ in $\Hcliquex$ with internal vertices in $C$ (and having length at
least 2).  We claim that $P$ is a path also in
$\Hclique{\contr{\bag}(t)}$. It is not possible that $P$ contains an
edge $xy$ of $\Hcliquex$ that does not appear in
$\Hclique{\contr{\bag}(t)}$: if $x,y\in C_i$ for some
$i\in \contr{\bag}(t)$, then $x,y\in \bag^1(t)$ and we have that $C$ is disjoint
from $\bag^1(t)$. Thus $P$ is a $\bag^0(t)$-path in
$\Hclique{\contr{\bag}(t)}$, but every such path is covered by
$g^1(t)=A^*_i\cup S^*_i\subseteq \bag^1(t)$, which is disjoint from $P$,
a contradiction.
\end{proof}

We are now ready to prove Theorem~\ref{thm:treedec2}.  The guarded cutvertex decomposition we need can be obtained by invoking
Lemma~\ref{lem:treedecmod} on the tree decomposition $(T^1,\beta^1)$
produced by Lemma~\ref{lem:treedec1}.

\begin{proof}[Proof (of Theorem~\ref{thm:treedec2})]
  Let us compute first the tree decomposition $(T^1,\bag^1)$ of
  $\Hcliquex$ given by Lemma~\ref{lem:treedec1} and then use
  Lemma~\ref{lem:treedecmod} to obtain the tree decomposition
  $(T^2,\bag^2)$ and the mapping $f: V(T^2)\to V(T^1)$. For every
  $t\in V(T^2)$, we define $g^2(t)=g^1(f(t))\cap \bag^2(t)$. We claim that
  $(T^2,\bag^2)$ and $g^2$ satisfy the required properties:
\begin{enumerate}
\item for every $t\in V(T^2)$, $g^2(t)$ is a vertex cover of $H[\bag^2(t)]$.
\item for every $t\in V(T^2)$ and every child $t'$ of $t$, we have $|\sep^2(t')\setminus g^2(t)|\le 1$.
\end{enumerate}

  For the first property, Lemma~\ref{lem:treedec1} implies that $g^1(t)$ is a vertex cover of
  $H[\bag^1(t)]$ for every $t\in V(T)$. Thus for every $t\in V(T^2)$,
  the subset $g^1(f(t))\subseteq \bag^1(f(t))$ is a vertex cover of
  $H[\bag^1(f(t))]$. As $\bag^2(t)\subseteq \bag^1(f(t))$ by \ref{p:treesub}, it follows
  that $g^2(t)=g^1(f(t))\cap \bag^2(t)$ is a vertex cover of
  $H[\bag^1(f(t))\cap \bag^2(t)]=H[\bag^2(t)]$, as required

  For the second property, consider a child $t'$ of $t$.  By
  \ref{p:treeconn}, we have that $C:=\Hcliquex[\comp^2(t')]$ is
  connected, that is, $C$ is a connected component of
  $\Hcliquex-\bag^2(t)$ with $V(C)\subseteq \cone^2(t)$.  Property
  \ref{p:treesub} and the definition of $\cone$ imply
  $\cone^2(t)\subseteq \cone^1(f(t))$, hence
  $V(C)\subseteq \cone^1(f(t))$ also holds.  By \ref{p:treeex}, we
  have that $\bag^1(f(t))\setminus \bag^2(t)$ is disjoint from $V(C)$,
  hence not only $\bag^2(t)$, but its superset $\bag^1(t)$ is also
  disjoint from $V(C)$. Therefore, $C$ is also a connected component
  of $\Hcliquex-\bag^1(f(t))$ and since it is in $\cone^1(f(t))$, it
  is a connected component of
  $\Hcliquex[\cone^1(f(t))]-\bag^1(f(t))$. Now
  Lemma~\ref{lem:treedec1} implies that
  $|N_{\Hcliquex}(C)\setminus g^1(f(t))|\le 1$.  By \ref{p:treesep},
  we have $\sep^2(t')= N_{\Hcliquex}(C)$. As
  $\sep^2(t')\subseteq \bag^2(t)$ by definition and
  $g^1(f(t))\setminus g^2(t)$ is disjoint from $\bag^2(t)$ by the way
  we defined $g^2(t)$, we have
\[
|\sep^2(t')\setminus g^2(t)|=|\sep^2(t')\setminus g^1(f(t))|=|N_{\Hcliquex}(C)\setminus g^1(f(t)|\le 1,
\]
 what we had to show for the second property.
\end{proof}

\subsubsection{Proof of Theorem~\ref{thm:spasmboundchar}}

We prove Theorem~\ref{thm:spasmboundchar} by showing three implications. The most substantial one is the (2) $\Rightarrow$ (3) implication, which was already shown by Theorem~\ref{thm:treedec2}. The other two implications, (3) $\Rightarrow$ (1) and (1) $\Rightarrow$ (2), are much more straightforward.

\begin{proof}[Proof (of Theorem~\ref{thm:spasmboundchar})]
(1) $\Rightarrow$ (2). We prove the contrapositive: suppose that (2) is false and let us prove that (1) is false as well. If there is a graph $H\in \cH$ such that $\Hcontr$ has treewidth $c$, then (as $\Hcontr\in \spasm H$), there is obviously a graph in $\spasm H$ with treewidth at least $c$. Suppose that there is a $\binom{c}{2}$-flower centered at color class $C_i$ of some $H\in\cH$. Then by grouping the $c(c-1)$ endpoints of these paths into $c$ blocks of size $c-1$ in appropriate way and consolidating each group, we obtain a graph in $\spasm H$ where there are $c$ vertices in class $C_i$ with internally vertex disjoint paths between any two of them. Such a graph contains a $c$-clique minor and has treewidth at least $c-1$. Thus if (2) is not true for any $c$, then $\spasm H$ contains graphs with arbitrary large treewidth, making (1) also false.

(2) $\Rightarrow$ (3). This is proved by Theorem~\ref{thm:treedec2}.

(3) $\Rightarrow$ (1). Let $H\in\cH$ be a colored graph and let $H'\in \spasm H$. First we claim that if $H$ has a guarded cutvertex decomposition $(T,\beta,g)$ with guard size $c$, then there is such a decomposition $(T,\beta',g')$ for $H'$ as well.  Let $f:V(H)\to V(H')$ be a mapping representing the partition used to obtain $H'$, that is, every vertex $v$ of $H$ with $f(v)=u$ was consolidated into $u$. It is not very difficult to show that setting $\beta'(t)=\{f(v)\mid v\in \beta(t)\}$ and 
$g'(t)=\{f(v)\mid v\in g(t)\}$ yields such a decomposition. The crucial fact to verify is that $X'_u=\{t\in V(T)\mid u\in \beta'(t)\}$ is connected in the tree $T$ for every $u\in V(H')$. Observe that $u$ appears in $\beta'(t)$ if and only if some vertex of $f^{-1}(u)$ appears in $\beta(t)$, thus $X'_u$ is the union of the subtrees $X_{v}=\{t\in V(T)\mid v\in \beta(t)\}$ for $v\in f^{-1}(u)$. As every vertex of $f^{-1}(u)$ has the same color and $(T,\beta)$ is a tree decomposition of $\Hcliquex$, there is a node $t\in V(T)$ that is contained in every subtree $X_v$ corresponding to a vertex of $v\in f^{-1}(u)$, hence their union is also connected. It is easy to verify that $g'(t)$ is a vertex cover of $H'[\beta'(t)]$ for any $t\in V(T)$.
Thus for every $H'\in\spasm H$, we have that $H'$ has a guarded cutvertex decomposition with guard size at most $c$ and then Lemma~\ref{lem:guardedbounded} implies that every graph in $\spasm H$ has treewidth at most $c$.
\end{proof}

\subsection{Half-colorful matchings}
In this section, we look at half-colorful matchings, which are the main reason why certain FPT cases cannot be solved in polynomial time.
First, if we have guraded cutvertex decomposition where $\lambda(t)$ is large for some node $t$, then this can be used to extract a large half-colorful matching.

\begin{lem}\label{lem:lambdatohalf}
Let $(T,\bag,g)$ be a guarded cutvertex decomposition of a colored graph $H$ with guard size at most $c$ and the bags having at most $c$ colors. If $\lambda(t)\ge k$ for some $t\in V(T)$, then $H$ contains a half-colorful matching of size $k/c^2-1$.
\end{lem}
\begin{proof}
Let $x_1$, $\dots$, $x_k$ be vertices of $\lambda(t)$. By definition, $x_i$ has a neighbor $y_i$ that is in $\comp(t_i)$ for some child $t_i$ of $t$ and all the $t_i$'s are distinct. The $x_i$'s are in $\bag(t)$, hence they come from at most $c$ different color classes. Thus there are at least $k/c$ of them from the same color class, say, $C_0$. By Lemma~\ref{lem:guardnomatching}, each color can appear at most $c$ times on the $y_i$'s. This means that we can select $k/c^2-1$ of the $y_i$'s such that they are from distinct color classes different from $C_0$. This gives a half-colorful matching of the required size.
\end{proof}
The following lemma shows that a half-colorful matching can be cleaned in certain way, assuming there are no large flowers in the graph (in which case we know that he problem is hard anyway).

\begin{lem}\label{lem:halftonicehalf}
Let $H$ be a colored graph where there is no $c$-flower centered at any color class for some $c\ge 1$. If there is a half-colorful $k$-matching centered at color class $C_0$, then there is a set $S\subseteq C_0$  with $|S|\le 3c$, vertices $x_1$, $\dots$, $x_{k-3c}$ in $C_0$, and $k-3c$ color classes $C_{i_1}$, $\dots$, $C_{i_{k-3c}}$ such that $x_j$ is the unique neighbor of $C_{i_j}$ in $C_0\setminus S$.  
\end{lem}
\begin{proof}
Let us consider all triples $(x,y,C_i)$ where $x,y\in C_0$ are two distinct vertices and $C_i$ is a color class different from $C_0$ such that $x,y\in N(C_i)$.
These triples can be considered as 3-element subsets of a universe $U$ consisting of the vertices of $C_0$ and elements representing the color classes different from $C_0$. Let us select a maximum collection of pairwise disjoint triples. We claim that if there are $m$ pairwise disjoint triples, then 
there is an $m$-flower centered at $C_0$ in $H$. Indeed, the triple $(x,y,C_i)$ means that there is a path from $x$ to $y$ in $\Hclique{0}$ that has either one internal vertex in $C_i$ (if $x$ and $y$ have common neighbor in $C_i$) or two internal vertices in $C_i$ (if $x$ and $y$ have distinct neighbors in $C_i$). Thus by the assumption of the lemma, the maximum collection contains less than $c$ triples. Let $S$ be the set of all vertices of $C_0$ appearing in the $m$ selected triples and let $\mathcal{C}$ contain all the color classes appearing in these triples; we have $|S|+|\mathcal{C}|\le 3m<3c$.

Suppose that $x_1y_1$, $\dots$, $x_ky_k$ is a half-colorful matching where every $x_i$ is in $C_0$. Let us ignore those edges $x_iy_i$ of the matching where $x_i\in S$ or the color class of $b_i$ is in $\mathcal{C}$. This way, we ignore at most $|S|+|\mathcal{C}|\le 3c$ edges, hence at least $k-3c$ of them remain. Assume without loss of generality that $x_1y_1$, $\dots$, $x_{k-3c}y_{k-3c}$ are among the remaining edges and let $C_{i_j}$ be the color class of $y_{j}$. It is clear that $x_{j}$ is a neighbor of the color class $C_{i_j}$ in $C_0\setminus S$. Moreover, if there it has another neighbor $x'_j\in C_0\setminus S$, then $(x_j,x'_j,C_{i_j})$ is a triple where $x_j,x'_j\not\in S$ and $C_{i_j}\not\in\mathcal{C}$. However, this means that the triple would be disjoint from each of the selected $m$ triples, contradicting the maximality of the selection. Thus $x_1$, $\dots$, $x_{k-3c}$ satisfies the requirements of the lemma.
\end{proof}

Finally, we show how a half-colorful matching with the properties given in Lemma~\ref{lem:halftonicehalf} can be exploited in a reduction from counting perfect matchings.

\begin{lem}\label{lem:halfhard}
Let $H$ be a colored graph and suppose that we are given a set $S\subseteq C_0$  with $|S|\le c$, vertices $x_1$, $\dots$, $x_{k}$ in $C_0$, and $k$ color classes $C_1$, $\dots$, $C_k$ such that $x_j$ is the unique neighbor of $C_j$ in $C_0\setminus S$.  Then given a $k+k$ vertex bipartite graph $G$ and oracle access to $\Sub{H}{\star}$, we can count the number of perfect matchings in $G$ in time $f(c)|V(H)|^{O(1)}$.
\end{lem}
\begin{proof}
We sketch the proof under the simplifying assumption that there is no color class outside $C_0$, $\dots$, $C_k$, there are no edges inside color classes or between $C_i$ and $C_j$ with $i,j>0$, and there are no isolated vertices. It is easy to extend the proof to the general case.

Let $H_0$ be the graph $H$ with all the edges incident to
$C_0\setminus S=\{x_1,\dots,x_k\}$ removed. Suppose that $a_1$,
$\dots$, $a_k$, $b_1$, $\dots$, $b_k$ are the vertices of $G$. We
define the graph $H_G$ by starting with $H_0$, and for every edge
$a_ib_j$ of $G$, we connect $N(a_i)\subseteq C_i$ with $b_j$.

The vertices of $C_0$ in $H$ are of two types: let $C^1_0$ contain those vertices that are adjacent to exactly one color class $C_i$ and let $C^2_0$ contain the remaining vertices, which are adjacent to more than one class. Clearly, we have $C^2_0\subseteq S$ and in particular $|C^2_0|\le c$. For every $v\in C^2_0$, let us choose two edges that connect it with two different color classes; let $F$ be the set of these $2|C^2_0|\le 2c$ edges.

With oracle access to $\Sub{H}{\star}$, we can use standard
inclusion-exclusion techniques to count the number of subgraphs of
$H_G$ that are isomorphic to $H$ and moreover contains every vertex of
$S$ and every edge of $F$. We claim that the number of such subgraphs
is exactly the number of perfect matchings in $G$.

Consider a subgraph  $H^*\subseteq H_G$ isomorphic to $H$ and let $f:V(H)\to V(H_G)$ be a corresponding embedding. Because every edge of $F$ appears in the subgraph, it is clear that $f$ has to map vertices of $C^2_0$ in $H$ to vertices of $C^2_0$ in $H_G$, because no other vertex in $H$ can be adjacent to two different color classes. It follows that it is also true that $f$ maps vertices of $C^1_0$ in $H$ to vertices of $C^1_0$ in $H_G$. Let $C^{1,i}_0$ be the subset of $C^1_0$ containng those vertices of $C^1_0$ that is adjacent only to color class $C_i$. Then exactly $|C^{1,i}_0|-1$ of these vertices appear in $S$. As every vertex of $S$ appears in the subgraph $H^*$, these vertices of $H_G$ are adjacent only to color class $C_i$, and there are no isolated vertices in $H$, it follows that $f$ should map vertices of $C^{1,i}_0\cap S$ of $H$ to vertices of $C^{1,i}_0\cap S$ of $H_G$ and maps exactly one vertex of $C^{1,i}_0$ to $C_0\setminus S$ in $H_G$. Thus the subgraph $H^*$ describes a matching of $H_G$: the edges of $H^*$ connect every $x_i$ to a distinct color class $C_j$. Also, it is not hard to see that $H^*$ actually uses every edge betwen $x_i$ and this color class $C_j$ and $H^*$ uses every edge incident to $S$. This means that the number of edge sets of $H_G$ that form subgraphs isomorphic to $H$ is exactly the number of perfect matchings in $G$.
\end{proof}

\subsection{Algorithm}

\global\long\def\SubTo#1#2{\mathrm{Sub}(#1\to#2)}
\global\long\def\EmbTo#1#2{\mathrm{Emb}(#1\to#2)}
\global\long\def\OrdEmbTo#1#2{\textup{$\Pi$-OrdEmb}(#1\to#2)}
\newcommand{\cS}{\mathcal{S}}
\newcommand{\cF}{\mathcal{F}}
\newcommand{\cW}{\mathcal{W}}

The goal of this section is to prove the algorithmic result in Lemma~\ref{lem:ordembalg}.
Given a colored graph $H$, we say that two vertices $u,v\in V(H)$ are {\em
  similar} if they have the same color and they have the same open neighborhood, that is,
$N_H(u)=N_H(v)$. Note that if $u$ and $v$ are similar, then this
means in particular that they are not adjacent. Clearly, similarity is an equivalence relation. We say that a partition $\Pi$ of $V(G)$ {\em respects similarity} if each class of $\Pi$ consists of vertices similar to each other (but it is possible that vertices in different classes are also similar).

We say that a colored graph $G$ is {\em ordered} if it is equipped
with a total order $<$ on its vertices, for example, the vertex set is
$[n]$ for some integer $n\ge 1$.   Let $H$ and $G$
be two ordered graphs and let $\Pi$ be a partition of $V(H)$. We say that a subgraph embedding
$\phi:V(H)\to V(G)$ is {\em $\Pi$-ordered} if whenever $u,v\in V(H)$ are
in the same class and $u<v$ holds, then we have $\phi(u)<\phi(v)$.  We denote by
$\OrdEmbTo{H}{G}$ the number of $\Pi$-ordered embeddings from $H$ to
$G$.

Observe that if $u$ and $v$ are similar vertices of $H$ and $\phi$ is a subgraph embedding from $H$ to $G$, then the values of $\phi(u)$ and $\phi(v)$ can be exchanged and the resulting mapping is still a valid subgraph embedding from $H$ to $G$. In fact, every subgraph embedding from $H$ to $G$ can be obtained from a $\Pi$-ordered subgraph embedding by permuting the values of $\phi$ inside each class of $\Pi$. Thus, as the following lemma states, the number of embeddings can be recovered easily from the number of $\Pi$-ordered subgraph embeddings.

\begin{lem}\label{lem:ordemb}
Let $H$ and $G$ be two ordered graphs and let $\Pi=(P_1,\dots,P_p)$ be a partition of $V(H)$ respecting similarity. Then we have
\[
\EmbTo{H}{G}=\OrdEmbTo{H}{G}\cdot \prod_{i=1}^p (|P_i|!).
\]
\end{lem}

For an ordered graph $G$ and subset $S\subseteq V(G)$, we say that $S$
is {\em $\Pi$-prefix} if whenever $u,v\in V(H)$ are in the same class
of $\Pi$ with $u<v$ and $v\in S$ holds, then $u\in S$ holds as well. If
$\phi:V(H)\to V(G)$ is a $\Pi$-ordered subgraph embedding and $i$ is some
vertex of $G$, then the set $S_{\le i}=\{v\in V(H) \mid \phi(v)\le i\}$ of vertices that are mapped to vertices of $G$ not greater than $i$ is a $\Pi$-prefix set.

\begin{proof}[Proof (of Lemma~\ref{lem:ordembalg})]

It will be convenient to assume that that the tree decomposition is tight: every vertex of $\sep(t)$ has a neighbor in $\comp(t)$. As discussed in Section~\ref{sec:mass-tree-decomp}, this can be achieved by easy transformations and these transformations do not increase the size of the sets $\lambda(t)$. Then $t'\in h_t(v)$ is equivalent to $\sep(t')=g(t)\cup \{v\}$ (if the decomposition is not tight, then $\sep(t')=g(t)\cup \{v\}$ is possible even if $v$ is not adjacent to any vertex of $\comp(t')$ and hence $t'\not\in h_t(v)$).

Let us define a partition $\Pi$ of $V(H)$ the following way. For every
  $t\in V(T)$, we partition $\bag(t)\setminus (g(t)\cup \lambda(t))$ according to
  similarity and let each such class be a class of $\Pi$ (by
  Lemma~\ref{cl:nonkappaneighbor}(1), these classes are disjoint). If
  a vertex does not appear in any of these classes, then let it appear in
  a singleton class of $\Pi$. Observe that $\Pi$ respects similarity.
  We present an algorithm for computing $\OrdEmbTo{H}{G}$. Then
  Lemma~\ref{lem:ordemb} can be used to compute $\EmbTo{H}{G}$.

  The algorithm uses two layers of dynamic programming: the {\em
    outer} dynamic programming procedure uses the standard method of
  considering the nodes of the tree decomposition in a bottom-up
  order, and additionally there is an {\em inner} dynamic programming
  procedure at each node $t$, which restricts, for increasing values
  of $i$, that the vertices of $\bag(t)\setminus g(t)$ can be mapped
  only to the first $i$ vertices.

\textbf{The outer dynamic programming.} For a node $t\in V(T)$, let $H_t=H[\cone(t)]$.
For a set $S\subseteq V(H)$, let $\cF_S$ be the set of all injective mappings from $S$ to $V(G)$. The goal of the outer dynamic programming procedure is to compute, for every node $t\in V(T)$ and function $f\in \cF_{\sep(t)}$, the size of the set
\[
\cW_{t,f}=\{\phi\in \OrdEmbTo{H_t}{G} : \text{$\phi$ restricted to $\sep(t)$ is $f$}\}.
\]
If $r$ is the root of $H$, then $H_r=H$ and $\sep(r)=\emptyset$ holds,
hence $\OrdEmbTo{H_t}{G}=\cW_{t,f}$, where $f$ is the unique empty
function $\emptyset\to V(H)$. There is a slight abuse of notation
here: in $\OrdEmbTo{H_t}{G}$, we would need to use the restriction of
$\Pi$ to $V_t$ instead of $\Pi$. However, observe that every class of
$\Pi$ is either contained in $V_t$ or disjoint from $V_t$, so this
does not create any ambiguity.

We compute the values in a bottom up way: when computing $\#\cW_{t,f}$, we
assume that $\#\cW_{t',f'}$ is already available for every child $t'$ of
$t$ and every $f\in \cF_{\sep(t')}$. To compute $\#\cW_{t,f}$, we solve a slightly more restricted problem: we fix the value of the embedding not only on $\sep(t)$, but on $g(t)\supseteq \sep(t)$. For every $t\in V(T)$ and $\overline f\in \cF_{g(t)}$, we define 
\[
\overline \cW_{t,\overline f}=\{\phi\in \OrdEmbTo{H_t}{G} : \text{$\phi$ restricted to $g(t)$ is $\overline f$}\}.
\]
Clearly, we have $\#\cW_{t,f}=\sum_{\overline f\in \cF_{g(t)}, \overline f_{|\sep(t)}=f} \#\overline \cW_{t,\overline f}$, thus computing the values $\#\overline \cW_{t,\overline f}$ is sufficient to compute $\#\cW_{t,f}$.

\textbf{The inner dynamic programming.}
The goal of the inner dynamic programming is to compute
$\#\overline \cW_{t,\overline f}$ for a given $t\in V(T)$ and
$\overline f\in\cF_{g(t)}$, assuming that the values $\#\cW_{t',f'}$ are
available for every child $t'$ of $t$ and $f'\in\cF_{\sep(t')}$.  To
describe the subproblems of the inner dynamic programming, we need some further definitions.

For any subset $g(t)\subseteq S\subseteq \bag(t)$, we define
\[
V_{t,S}:=S\cup \bigcup_{v\in S\setminus g(t)}\bigcup_{t'\in h_t(v)}\comp(t')=
S\cup \bigcup_{\textup{$t'$ is a child of $t$, $\sep(t')\subseteq S$}} \comp(t')
\]
and $H_{t,S}:=H[V_{t,S}]$. That is, $V_{t,S}$ contains the subset $S$
of $\bag(t)$, and those branches of the tree decomposition that are
rooted at some child $t'$ of $t$ that is hanging at some vertex in
$S$, or more formally, $\sep(t')\subseteq S$. Here we use the
assumption $|\sep(t')\setminus g(t)|\le 1$, which implies that
$\sep(t')\subseteq S$ is equivalent to $t'\in h_t(v)$ for some
$v\in S$.

For every $t\in V(T)$, we define $\cS_t$ to contain every subset $S$
with $g(t)\subseteq S \subseteq \bag(t)$ that is a $\Pi$-prefix subset
of $V(H)$. From the way we defined $\Pi$, it is clear that $\bag(t)$
itself is a $\Pi$-prefix subset of $V(H)$, hence $\bag(t)\in
\cS_t$. It is not difficult to bound the size of $\cS_t$.

\begin{claim}\label{cl:twinorderedsubset}
For every $t\in V(T)$, we have $|\cS_t|=2^{d}\cdot |V(G)|^{c\cdot 2^c}$ and the collection $\cS_t$ can be constructed in time $2^{O(d)}|V(G)|^{2^{O(c)}}$. 
\end{claim}
\begin{proof}
The partition $\Pi$ classifies the vertices of 
  $\bag(t)\setminus (g(t)\cup \lambda(t))$ according to the similarity relation. 
By Lemma~\ref{cl:nonkappaneighbor}(3), these vertices have at most $2^{|g(t)|}\le 2^c$ possible neighborhoods and they have $c$ possible colors, thus $\Pi$ partitions 
  $\bag(t)\setminus (g(t)\cup \lambda(t))$ into at most $c\cdot 2^c$ classes $P_1$, $\dots$, $P_s$.
If a $\Pi$-prefix set contains exactly $j$ vertices of one such class $P_i$, then we know that it contains
  exactly the first $j$ vertices of $P_i$. Thus each set in $\cS_t$
  can be completely specified by describing its intersection with
  $\lambda(t)$ and specifying the size of its intersection with each
  $P_i$. As $|\lambda(t)|\le d$ by assumption, this gives $2^{d}\cdot n^{c\cdot 2^c}$ different
  possibilities.
\end{proof}

We are now ready to define the subproblems of the inner dynamic programming.  
For every $t\in V(T)$, $\overline f\in \cF_{g(t)}$, $S\in \cS_t$, and $0\le i \le |V(G)|$, 
we define
\[
\#\overline \cW_{t,\overline f,S,i}=\{\phi\in \OrdEmbTo{H_{t,S}}{G} : \text{$\phi$ restricted to $g(t)$ is $\overline f$ and $\phi(v)\le i$ for every $v\in \bag(t)\setminus g(t)$}\}.
\]
That is, we need to map only $H_{t,S}$, but we have a restriction on where the vertices of $\bag(t)\setminus g(t)$ can be mapped. As we have observed that $\bag(t)\in \cS_t$ holds, the value
$\#\overline \cW_{t,\overline f,\bag(t),|V(G)|}$ is defined, and it is equal to
$\#\overline \cW_{t,\overline f}$.

\textbf{Solving the subproblems.}  The inner dynamic programming
proceeds by solving the subproblems by increasing value of $i$: when
computing $\#\overline \cW_{t,\overline f,S,i}$, we assume that that values
$\#\overline \cW_{t,\overline f,S',i'}$ are available for every $i'<i$ and
$S'\in\cS_t$.

For $i=0$, determining $\#\overline \cW_{t,\overline f,S,i}$ is trivial:
its value is 1 if $S=g(t)$ and it is 0 if $S\supset g(t)$ (as then an
embedding would need to map a vertex of $S\setminus g(t)$ to the first
$i=0$ vertices). For $i\ge 1$, consider vertex $i$ of $G$.
If an embedding of  $\overline \cW_{t,\overline f,S,i}$ maps no vertex to $i$, or maps a vertex not in $\bag(t)\setminus g(t)$ to $i$, then this embedding already appears in $\overline \cW_{t,\overline f,S,i-1}$. In particular, if $i\in \overline f(g(t))$, then every mapping of $\overline \cW_{t,\overline f,S,i}$ maps a vertex of $g(t)$ to $i$, hence we have 
$\#\overline \cW_{t,\overline f,S,i}=\#\overline \cW_{t,\overline f,S,i-1}$. Thus in the following, we can assume that $i\not\in \overline f(g(t))$.
Let us fix a vertex $v\in \bag(t)\setminus g(t)$ and count the number of embeddings $\phi\in \overline \cW_{t,\overline f,S,i}$ that map $v$ to $i$; summing these values 
for every $v\in \bag(t)\setminus g(t)$ and adding the number $\#\overline \cW_{t,\overline f,S,i-1}$ of embeddings that do not map any vertex of 
$\bag(t)\setminus g(t)$ to $i$ gives exactly the required value $\#\overline \cW_{t,\overline f,S,i}$. 
We consider two cases depending on whether $v$ is in $\lambda(t)$ or not.

\textit{Case 1.} $v\not\in \lambda(t)$. There are three obvious conditions that are necessary for the existence of embeddings in $\overline \cW_{t,\overline f,S,i}$ that map $v$ to $i$.
\begin{itemize}
\item Vertex $i$ needs to have the same color as $v$.
\item For every neighbor $u$ of $v$ in $H$ (note that $u$ has to be in $g(t)$ by Lemma~\ref{cl:nonkappaneighbor}(3)), vertex $\overline f(u)$ should be a neighbor of $i$. 
\item The set $S':=S\setminus \{v\}$ should be $\Pi$-prefix: as $S$ is
  $\Pi$-prefix, the only way for $S'$ to lose this property is if there is
  a $u\in S'$ with $v<u$ that is in the same class of $\Pi$ as $v$, which in particular means that $u\in \bag(t)\setminus (g(t)\cup \lambda(t))$. But then the $\Pi$-ordered embedding $\phi\in
\overline \cW_{t,\overline f,S,i}$ should map $u\in \bag(t)\setminus g(t)$ to a vertex greater than $i$, which is not possible by the definition of $\overline \cW_{t,\overline f,S,i}$.
\end{itemize}

We claim that if $v$ satisfies these three conditions, then any
embedding of $\overline \cW_{t,\overline f,S',i-1}$ can be extended to an
embedding of $\overline \cW_{t,\overline f,S,i}$ by mapping $v$ to $i$. The
only subtle point here is that we need to argue that the embeddings in
$\overline \cW_{t,\overline f,S',i-1}$ cannot already use $i$: this is because
every vertex of $H$ with same color as $i$ and $v$ appears in
$\bag(t)$ (Lemma~\ref{cl:nonkappaneighbor}(2)), we assumed that $\overline f$ maps no vertex of
$g(t)$ to $i$, and an embedding in $\overline \cW_{t,\overline f,S',i-1}$ cannot
map a vertex of $\bag(t)\setminus g(t)$ to $i$ by definition. Thus we can conclude
that if the three conditions hold, the number of embeddings in
$\overline \cW_{t,\overline f,S,i}$ that map $v$ to $i$ is exactly
$\#\overline \cW_{t,\overline f,S',i-1}$.

\textit{Case 2.} $v\in \lambda(t)$. Again, we have two obvious conditions: 
\begin{itemize}
\item Vertex $i$ needs to have the same color as $v$.
\item For very neighbor $u$ of $v$ in $H$ with $u\in g(t)$ (note that
  now $v$ can have neighbors not in $g(t)$), vertex $\overline f(u)$
  should be a neighbor of $i$.
\end{itemize}
The set $S':=S\setminus \{v\}$ is always $\Pi$-prefix:
$v\in \lambda(t)$ implies that $v$ is in a singleton class of
$\Pi$. Again, we can argue that no embedding in
$\overline \cW_{t,\overline f,S',i-1}$ can use vertex $i$ of $G$. This
means that any embedding $\psi:V_{t,S'}\to V(G)$ appearing in
$\overline \cW_{t,\overline f,S,i}$ can be extended to an embedding
$\psi^{+}: V_{t,S'}\cup \{v\}\to V(G)$ by setting
$\psi^{+}(v)=i$. However, $V_{t,S}$ contains more than just
$V_{t,S'}\cup \{v\}$: it contains $\comp(t')$ for every $t'\in
h_t(v)$.
For any $t'\in h_t(v)$, let $f_{t'}$ be the restriction of $\psi^{+}$
to $\sep(t')$. As $\sep(t')\subseteq g(t)\cup \{v\}$ (by definition of
$t'\in h_t(v)$), the function $f_{t'}$ depends only on $\overline f$
and on $i$. We have already computed the value $\#\cW_{t',f_{t'}}$, which
is the number of extensions of $f_{t'}$ to a embedding of $H_t$ to
$G$. The crucial observation is that an embedding
$\psi_{t'}\in \cW_{t',f_{t'}}$ cannot conflict with $\psi^{+}$ in the
sense that it is not possible that there is a vertex
$w\in V_{t,S'}\cup \{v\}$ and a vertex $w'\in \comp(t')$ with
$\psi^{+}(w)=\psi_{t'}(w')$. This would be possible only if $w$ and
$w'$ have the same color, that is, they are adjacent in
$\Hcliquex$. But then $w'\in \comp(t')$ would imply $w\in \cone(t')$,
which is only possible if $w\in \sep(t')$. Mappings $\psi^{+}$ and
$\psi_{t'}$ agree on $\sep(t')$ (as $f_{t'}$ is the restriction of $\psi^{+}$ to $\sep(t')$, while $\psi_{t'}$ extends $f_{t'}$), hence $w\in\sep(t')$ cannot conflict
with $w'$. Thus any $\psi_{t'}\in \cW_{t',f_{t'}}$ can be used to extend
$\psi^{+}$ to an embedding from $H_{t,S'}\cup \{v\}\cup \comp(t')$.
Moreover, by the same reasoning, if $t',t''\in h_t(v)$ are two
distinct children of $t$, then two embeddings
$\psi_{t'}\in \cW_{t',f_{t'}}$ and $\psi_{t''}\in \cW_{t'',f_{t''}}$
cannot conflict either: a vertex of $\comp(t')$ cannot have the same
color as a vertex of $\comp(t'')$. Therefore, if we pick any
combination of embeddings $\psi_{t'}\in \cW_{t',f_{t'}}$ for each
$t'\in h_t(v)$, then together they can be used to extend $\psi^{+}$ to
an embedding of $H_{t,S}$. We can conclude that if the two conditions
hold, then the number of embeddings in
$\overline \cW_{t,\overline f,S,i}$ that map $v$ to $i$ is exactly
$\#\overline \cW_{t,\overline f,S,i-1}$ times the product of the value
$\#\cW_{t',f_{t'}}$ for every $t'\in h_t(v)$.

We have shown that $\#\overline \cW_{t,\overline f,S,i}$ can be computed
in polynomial time, assuming we have already computed
$\#\overline \cW_{t,\overline f,S',i-1}$ for every $S'\in \cS_t$, and
$\#\cW_{t',f'}$ for every child $t'$ of $t$ and mapping
$f'\in \cF_{\sep(t')}$. For a given $t$ and $\overline f$, there are
$2^{d}\cdot |V(G)|^{2^{O(c)}}$ values $\#\overline \cW_{t,\overline f}$ compute (Claim~\ref{cl:twinorderedsubset}) and solving these subproblems allows us to compute $\#\overline \cW_{t,\overline f}$, hence we can conclude that $\#\overline \cW_{t,\overline f}$ can be computed in time $2^{d}\cdot |V(G)|^{2^{O(c)}}$. For a given $t\in V(T)$, the number of possibilities for $\overline f\in \cF_{g(t)}$ is at most $|V(G)|^{g(t)}\le |V(G)|^c$, which means that the at most $|V(H)|\cdot |V(G)|^c$ subproblems 
$\#\overline \cW_{t,\overline f}$ can be all solved in total time $2^{d}\cdot |V(G)|^{2^{O(c)}}$.    
\end{proof}

\subsection{Putting it together}
Finally, we are ready to prove the main result of the section, classifying which of the FPT cases of counting colored patterns is polynomial-time solvable. The result is under assuming the Nonuniform Counting Exponential Time Hypothesis, which states that there is an $\epsilon>0$ such that there is no infinite collection of algorithms $\{A_n\mid n\in N\}$ for some infinite set $N\subseteq \mathbb{Z}^+$ such that algorihm $A_n$ solves $n$-variable \#3-SAT in time $2^{\epsilon n}$. By known reductions, we can replace \#3-SAT with counting perfect matchings in an $n+n$ vertex bipartite graph \cite{DBLP:conf/icalp/Curticapean15}.

\begin{thm}\label{thm:coloredpoly}
Let $\cH$ be a class of colored graph where, for every $H\in \cH$, the graphs in $\spasm H$ have treewidth at most $c$.
\begin{enumerate}
\item If there is a constant $h$ such that the largest half-colorful matching in every $H\in \cH$ is at most $h\log|V(H)|$, then $\SubProb{\cH}$ is polynomial-time solvable.
\item Otherwise, $\SubProb{\cH}$ is not polynomial-time solvable, unless the Nonuniform Counting Exponential Time Hypothesis.
\end{enumerate}
\end{thm}
\begin{proof}
  Suppose that the first statement holds. Consider an instance of
  $\SubProb(\cH)$ with inputs $H$ and $G$. By
  Theorem~\ref{thm:spasmboundchar}, there is a constant $c'$ such that
  $\Hcontr$ has treewidth at most $c'$ and there is no $c'$-flower
  centered at any color class of $H$. Thus by
  Theorem~\ref{thm:treedec2}, there is a constant $c''$ such that we
  can obtain in polynomial time a guarded cutvertex decomposition with
  guard size $c''$ and at most $c''$ colors in each bag. If $\lambda(t)=k$ for some node $t$ of the decomposition, then Lemma~\ref{lem:lambdatohalf} implies that there is a half-colorful matching of size $k/(c'')^2-1$ in $H$. Thus the assumption that there is no such matching larger than $h\log |V(H)|$ implies that $\lambda(t)=O(h(c'')^2\log |V(H)|)$. It follows that the running time of the algorithm of Lemma~\ref{lem:ordembalg} is $|V(H)|^{h(c'')^2}|V(H)|^{2^{O(c'')}}$, which is polynomial time for fixed constants $h$ and $c''$.

  For the second statement, suppose that there is no such $h$. From the assumption that $\spasm H$ has treewidth at most $c$ and Theorem~\ref{thm:spasmboundchar}, there is a constant $c'$ such that there is no $c'$-flower in $H$. Let us fix any $\epsilon>0$. By assumption, there are infinitely many graphs  $H\in\cH$ where the size of the largest half-colorful matching is at least $(1/\epsilon) \log |V(H)|$.
  For every $n\ge 1$, if there is such a graph where the size of the largest half-colorful matching is exactly $n+3c'$, then let us fix such a graph $H_n$. Note that $n\ge (1/\epsilon)\log |V(H)|$ means $|V(H)|\le 2^{\epsilon n}$. Lemma~\ref{lem:halftonicehalf} shows that there is a half-colorful matching of size $n$ satisfying certain conditions and then Lemma~\ref{lem:halfhard} can be used to reduce counting perfect matchings in an $n+n$ vertex bipartite graph to $\SubProb{\cH}$ and use the assumed polynomial-time algorithm. This way, $A_n$ solves the problem of counting perfect matchings in time $(2^{\epsilon n})^{O(1)}$. As we can construct such an infinite sequence of algorithms for any $\epsilon>0$, our complexity assumption fails.
\end{proof}

\section{Open Problems}

We have defined the space of graph motif parameters and explored three useful bases thereof, namely, $\hom$, $\sub$, and $\indsub$.
These bases capture well-studied classes of counting problems, and we could use basis changes to transfer results between these classes.
Are there other computationally interesting bases?
Moreover, are there other interesting subspaces of the space of all graph parameters other than the graph motif parameters?

\section*{Acknowledgments}

Thanks a lot to Édouard Bonnet for pointing out \cite{Scott2007260} and \cite{dvorak2016strongly}.

\bibliographystyle{plainurl}
\bibliography{references}

\appendix
\section{Appendix}

In this section, we provide those proofs that we chose not to include in the main text.

\begin{proof}[Proof of Proposition~\ref{prop: hom-treewidth-algo}]
  We begin by computing an optimal tree decomposition $(T,\bag)$ of~$H$, for
  example via the $O(1.7549^{\abs{V(F)}})$ time exact algorithm by Fomin and
  Villanger~\cite{DBLP:journals/combinatorica/FominV12}.
  Without loss of generality, we assume it to be a \emph{nice} tree 
  decomposition, in which each node has at most two children; the leaves satisfy 
  $\bag(v)=\emptyset$, the nodes with two children $w,w'$ satisfy
  $\bag(v)=\bag(w)=\bag(w')$ and 
  called \emph{join} nodes, the nodes with one child~$w$ satisfy 
  $\bag(v)=\bag(w)\cup\{x\}$ and called \emph{introduce} nodes, or
  $\bag(v)\cup\{x\}=\bag(w)$ and 
  called \emph{forget} nodes.

  Recall that $\comp(v)$ be the union of all bags at 
  or below vertex~$v$ in the tree~$T$.
  The goal of the dynamic programming algorithm is to build up the following 
  information~$I_v$ at each vertex~$v$ of~$T$:
  For each $h\in\Hom{H[\bag(v)]}{G}$, we let $I_v(h)$ be the number of homomorphisms 
  $\bar h\in\Hom{H[\comp(v)]}{G}$ such that $\bar h$ extends $h$.
  We observe some properties of $I_v$.

  Let $v$ be a leaf of~$T$.
  Then $\bag(v)=\emptyset$ and we define $I_v(h)=1$ for the empty function~$h$.
  For each leaf~$v$, we can set up the dynamic programming table entry $I_v$ in 
  constant time.

  Let $v$ be an \emph{introduce} vertex of~$T$.
  Let $w$ be its unique child in the tree.
  Suppose the vertex~$x\in V(G)$ is introduced at this node, that is, 
  $\bag(w)\cup\{x\}=\bag(v)$.
  Let $h\in\Hom{H[\bag(v)]}{G}$, and let $h'$ be $h$ with~$x$ removed from its 
  domain.
  Then $I_v(h)=I_w(h')$ holds because $N_H(x) \cap \comp(v) = N_H(x)\cap 
  \bag(v)$.
  Thus at introduce vertices, we can compute~$I_v$ from the table $I_w$ by going 
  over all table entries $h$ (there are at most $|V(G)|^{|\bag(v)|} \le 
  |V(G)|^{\tw{H}+1}$ of them) and writing the value $I_w(h')$ to $I_v(h)$.

  Let $v$ be a \emph{forget} vertex of~$T$.
  Let $w$ be its unique child in the tree.
  Suppose the vertex~$x\in V(G)$ is forgotten at this node, that is, 
  $\bag(w)\setminus\{x\}=\bag(v)$.
  Then the neighborhood of $x$ is contained in $\comp(w)$.
  Let $h\in\Hom{H[\bag(v)]}{G}$.
  For each $z\in V(G)$, let $h_{z}:\bag(w)\to V(G)$ be the unique function that is 
  consistent with $h$ and satisfies $h(x)=z$.
  Then we have
  \[
    I_v(h) = \sum_{z\in V(G)} \Big[h_z\in\Hom{H[\bag(w)]}{G}\Big]\cdot I_w(h_z)\,.
  \]
  For each of the at most $|V(G)|^{|\bag(v)|} = |V(G)|^{|\bag(w)|-1} \le 
  |V(G)|^{\tw{H}}$ functions~$h$, we compute $I_v(h)$ as follows:
  first we determine the set~$Z$ of all~$z$ for which $h_z$ is a homomorphisms, 
  and then we sum up the corresponding entries $I_w(h_z)$.
  Let $N(x)$ be the open neighborhood of $x$ in $H$, and let $N'$ be the image 
  of $N(x)\cap \bag(v)$ under~$h$.
  Then the set $Z$ is the set of all vertices~$z$ that are adjacent to all 
  vertices in~$N'$.
  Clearly $N'$ has at most $|\bag(v)|\le\tw{H}$ elements.
  Thus $Z=\bigcap_{y\in N'} N_G(y)$ holds, and we can compute this intersection 
  in time $\tilde O(\tw{H} \cdot |V(G)|)$.
  The running time for forget vertices is therefore $\tilde O(\tw{H} \cdot 
  |V(G)|^{\tw{H}+1})$

  Let $v$ be a \emph{join} vertex of~$T$, that is, it has exactly two 
  children~$w$ and~$w'$ with $\bag(v)=\bag(w)=\bag(w')$.
  Let $h\in\Hom{H[\bag(v)]}G$.
  Then $I_v(h) = I_w(h) \cdot I_{w'}(h)$.
  This operation can be computed in time $|V(G)|^{|\bag(v)|} \le |V(G)|^{\tw{H}+1}$.
\end{proof}

To prove Proposition~\ref{prop: hom ETH hard}, we state a
vertex-colorful version of the subgraph counting problem.
For two graphs~$F$ and $G$, we say that $G$ is $F$-colored graph if there is a
homomorphism~$f\in\Hom{G}{F}$ from $G$ to $F$.
Let $\PartitionedSub F G$ be the set of all subgraphs~$H$ of~$G$ that are
isomorphic to $F$ and \emph{vertex-colorful}, that is, these subgraphs~$H$
satisfy~$\abs{f^{-1}(v)\cap V(H)} = 1$ for all $v\in V(F)$.
In the problem $\PartitionedSubProb{\mathcal F}$ for a fixed graph class
$\mathcal F$, we are given a graph~$F\in\mathcal F$ and an $F$-colored
graph~$G$, and we are asked to compute the number~$\#\PartitionedSub F G$.
A full dichotomy for this class of problems is known, and it establishes
treewidth as the tractability criterion, together with near-tight lower bounds
under $\sharpETH$.

\begin{thm}[Corollary~6.2 and~6.3 of \cite{marx2007can}]%
  \label{thm: partitionedsub ETH hardness}
  Let $\mathcal F$ be a recursively enumerable family of graphs such that the
  treewidth of graphs in~$\mathcal F$ is unbounded.
  If $\sharpETH$ is true, the problem $\PartitionedSubProb{\mathcal F}$ cannot
  be solved in time $f(H)\cdot n^{o(\tw H / \log \tw H)}$ for patterns $H \in
  \mathcal F$ of treewidth $\tw H$.
\end{thm}

We are ready to prove the following proposition.

\begin{repproposition}{prop: hom ETH hard}
  Let $\mathcal F$ be a recursively enumerable class of graphs of unbounded
  treewidth.
  If $\sharpETH$ holds, there is no $f(H)\cdot \abs{V(G)}^{o(\tw H / \log \tw
  H)}$ time algorithm to compute $\HomProb {\mathcal F}$ for given
  graphs~$H\in\mathcal F$ and~$G$.
\end{repproposition}
\begin{proof}
  We reduce from $\PartitionedSubProb{\mathcal F}$ to $\HomProb{\mathcal F}$.
  Given an instance $(H,G,f)$ with graphs $H\in\mathcal F$ and $G$ and a given
  homomorphism~$f\in\Hom{G}{F}$, the reduction wants to compute $\PartitionedSub{H}{G}$.
  The reduction uses a straightforward inclusion--exclusion argument.
  For a set $A\subseteq V(H)$, let $G_{-A}$ be the graph obtained from~$G$ by
  deleting all vertices~$v\in V(G)$ with $f(v)\in A$.
  Then the set 
  $\Hom{H}{G}\setminus\bigcup_{\emptyset\ne A\subseteq V(H)}\Hom{H}{G_{-A}}$
  contains all homomorphisms~$h\in\Hom{H}{G}$ that are color-preserving, that
  is, they map each vertex~$v\in V(H)$ to a vertex~$h(v)\in V(G)$ with the
  property that $f(h(v))=v$.
  In particular, these homomorphisms are injective.
  By the principle of inclusion and exclusion, their number is equal to
  \[
    \sum_{\emptyset\ne A\subseteq V(H)}
    (-1)^{\abs{A}}
    \cdot
    \#\Hom{H}{G_{-A}}
    \,.
  \]
  Dividing by~$\#\Aut{H}$ yields the number of vertex-colorful~$H$-copies in~$G$
  as required.
\end{proof}

\end{document}